\documentclass[a4paper,12pt]{amsart}
\usepackage[left=1in, right=1in, top = 1.2in, bottom = 1.2in]{geometry}
\usepackage[utf8]{inputenc}
\usepackage[all]{xy}
\usepackage{amssymb}
\usepackage{amsmath}
\usepackage{csquotes}
\usepackage{subfig}
\usepackage{enumerate}
\usepackage{graphicx}
\usepackage{enumerate}
\usepackage{amsaddr}
\usepackage{color}
\usepackage{mwe}
\usepackage{appendix}

\usepackage{soul}
\usepackage{xcolor}
\newcommand\independent{\protect\mathpalette{\protect\independenT}{\perp}}
\def\independenT#1#2{\mathrel{\rlap{$#1#2$}\mkern2mu{#1#2}}}

\usepackage{tikz}
\usetikzlibrary{arrows.meta,shapes}

\usepackage{setspace}
\newtheorem{lemma}{Lemma}

\newtheorem{theorem}{Theorem}

\title[]{Generalized interpretation and identification of separable effects in competing event settings}

\author{Mats J. Stensrud$^{1,2}$, Miguel A. Hern\'{a}n$^{1,3}$, Eric J. Tchetgen Tchetgen$^{4}$, James M. Robins$^{1,3}$, Vanessa Didelez$^{5,6}$, Jessica G. Young$^7$} \address{ $^1$ Department of Epidemiology, Harvard T. H. Chan School of Public Health, USA \\
$^2$Department of Biostatistics, University of Oslo, Norway\\
$^3$ Department of Biostatistics, Harvard T. H. Chan School of Public Health, USA  \\
$^4$ Department of Statistics, The Wharton School, University of Pennsylvania, USA \\
$^5$ Leibniz Institute for Prevention Research and Epidemiology – BIPS, Germany \\
$^6$ Faculty of Mathematics / Computer Science,
University of Bremen, Germany \\
$^7$ Department of Population Medicine,
Harvard Medical School and Harvard
Pilgrim Health Care Institute, USA \\
}
\date{\today}
\MakeOuterQuote{"}\EnableQuotes
\DeclareUnicodeCharacter{00A0}{ }

\begin{document}

\begin{abstract}
In competing event settings, a counterfactual contrast of cause-specific cumulative incidences quantifies the total causal effect of a treatment on the event of interest. However, effects of treatment on the competing event may indirectly contribute to this total effect, complicating its interpretation. We previously proposed the \textsl{separable effects} (Stensrud et al, 2019) to define direct and indirect effects of the treatment on the event of interest. This definition presupposes a treatment decomposition into two components acting along two separate causal pathways, one exclusively outside of the competing event and the other exclusively through it. Unlike previous definitions of direct and indirect effects, the separable effects can be subject to empirical scrutiny in a study where separate interventions on the treatment components are available. Here we extend and generalize the notion of the separable effects in several ways, allowing for interpretation, identification and estimation under considerably weaker assumptions. We propose and discuss a definition of separable effects that is applicable to general time-varying structures, where the separable effects can still be meaningfully interpreted, even when they cannot be regarded as direct and indirect effects. We further derive weaker conditions for identification of separable effects in observational studies where decomposed treatments are not yet available; in particular, these conditions allow for time-varying common causes of the event of interest, the competing events and loss to follow-up. For these general settings, we propose semi-parametric weighted estimators that are straightforward to implement. As an illustration, we apply the estimators to study the separable effects of intensive blood pressure therapy on acute kidney injury, using data from a randomized clinical trial.

\end{abstract}

\maketitle


\section{Introduction}
\label{sec: introduction}
Researchers are often interested in treatment effects on an event of interest that is subject to competing events, that is, events that make it impossible for the event of interest to subsequently occur.  For example, when the event of interest is kidney injury, death is a competing event because any individual who dies prior to kidney injury cannot subsequently suffer from kidney injury. A counterfactual contrast in cause-specific cumulative incidences (risks) quantifies the \textsl{total effect} of the treatment on the event of interest through all causal pathways. When the treatment affects competing events, the total effect also partly includes pathways mediated by these competing events \cite{robins1986new, young2018choice}. For example, a harmful total effect of blood pressure therapy on the risk of kidney injury may be due to a biological side-effect on the kidneys, but could also be fully or partly explained by a protective treatment effect on cardiovascular death. 

As previously discussed \cite{robins1986new, young2018choice}, other popular estimands in competing events settings do not address this interpretational problem. Counterfactual contrasts in cause-specific or subdistribution hazard ratios do not generally have a causal interpretation as they contrast counterfactual outcomes in different populations (those who survive under different treatments). Other estimands that do have a causal interpretation may be of limited practical relevance.  Specifically, a counterfactual contrast in marginal cumulative incidences (or net risks \cite{geskus}) is a special case of a controlled direct effect \cite{robins1992identifiability}, which is defined relative to outcomes under an intervention to eliminate competing events. Alternatively, the survivor average causal effect defines the total effect in the subset of the population that would survive the competing event under any level of treatment.  This subset, which may not even exist, is never identified as it requires knowledge of cross-world treatment conditions that are unobservable. Finally, pure (natural) direct and indirect effects \cite{robins1992identifiability, pearl2009causality} are defined relative to outcomes under interventions that assign individuals competing event status under unobservable cross-world treatment conditions. 


To address the shortcomings of these existing estimands, we recently proposed the \emph{separable effects} for causal inference in competing event settings \cite{stensrud2019separable}, inspired by Robins and Richardson's extended graphical approach to mediation analysis \cite{robins2010alternative, didelez2018defining}. Given a plausible decomposition of the treatment into different components, we defined these effects as counterfactual contrasts indexed by hypothetical interventions that assign these components different values. The separable effects have clear advantages over the existing causal estimands. In particular, the separable effects do not require hypothetical interventions that eliminate competing events and avoid cross-world counterfactuals, which can never be subject to empirical scrutiny \cite{robins2010alternative}. Instead, the separable effects can, at least in principle, be directly identified in a future experiment where the treatment components are assigned different values \cite{stensrud2019separable}.

In Stensrud et al. \cite{stensrud2019separable}, we defined separable effects under the special case where neither treatment component affects common causes of the event of interest and the competing event. In these settings, the separable effects can be interpreted as direct and indirect effects, quantifying treatment effects acting exclusively outside of, and exclusively through, the competing event, respectively. We also formalized conditions under which the separable effects can be identified in a study where the original treatment has not yet been decomposed and only baseline covariates are measured.

In this paper, we extend the results of Stensrud et al. \cite{stensrud2019separable} in several ways. We define separable effects under a generalized treatment decomposition assumption, allowing settings where the treatment components affect common causes of the event of interest and the competing event. We describe the interpretation of separable effects in this more general setting, where they may not exclusively quantify direct and indirect effects. We also give weaker assumptions for identification of the separable effects, which not only depend on measurements of baseline covariates, but also time-varying covariates. Finally we present semi-parametric weighted
estimators of the separable effects in these generalized settings.


The manuscript is organized as follows. In Section \ref{sec: observed data structure}, we describe the observed data structure in which the event of interest is subject to competing events and both baseline and time-varying covariates are measured. In Section \ref{sec: defining treatment eff}, we review the definition of the total effect on an event of interest subject to competing events. In Section \ref{sec: def decomposition}, we define the generalized decomposition assumption that is agnostic to the mechanism by which the treatment exerts effects on the competing event and the event of interest. In Section \ref{sec: defining sep effects}, we formally define the separable effects. In Section \ref{sec: isolation and interpret}, we formalize a range of conditions by which the treatment components may exert effects on future outcomes and explain the interpretation of the separable effects in each case. The most restrictive of these conditions coincides with those considered by Stensrud et al. \cite{stensrud2019separable} under which the separable effects can be interpreted as direct and indirect effects. In Section \ref{sec: identifiability conditions}, we give conditions that allow identification of the separable effects under the observed data structure by a particular g-formula \cite{robins1986new}. In Section \ref{sec: censoring}, we generalize identification results to allow for censored data. In Section \ref{sec: estimation}, we provide two weighted representations of the g-formula for the separable effects and use these representations to motivate weighted estimators, which are supplemented with sensitivity analysis techniques in Appendix \ref{sec: sensitivity analysis}. In Section \ref{sec: sprint implementation}, we apply these results to a randomized study of the effect of intensive versus standard blood pressure therapy on acute kidney injury. In Section \ref{sec: discussion}, we provide a discussion.

\section{Observed data structure}
\label{sec: observed data structure}
We consider an experiment in which $i=1,\ldots,n$ individuals are randomly assigned to one of two treatment arms $A\in \{0,1\}$ at baseline (e.g. $A=0$ and $A=1$ denote assignment to standard and intensive blood pressure therapy, respectively).  We assume that observations are independent and identically distributed and suppress the $i$ subscript. Let $k=0,1,2,...,K+1$ be equally spaced time intervals with interval $k=0$ corresponding to baseline (the interval of randomization) and interval $k=K+1$ the maximum follow-up of interest at or before the administrative end of follow-up (e.g. 60 months).

For $k>0$, let $Y_{k}$ and $D_{k}$ denote indicators of an event of interest (e.g.\ kidney injury) and a competing event (e.g.\ death) by interval $k$, respectively, and $L_{k}$ a vector of individual time-varying covariates in that interval. Define $ D_{0} \equiv Y_{0} \equiv  0$, i.e.\ the population is restricted to those alive and at risk of all events prior to randomization. Further, define $L_0$ as a vector of pre-randomization covariates. We denote the history of a random variable by an overbar, e.g. $\bar{Y}_{k}=(Y_{0},Y_{1},...,Y_{k})$ is the history of the event of interest through interval $k$, and the future of a random variable through $K+1$ by an underline, e.g. $\underline{Y}_{k}=(Y_{k},Y_{k+1},...,Y_{K+1})$. 

Throughout, we assume a temporal order $(D_{k},Y_{k},L_{k})$ in each interval $k > 0$.  Importantly, as interval lengths become arbitrarily small, this temporal order assumption is guaranteed because the probability that two events of any type occur within that interval approaches zero (equivalent to the common assumption in survival analysis of no tied event times).  In this case, the time-varying event history $\overline{D}_{K+1},\overline{Y}_{K+1}$ coincides with the more familiar competing events data structure $\{\tilde{T}=\min(T,G),J\}$ for $T$ the time to failure from any cause, $G$ a censoring time and $J$ an indicator of cause of failure such that $J=0$ when $\tilde{T}=G$ and $J>0$ otherwise (e.g. $J=1$ if failure from kidney injury and $J=2$ if failure from death).  Defining the observed data structure in terms of time-varying failure status, as opposed to the summarized $(\tilde{T},J)$, is essential for understanding identification and interpretation of many causal estimands in survival analysis, including those considered here, and further avoids the assumption that there exists a censoring time $G$ for individuals who are observed to fail (e.g. die) during the follow-up \cite{young2018choice}.

By definition of a competing event, if an individual experiences this event by interval $k$ without history of the event of interest $(Y_{k-1}=0, D_k=1)$ then $\underline{Y}_{k}=0$; an individual who experiences the competing event cannot subsequently experience the event of interest.  For ease of presentation, we will assume no individual is censored by loss to follow-up (that is, $\overline{D}_{K+1},\overline{Y}_{K+1}$ is fully observed for all individuals randomized at baseline) until Section \ref{sec: censoring}.

\section{The total treatment effect on the event of interest}
\label{sec: defining treatment eff}
For any individual in the study population and for $k \in \{0,\ldots,K\}$, let $Y_{k+1}^{a}$ be the indicator of the event of interest by interval $k+1$ had, possibly contrary to fact, he/she been assigned to $A=a$. The contrast
\begin{align}
    \Pr(Y_{k+1}^{a=1}=1) \text{ vs. } \Pr(Y_{k+1}^{a=0}=1)
\label{eq: basic treatment contrast}
\end{align}
is then a \textsl{total effect} of treatment $A$ on the risk of the event of interest by interval $k+1$ in this study population, which also may include treatment effects on the competing event \cite{young2018choice}. 

We will use causal directed acyclic graphs (DAGs) \cite{pearl2009causality} to represent underlying assumptions on how random variables in the study of Section \ref{sec: observed data structure} are generated.  A causal DAG must represent all common causes of any variable represented on the DAG.  For example, the causal DAG in Figure \ref{fig: simplest graphs}a represents a generally restrictive assumption on this data generating process for a subset of time points because it depicts no common causes (measured or unmeasured) of event status over time.  Throughout we will assume that causal DAGs represent a Finest Fully Randomized Causally Interpreted Structural Tree Graph (FFRCISTG) model, a type of counterfactual causal model that generalizes the non-parametric structural equation model with independent errors (NPSEM-IE) \cite{robins1986new, robins2010alternative,pearl2009causality}, and we assume that statistical independencies in the data are faithful to the DAG \cite{verma1991equivalence}. 

The total effect of $A$ on $Y_2$ in Figure \ref{fig: simplest graphs}a includes all directed (causal) paths between $A$ and $Y_2$.  This includes causal paths that do not capture the treatment's effect on the competing event (e.g.\, $A \rightarrow Y_1 \rightarrow Y_2$ and $A \rightarrow Y_2$) as well as causal paths that capture this effect (e.g.\ $A \rightarrow D_1 \rightarrow D_2 \rightarrow Y_2$ and $A \rightarrow D_2 \rightarrow Y_2$). While the total effect can be straightforward to identify from a study in which $A$ is randomly assigned, its interpretation is complicated when pathways like $A\rightarrow D_2\rightarrow Y_2$ in Figure \ref{fig: simplest graphs}a are present \cite{young2018choice, stensrud2019separable}. For example, a harmful total effect of intensive versus standard blood pressure therapy on kidney injury, i.e. $\Pr(Y_{k+1}^{a=1}=1)>\Pr(Y_{k+1}^{a=0}=1)$, may be wholly or partially explained by one of these pathways (e.g.\ a protective effect of intensive therapy on death).  

\section{Generalized decomposition assumption}
\label{sec: def decomposition}
Consider the following assumption:\\
\begin{enumerate} 
    \item[] \underline{Generalized decomposition assumption}: \begin{align}
    & \text{The treatment $A$ can be decomposed into two binary} \nonumber \\
    & \text{components $A_Y\in \{0,1\}$ and $A_D\in \{0,1\}$ such that,} \nonumber \\
    & \text{in the observed data, the following determinism holds} \nonumber \\
    & A\equiv A_D\equiv A_Y, \text{but in a future study, $A_Y$ and $A_D$} \nonumber \\
    & \text{could, in principle, be assigned different values.}
    \label{assumption: Determinsm}
\end{align}
Let $\overline{Z}_k$, $k\in \{0,  \dots ,K\}$, be the vector of all (direct or indirect) causes of $\underbar{Y}_{k+1}$ and/or $\underbar{D}_{k+1}$, excluding $(A_Y,A_D)$, and $Z_j, j=0,\dots,k$ are the values of these causes in interval $j$.
 We also assume that an intervention that assigns $A=a$ results in the same outcome as an intervention that assigns $A_Y=A_D=a$, that is, 
\begin{align}
& Y_{k+1}^{a_Y=a,a_D=a}=Y_{k+1}^{a} \nonumber \\
& D_{k+1}^{a_Y=a,a_D=a}=D_{k+1}^{a},  \nonumber \\
& Z_{k+1}^{a_Y=a,a_D=a}=Z_{k+1}^{a}, \quad k \in \{0,  \dots ,K\}, \label{eq: definition A=Ay=Ad}
\end{align}
where $W_{k+1}^{a_Y,a_D}$ for $W_{k+1} \in \{Y_{k+1},D_{k+1},Z_{k+1}\}$ $k \in \{0,\ldots,K\}$, is the value of $W_{k+1}$ had, contrary to fact, he/she been assigned the components $A_Y=a_Y$ and $A_D=a_D$, in place of assignment to a value of the original treatment $A$. 
\end{enumerate}


Beyond \eqref{eq: definition A=Ay=Ad}, the generalized decomposition assumption makes no mechanistic assumptions on the effects exerted by $A_Y$ and $A_D$. We will consider different examples of treatment decompositions in Section \ref{sec: isolation and interpret}. In Appendix \ref{sec: modified treatment} we consider straightforward further generalizations of our results to settings where $A_Y$ and $A_D$ are not a decomposition of $A$, violating \eqref{assumption: Determinsm}, but are still treatments satisfying \eqref{eq: definition A=Ay=Ad}. 

\section{Separable effects} 
\label{sec: defining sep effects}

Following Stensrud et al \cite{stensrud2019separable}, for $k \in \{0,\ldots,K\}$,
\begin{equation}
\Pr (Y_{k+1}^{a_Y=1,a_D}=1)\text{ vs. }\Pr
(Y_{k+1}^{a_Y=0,a_D}=1), \quad a_D \in \{0,1\},
\label{eq: A_Y sep}
\end{equation}
quantifies the causal effect of the $A_Y$ component on the risk of the event of interest by $k+1$ under an intervention that assigns $A_D=a_D$.  Similarly
\begin{equation}
\Pr (Y_{k+1}^{a_Y,a_D=1}=1)\text{ vs. }\Pr
(Y_{k+1}^{a_Y,a_D=0}=1), \quad a_Y \in \{0,1\},
\label{eq: A_D sep}
\end{equation}%
quantifies the causal effect of the $A_D$ component on the risk of the event of interest by $k+1$ under an intervention that assigns $A_Y=a_Y$.

We will refer to \eqref{eq: A_Y sep} as the \textsl{$A_Y$ separable effect under $A_D=a_D$}, $a_D\in\{0,1\}$ and \eqref{eq: A_D sep} as the \textsl{$A_D$ separable effect under $A_Y=a_Y$}, $a_Y\in\{0,1\}$.  Given the generalized decomposition assumption, the total effect can be expressed as a sum of particular $A_Y$ and $A_D$ separable effects, for example,
\begin{align*}
& \Pr (Y_{k+1}^{a_Y=1,a_D=1}=1)-\Pr (Y_{k+1}^{a_Y=0,a_D=1}=1) \\
& +\Pr (Y_{k+1}^{a_Y=0,a_D=1}=1)-\Pr (Y_{k+1}^{a_Y=a_D=0}=1)  \\
& =\Pr (Y_{k+1}^{a=1}=1)-\Pr (Y_{k+1}^{a=0}=1).
\end{align*}

\section{Isolation conditions and interpretation of separable effects}
\label{sec: isolation and interpret}

In this section, we consider conditions, beyond the generalized decomposition assumption, under which we can ascribe a more precise interpretation to the separable effects \eqref{eq: A_Y sep} and \eqref{eq: A_D sep}. The strongest of these assumptions allows interpretation of these effects  as the separable direct and indirect effects of Stensrud et al \cite{stensrud2019separable}. To formally define these additional conditions, we will first review the definition of an \emph{extended causal DAG} \cite{robins2010alternative}: an extended causal DAG augments the original causal DAG with additional nodes representing components of the treatment, and bold vertices representing the deterministic relation between these components and the full treatment in the observed data. For example, the extended causal DAG in Figure \ref{fig: simplest graphs}b is an augmented version of the causal DAG in Figure \ref{fig: simplest graphs}a. The extended causal DAG also encodes assumptions, not represented on the original causal DAG, on the mechanisms by which each treatment component exerts effects on future variables.

The interpretation of the separable effects will also depend on the nature of $\overline{Z}_k$, $k>0$ as defined in Section \ref{sec: def decomposition}. We will remain agnostic until Section \ref{sec: identifiability conditions} as to whether all or some of the components of $\overline{Z}_k$ are measured or unmeasured in our study of Section \ref{sec: observed data structure}. Therefore the components of $\overline{Z}_k$ may or may not partially or fully coincide with the components of the measured covariate history $\overline{L}_k$.  As we will discuss in Section \ref{sec: identifiability conditions}, the overlap (or lack thereof) of $\overline{Z}_k$ and $\overline{L}_k$ will impact our ability to identify the separable effects under the observed data structure of Section \ref{sec: observed data structure}.  


\subsection{Full isolation}
\label{sec: full iso}
Consider an extended causal DAG in which $A$ is decomposed into two components $A_Y$ and $A_D$ satisfying the generalized decomposition assumption \eqref{eq: definition A=Ay=Ad}, and define the following conditions:
\begin{align}
    & \text{The only causal paths from } A_Y \text{ to } D_{k+1}, k=0,...,K \text{ are directed} \nonumber \\
    &\text{paths intersected by }
     Y_{j},  j=0,...,k. \label{def: Ay partial iso} \\ 
    & \text{The only causal paths from } A_D \text{ to } Y_{k+1},  k=0,...,K \text{ are directed} \nonumber \\ 
    &\text{paths intersected by }
    D_{j+1}, j=0,...,k. \label{def: Ad partial iso} 
\end{align}
When both conditions \eqref{def: Ay partial iso} and \eqref{def: Ad partial iso} hold we will say there is \textsl{full isolation}. This assumption is satisfied in Figure \ref{fig: simplest graphs}b which assumes there are no common causes of the event of interest and the competing event. It is also satisfied in Figure \ref{fig:full isolation}b which allows the presence of both pre-randomization ($Z_0$) and post-randomization ($Z_1$) common causes.

Under the generalized decomposition assumption and full isolation, \eqref{eq: A_Y sep} are the separable direct effects of $A$ on the risk of the event of interest by $k+1$, which do not capture the treatment's effect on the competing event, and \eqref{eq: A_D sep} are the separable indirect effects of $A$ on this risk, which only capture the treatment's effect on the competing event. Full isolation coincides with the settings considered by Stensrud et al \cite{stensrud2019separable}, which allowed for the presence of pre-randomization, but not post-randomization, common causes of the event of interest and the competing event.

Returning to our running example, assume that the blood pressure treatment $A$ can be decomposed into a component $A_Y$ that binds to receptors in the kidneys, e.g.\ by relaxing the efferent arterioles which is a well-known biological effect of commonly used blood pressure drugs such as angiotensin-converting-enzyme inhibitors (ACE)
and angiotensin II receptor blockers (ARB), and a component $A_D$ that includes the remaining components of the antihypertensive therapy, some of which lead, for example, to  reductions in systemic blood pressure.

Then, $A_Y=1$ and $A_Y=0$ are the levels (doses) of the $A_Y$ component under standard and intensive therapy, respectively, and $A_D=1$ and $A_D=0$ are defined analogously.

Full isolation would be satisfied in this case if (i) the $A_Y$ component only exerts effects on death through its effects on kidney function and (ii) the remaining $A_D$ component only exert effects on kidney function through its effects on survival.  In Section \ref{sec: Ay iso}, however, we argue that the assumption of full isolation may not be reasonable in this setting. 

\subsection{$A_Y$ partial isolation}
\label{sec: Ay iso}
The causal graphs in Figures \ref{fig: simplest graphs} and \ref{fig:full isolation} make the restrictive assumption that there are no common causes of the event of interest and competing event that are, themselves, affected by treatment.  In our running example, this assumption likely fails: a reduction in blood pressure may increase the risk of kidney injury (the event of interest) due to hypoperfusion of the kidneys (for example, when patients are dehydrated) \cite{aalen2019time} and also may affect the risk of mortality (the competing event).  Further, blood pressure itself clearly may be affected by the blood pressure treatment. The causal DAG in Figure \ref{fig:dagAarrowLk} depicts the more realistic assumption that blood pressure ($Z_1$) is both a possible common cause of future kidney injury $Y_2$ and mortality $D_2$ and also, itself, affected by treatment $A$ (represented by the blue arrow connecting $A$ to $Z_1$).    

Suppose, however, that the $A_Y$ component of the treatment $A$ (that which binds to receptors in the kidneys) has no effect on blood pressure outside of its possible effect on kidney function, such that only the remaining components of treatment, $A_D$, can directly affect blood pressure. The extended DAG in Figure \ref{fig:partial isolation}a, which is one possible extension of the causal DAG in Figure \ref{fig:dagAarrowLk}, represents this assumption by the blue arrow from $A_D$ into $Z_1$ and the absence of an arrow from $A_Y$ into $Z_1$. In this case, condition \eqref{def: Ay partial iso} holds but \eqref{def: Ad partial iso} does not.  When only the condition \eqref{def: Ay partial iso} holds, but \eqref{def: Ad partial iso} fails, we will say there is \textsl{$A_Y$ partial isolation}.  

Unlike under full isolation, under $A_Y$ partial isolation, the $A_D$ separable effects \eqref{eq: A_D sep} quantify \textsl{both} direct effects of the treatment on the event of interest not through the competing event (e.g. the path $A_D \rightarrow Z_1 \rightarrow Y_2$ in Figure \ref{fig:partial isolation}a) and indirect effects through the competing event (e.g. the path $A_D \rightarrow D_1 \rightarrow Y_1 \rightarrow Y_2$ in Figure \ref{fig:partial isolation}a).  By contrast, the $A_Y$ separable effects \textsl{only} quantify direct effects not through the competing event. However, the $A_Y$ separable effects do not capture all direct effects in this case, because some of these pathways may originate from $A_D$ as described above. In the current example, the $A_Y$ separable effect evaluated at $a_D=1$ may be of particular clinical interest, quantifying the effect of assignment to the current intensive therapy containing all components versus a modified intensive therapy that lacks the component possibly affecting the kidneys. 

\subsection{$A_D$ partial isolation}
\label{sec: Ad iso}
When \eqref{def: Ad partial iso} holds, but \eqref{def: Ay partial iso} fails, we will say there is \textsl{$A_D$ partial isolation}. $A_D$ partial isolation is represented in Figure \ref{fig:partial isolation}b, depicting an alternative augmentation of the causal DAG in Figure \ref{fig:dagAarrowLk}.  Under $A_D$ partial isolation, the $A_Y$ separable effects \eqref{eq: A_Y sep} quantify both direct effects of the treatment on the event of interest not through the competing event (e.g. the path $A_Y \rightarrow Z_1 \rightarrow Y_2$ in Figure \ref{fig:partial isolation}b) and indirect effects through the competing event (e.g. the path $A_Y \rightarrow Z_1 \rightarrow D_2 \rightarrow Y_2$ in Figure \ref{fig:partial isolation}b).  By contrast, the $A_D$ separable effects \textsl{only} quantify indirect effects through the competing event. However, the $A_D$ separable effects do not capture all indirect effects in this case, because some of these pathways may originate from $A_Y$ as above. 

As an example of $A_D$ partial isolation, trials have reported an increase in the risk of new-onset type 2 diabetes among patients assigned to statins \cite{sattar2010statins,ridker2012cardiovascular}.  However, statins also reduce the risk of all-cause mortality, a competing event for type 2 diabetes onset (the event of interest).  It is therefore unclear whether a total effect of statin treatment on type 2 diabetes is due a protective treatment effect on mortality, a biologically harmful process leading to type 2 diabetes onset or some combination.  

Figure \ref{fig:dagAarrowLk} illustrates a possible underlying causal structure for a trial with random assignment to statin therapy relating treatment assignment $A$, mortality $D_k$ and new-onset type 2 diabetes $Y_k$, $k=1,2$. Body weight ($Z_1$) is a possible common cause of both mortality and onset of type 2 diabetes which may also be affected by statin treatment.  Consider a decomposition of $A$ (represented in Figure \ref{fig:partial isolation}b) where $A_D$ may lead to increased risk of diabetes only by reducing mortality risk (e.g.\ through $A_D \rightarrow D_1 \rightarrow D_2 \rightarrow Y_2$, where the reduction in mortality risk is likely due to reduced levels of low density lipoprotein in the blood), while a second component $A_Y$ exerts unintended effects of statins on diabetes through body weight (e.g.\ through $A_Y \rightarrow Z_1 \rightarrow Y_2$). As in the previous example of blood pressure therapy and kidney injury, the $A_Y$ separable effect of statin therapy on type 2 diabetes risk evaluated at $a_D=1$ may be of particular clinical interest, quantifying the effect of assignment to the original statin therapy containing both components versus a modified treatment that removes the component possibly leading to weight gain.  

\subsection{No isolation}
If there are direct arrows from $A_Y$ and $A_D$ into common causes of $Y_{k+1}$ and $D_{k+1}$, $k \in \{0,\ldots,K\}$, as illustrated in Figure \ref{fig:no isolation}, then both \eqref{def: Ay partial iso} and \eqref{def: Ad partial iso} fail. In this case, both the $A_Y$ separable effects \eqref{eq: A_Y sep} and the $A_D$ separable effects \eqref{eq: A_D sep} quantify direct and indirect effects of the treatment on the event of interest, outside of and through, the competing event. When both conditions \eqref{def: Ay partial iso} and \eqref{def: Ad partial iso} fail, we will say there is \textsl{no isolation}. 

There are two important cases of no isolation that have different implications for the interpretation of separable effects and, as we will see, their identification in a two-arm trial. First, suppose there are direct arrows from $A_Y$ and $A_D$ into the same set of common causes $Z_k$ of $\underline{Y}_{k+1}$ and $\underline{D}_{k+1}$, as illustrated in Figure \ref{fig:no isolation}a.  In this case, the $A_Y$ separable effects and the $A_D$ separable effects will capture common downstream pathways (e.g. $Z_1 \rightarrow Y_2$ in Figure \ref{fig:no isolation}a) between the original treatment $A$ and the event of interest $Y_{k+1}$.   

Alternatively, suppose $A_Y$ and $A_D$ may only exert effects on \emph{different} sets of common causes $Z_{A_Y,1}$ and $Z_{A_D,1}$ of $Y_{k+1}$ and $D_{k+1}$ as illustrated in Figure \ref{fig:no isolation}b; here $A_Y$ exerts effects on $Y_{k+1}$ through one set of causal paths from $A_Y$ to $Y_{k+1}$, and  $A_D$ exerts effects on $Y_{k+1}$ through a distinct set of causal paths.  In this case, the $A_Y$ separable effects and the $A_D$ separable effects will capture no common pathways between the original treatment $A$ and the event of interest $Y_{k+1}$. 

\subsection{$Z_k$ partition}
\label{sec: zk partition}
Suppose there exist vectors $Z_{A_D,k}, Z_{A_Y,k}$ such that $Z_k \equiv (Z_{A_D,k}, Z_{A_Y,k})$, $k>0$, and
\begin{align}
& \text{The only causal paths from } A_Y \text{ to } D_{k+1} \text{ and } Z_{A_D,k+1}, k=0,...,K \text{ are through } \nonumber  \\
&  Y_{j} \text{ or any component of } Z_{A_Y,j}, j=0,...,k.  \label{def: Ay partitioning Z}  \\
& \text{The only causal paths from } A_D \text{ to } Y_{k+1} \text{ and } Z_{A_Y,k+1},  k=0,...,K \text{ are through } \nonumber  \\ 
&  D_{j+1} \text{ or any component of } Z_{A_D,j},  j=0,...,k. \label{def: Ad partitioning Z}
\end{align}
When both conditions \eqref{def: Ay partitioning Z} and \eqref{def: Ad partitioning Z} hold we will say there is a \textsl{$Z_k$ partition}. 

The assumption of a $Z_k$ partition holds trivially under full isolation for any partition of $Z_k$ as illustrated in Figure \ref{fig:full isolation}b. However, this assumption will only hold in some cases of partial isolation (e.g. Figure \ref{fig:partial isolation}) and no isolation (e.g. Figure \ref{fig:no isolation}b).  $Z_k$ partition fails under the case of no isolation represented in Figure \ref{fig:no isolation}a.  It also fails under the case of $A_Y$ partial isolation represented in Figure \ref{fig: Z_k vs partial iso}a and $A_D$ partial isolation represented in Figure \ref{fig: Z_k vs partial iso}b.   Under any version of $Z_k$ partition, the $A_Y$ separable effects and the $A_D$ separable effects will capture no common pathways between the original treatment $A$ and the event of interest $Y_{k+1}$.

\section{Identification of separable effects}
\label{sec: identifiability conditions}
Regardless of the isolation assumptions that impact the interpretation of separable effects, if we had data from a four-arm trial in which $A_Y$ and $A_D$ were randomly assigned with no loss to follow-up, we would be guaranteed identification of the separable effects \cite{stensrud2019separable, robins2016uk}; that is, we could identify, for $k \in \{0,\ldots,K\}$, 

\begin{align}
& \Pr (Y_{k+1}^{a_Y,a_D}=1) \text{ for } a_Y,a_D \in \{0,1\}
\label{eq: estimand}
\end{align}%
by $\Pr (Y_{k+1}=1 \mid A_Y = a_Y, A_D = a_D)$ \cite{hernan2018causal}. However, in order to identify \eqref{eq: estimand} for $a_Y\neq a_D$ in the absence of a four-arm trial, we must make untestable assumptions that are not guaranteed to hold, even in a two-armed trial such as that described in Section \ref{sec: observed data structure} with no loss to follow-up. In addition to the generalized decomposition assumption, consider the following assumptions: 
\begin{enumerate}
\item[1.] Exchangeability: \begin{align} 
& \underline{Y}_{1}^{a},\underline{D}_{1}^{a},\underline{L}^{a}_{1} \independent A \mid L_{0}. \label{ass: exchangeability 1}
\end{align}
Exchangeability is expected to hold in a study where $A$ is randomly assigned,  possibly conditional on the measured baseline covariates $L_0$, but is otherwise not guaranteed to hold \cite{hernan2018causal}.  For example, Figure \ref{fig:partial isolationExpanded} illustrates various extended graphs that explicitly depict measured (e.g.\ $L_1$) and unmeasured (e.g. $U_{L,Y}$) variables. Exchangeability is represented in Figures \ref{fig:partial isolationExpanded}a-f by the absence of any unblocked backdoor paths between $A$ and $(\underline{Y}_{1},\underline{D}_{1},\underline{L}_{1}$) conditional on $L_0$ \cite{pearl2009causality}.

\item[2.] Consistency: 
\begin{align}
  & \text{If } A=a, \nonumber \\
  & \text{then } \bar{Y}_{k+1} = \bar{Y}^{a}_{k+1}, \bar{D}_{k+1} = \bar{D}^{a}_{k+1} \text{ and } \bar{L}_{k+1} = \bar{L}^{a}_{k+1} \text{ for } k \in \{0,\ldots,K\}.
  \label{ass: consistency}
\end{align}
Consistency states that if an individual has observed treatment consistent with an intervention that sets $A=a$, then that individual's future observed outcomes and time-varying covariates are equal to his/her counterfactual outcomes and time-varying covariates, respectively, under an intervention that sets $A=a$.  Consistency holds in a study where $A$ is randomly assigned but is otherwise not guaranteed.
\item[3.] Positivity: 
\begin{align}
& f_{L_0}(l_0)>0\implies \nonumber  \\
& \quad \Pr (A=a\mid  L_0=l_0)>0, \text{ for } a\in\{0,1\}  \label{eq: positivity of A} \\
& f_{\overline{L}_k,D_{k+1},Y_k}(\overline{l}_k,0,0)> 0  \implies \nonumber\\ 
& \quad  \Pr(A=a|D_{k+1}=Y_k=0,\overline{L}_k=\overline{l}_k)>0, \nonumber\\
&\text{ for } a\in\{0,1\}\text{ and }k \in \{0,\ldots,K\}. \label{eq: positivity of A part 2}
\end{align}%
Assumption \eqref{eq: positivity of A} states that, for any possibly observed level of the measured baseline covariates, there exist individuals with $A=1$ and individuals with $A=0$.  This assumption  will hold by design in a randomized controlled trial like that of Section \ref{sec: observed data structure}. Assumption \eqref{eq: positivity of A part 2} requires that for any possibly observed level of the measured time-varying covariate history among those surviving all events through each follow-up time, there exist individuals with $A=1$ and individuals with $A=0$. Even when $A$ is randomized, assumption \eqref{eq: positivity of A part 2} does not hold by design.  However, it can be assessed in the observed data. 

\item[4.] Dismissible component conditions: \\
Let $G$ refer to a hypothetical four-arm trial in which both $A_Y$ and $A_D$ are randomly assigned, possibly to different values.  We define the following conditions for $k \in \{0,\ldots,K\}$:
\begin{align} 
& Y_{k+1}(G) \independent A_D(G) \mid A_Y(G), D_{k+1}(G)=Y_{k}(G)=0, \bar{L}_k(G), \label{ass: delta 1}\\
& D_{k+1}(G) \independent A_Y(G) \mid A_D(G), D_{k}(G)=Y_{k}(G)=0, \bar{L}_k(G), \label{ass: delta 2} \\ 
& L_{A_Y,k}(G) \independent A_D(G) \mid A_Y(G), Y_{k}(G)=D_{k}(G)=0, \bar{L}_{k-1}(G),L_{A_D,k}(G),  \label{ass: delta 3a} \\
& L_{A_D,k}(G) \independent A_Y(G) \mid A_D(G), D_{k}(G)=Y_{k}(G)=0, \bar{L}_{k-1}(G),  \label{ass: delta 3b}
\end{align} 
where $Y_{k+1}(G)$, $ D_{k+1}(G)$ and $L_{k}(G) \equiv (L_{A_Y,k}(G),L_{A_D,k}(G))$ are values of the outcome of interest, the competing outcome and (a temporally ordered partitioning of) the measured covariates at $k+1$, respectively, had we, contrary to fact, randomly assigned $A_Y(G)$ and $A_D(G)$, the values of $A_Y$ and $A_D$ assigned an individual under $G$, respectively. It follows directly from the generalized decomposition assumption that, using d-separation rules \cite{robins2010alternative,pearl2009causality}, the dismissible component conditions can be read off of a transformation of the extended causal DAG, representing an augmented version of our original data generating assumption, in which $A$ and the deterministic arrows originating from $A$ are eliminated. See similar results in Didelez \cite[Figure 2]{didelez2018defining}. 

For example, consider Figure \ref{fig:transformG}a, a transformation of Figure \ref{fig:partial isolation}a, which assumes that $L_k \equiv Z_k$, (i.e., all common causes of the event of interest and competing event are measured). Assumption \eqref{ass: delta 1} holds in Figure \ref{fig:transformG}a by the absence of any unblocked backdoor paths between $A_D(G)$ and $Y_2(G)$ conditional on $A_Y(G)$, $D_1(G)$, $D_2(G)$, $L_1(G)$ and $Y_1(G)$, and similarly assumption \eqref{ass: delta 2} holds due to the absence of any unblocked paths between  $A_Y(G)$ and $D_2(G)$ conditional on $A_D(G)$, $D_1(G)$, $L_1(G)$ and $Y_2(G)$. Analogously, by choosing $L_{k}(G) = (L_{A_Y,k}(G),\emptyset),k=1,2$,  \eqref{ass: delta 3a} and \eqref{ass: delta 3b} also hold in Figure \ref{fig:transformG}a.

Consider also the examples in Figure \ref{fig:partial isolationExpanded}; under $G$ transformations of each graph, all dismissible component conditions hold in Figures \ref{fig:partial isolationExpanded}a-d, where $L_{A_D,1}= L_1$ and $L_{A_Y,1}= \emptyset$ in Figures \ref{fig:partial isolationExpanded}a-c. By contrast, Figures \ref{fig:partial isolationExpanded}e-f illustrate failure of these conditions under their $G$ transformations. For example, while \eqref{ass: delta 2}-\eqref{ass: delta 3b} hold in Figure \ref{fig:partial isolationExpanded}e, \eqref{ass: delta 1} is violated by the the unblocked collider path $A_D(G) \rightarrow \boxed{D_2(G)} \leftarrow U_{L,D} \rightarrow \boxed{L_1(G)} \leftarrow U_{L,Y} \rightarrow Y_2(G) $, regardless of whether we define $L_{A_D,1}= L_1$ and $L_{A_Y,1}= \emptyset$ or  $L_{A_Y,1}= L_1$ and $L_{A_D,1}= \emptyset$.  Similarly, in Figure \ref{fig:partial isolationExpanded}f, while \eqref{ass: delta 2}-\eqref{ass: delta 3b} hold when we define $L_{A_D,1}= L_1$ and $L_{A_Y,1}= \emptyset$, \eqref{ass: delta 1} is violated by the unblocked collider path $A_D(G) \rightarrow \boxed{L_1(G)} \leftarrow U_{L,Y} \rightarrow Y_2(G)$.
\end{enumerate}

\subsection{Relation between isolation and dismissible component conditions}\label{isodis}
Note that $Z_k$ partition is a \textsl{necessary} condition for the dismissible component conditions to hold for any choice of measured covariates $L_k$ and their partition (see proof in Appendix \ref{sec: appendix lemmas}). However, $Z_k$ partition is not \textsl{sufficient} to ensure these conditions as also illustrated by Figure \ref{fig:partial isolationExpanded}.  For example, in Figure \ref{fig:partial isolationExpanded}e, full isolation holds but, as we noted above, the dismissible component conditions fail due to failure to measure either the common cause $U_{L,D}$ or $U_{L,Y}$.  Similarly, the graph in Figure \ref{fig:partial isolationExpanded}f satisfies $Z_k$ partition, but, as we noted above, the dismissible component conditions fail due to failure to measure the common cause $U_{L,Y}$. These results on identification and those of the previous section coincide with previous identification results on identification of path-specific effects \cite{shpitser2013counterfactual, avin2005identifiability} but without explicit consideration of competing events.   

In Appendix \ref{sec: appendix lemmas} we also show that: (i) if the dismissible component conditions hold when we select the $L_k$ partition $L_{A_D,k}=L_k$ and $L_{A_Y,k}=\emptyset$ for all $k\in\{1,\dots,K\}$, then $A_Y$ partial isolation holds; (ii) if the dismissible component conditions hold when we select the $L_k$ partition $L_{A_Y,k}=L_k$ and $L_{A_D,k}=\emptyset$ for all $k\in\{1,\dots,K\}$, then $A_D$ partial isolation holds; and (iii) if the dismissible component conditions hold when we select either of the partitions in (i) and (ii) then full isolation holds and $L_k$ does not depend on $A$ at any $k$, given the measured past.

Note that we may consider a decomposition of $A$ into more than two components under a stronger version of the generalized decomposition assumption (Appendix \ref{sec: A_Z decomposition}). Then we can define alternative versions of separable effects indexed by interventions that may assign different values to these components. Under failure of $Z_k$ partition for a two-way decomposition, identification may be possible for separable effects indexed by multiple components but under additional assumptions. We briefly consider these alternative definitions of separable effects and their identification for the case of a three-way decomposition of $A$ in Appendix \ref{sec: A_Z decomposition}.

\subsection{The g-formula for separable effects}
\label{sec: identification g formula}
For $k \in \{0,\ldots,K\}$, let $l_k = (l_{A_y,k},l_{A_D,k})$ be a realization of the measured time-varying covariates at $k$, such that $l_{A_Y,k}$ and $l_{A_D,k}$ are possible realizations of $L_{A_y,k}$ and $L_{A_D,k}$, respectively (a chosen partition of $L_k$ under an assumed temporal order $L_{A_D,k}, L_{A_Y,k}$). Provided that exchangeability, consistency, positivity and the 4 dismissible component conditions hold, we can identify $\Pr(Y^{a_Y, a_D}_{k+1}=1)$ by 
\begin{align}
       & \sum_{\bar{l}_k}  \Big[ \sum_{s=0}^{k} \Pr(Y_{s+1}=1 \mid  D_{s+1}= Y_{s}=0, \bar{L}_{s} = \bar{l}_{s}, A = a_Y) \nonumber \\ 
       &  \prod_{j=0}^{s}  \big\{ \Pr(D_{j+1}=0 \mid D_{j}= Y_{j}=0, \bar{L}_{j} = \bar{l}_{j},  A = a_D)  \nonumber \\
      &  \times \Pr(Y_{j}=0 \mid Y_{j-1}=D_{j}=0, A = a_Y) \nonumber \\
    &\times \Pr(L_{A_Y,j}=l_{A_Y,j} \mid Y_{j} = D_{j} = 0, \bar{L}_{j-1}=\bar{l}_{j-1},  L_{A_D,j}= l_{A_D,j}, A = a_Y) \nonumber \\
    &\times \Pr(L_{A_D,j}=l_{A_D,j} \mid Y_{j} = D_{j} = 0, \bar{l}_{j-1},  A = a_D) \big\} \Big],
\label{eq: identifying formula nc}
\end{align}
$k \in \{0,\ldots,K\}$ where, generally, for any vector of random variables $A$ and $B$, $f(a|b)\equiv f_{A|B}(a|b)$ is the joint density of $A$ given $B$ evaluated at $a,b$. See Appendix \ref{sec: proof of idenditifiability} for proof. We will refer to expression \eqref{eq: identifying formula nc} as the \textsl{g-formula} \cite{robins1986new} for $\Pr(Y^{a_Y, a_D}_{k+1}=1)$. 

\section{Censored data}
\label{sec: censoring}
We now relax the assumption of no losses to follow-up, allowing that some individuals are censored at some point during the study. For $k>0$, let $C_{k}$ denote censoring by loss to follow-up by interval $k$, and assume a temporal order $(C_{k},D_{k},Y_{k},L_{k})$ in each interval $k > 0$. We remind the reader that the the temporal ordering assumption is analogous to assumptions about ties in continuous time settings, which becomes practically irrelevant when the time intervals are small. Hereby, we will implicitly redefine all counterfactual outcomes $Y_{k+1}^{a_Y, a_D}$ in terms of outcomes under an additional intervention that eliminates censoring. 

When censoring is present, the isolation conditions defined in Section \ref{sec: isolation and interpret} and their implications for interpretation of separable effects are unchanged.  However, in this case, additional exchangeability, positivity and consistency assumptions are required for identification of \eqref{eq: estimand} using only the observed data. Given assumptions \eqref{ass: E1 app}-\eqref{ass: delta 3b app} in Appendix \ref{sec: proof of idenditifiability}, which extend the  assumptions of Section \ref{sec: identifiability conditions} to allow that censoring is present and dependent on the measured time-varying risk factors $\overline{L}_k$, we can identify \eqref{eq: estimand} by 
\begin{align}
       & \sum_{\bar{l}_k}  \Big[ \sum_{s=0}^{k} \Pr(Y_{s+1}=1 \mid C_{s+1}= D_{s+1}= Y_{s}=0,\bar{L}_{s} = \bar{l}_{s}, A=a_Y) \nonumber \\ 
       &  \prod_{j=0}^{s}  \big\{ \Pr(D_{j+1}=0 \mid C_{j+1}=D_{j}= Y_{j}=0, \bar{L}_{j} = \bar{l}_{j},  A = a_D)  \nonumber \\
      &  \times \Pr(Y_{j}=0 \mid C_{j}=D_{j}= Y_{j-1}=0,\bar{L}_{j} = \bar{l}_{j}, A = a_Y) \nonumber \\
    &\times f(L_{A_Y,j}=l_{A_Y,j} \mid C_{j}= Y_{j} = D_{j} = 0, \bar{L}_{j-1} = \bar{l}_{j-1}, L_{A_D,j}= l_{A_D,j}, A = a_Y) \nonumber \\
    &\times f(L_{A_D,j}=l_{A_D,j} \mid C_{j}= Y_{j} = D_{j} = 0,  L_{j-1}=\bar{l}_{j-1},  A = a_D) \big\} \Big].
\label{eq: identifying formula}
\end{align}
See Appendix \ref{sec: proof of idenditifiability} for proof. We say expression \eqref{eq: identifying formula} is the g-formula for \eqref{eq: estimand} under elimination of censoring. When assumptions \eqref{ass: E1 app}-\eqref{ass: delta 3b app} hold replacing $\overline{L}_k = L_0$, $k \in \{0,\ldots,K\}$ then identification of \eqref{eq: estimand} is achieved by a simplified version of \eqref{eq: identifying formula} which was given in Stensrud et al \cite{stensrud2019separable}.

\section{Estimation of separable effects}
\label{sec: estimation}
The g-formula \eqref{eq: identifying formula} has the following alternative representations, 
\begin{align}
   \sum_{s=0}^{k} E & [ W_{C,s}(a_Y) W_{D,s}(a_Y,a_D) W_{L_{A_D},s}(a_Y,a_D) (1-Y_{s}) (1-D_{s+1}) Y_{s+1} \mid A=a_Y], 
    \label{eq: alternative id formula 1}
\end{align}
where 
\begin{align*}
W_{D,s} (a_Y,a_D) &= \frac{\prod_{j=0}^{s}  \Pr(D_{j+1}=0 \mid C_{j+1}=D_{j}= Y_{j}=0, \bar{L}_{j},  A = a_D) }{ \prod_{j=0}^{s} \Pr(D_{j+1}=0 \mid C_{j+1}=D_{j}= Y_{j}=0, \bar{L}_{j},  A = a_Y) }, \\
W_{L_{A_D},s} (a_Y,a_D) &= \frac{\prod_{j=0}^{s}   \Pr(A =a_D \mid C_{j}= Y_{j} = D_{j} = 0, L_{A_D,j}, \bar{L}_{j-1}) }{ \prod_{j=0}^{s}   \Pr(A =a_Y \mid C_{j}= Y_{j} = D_{j} = 0,  L_{A_D,j}, \bar{L}_{j-1}) } \nonumber \\
& \times  \frac{\prod_{j=0}^{s}   \Pr(A =a_Y \mid C_{j}= Y_{j} = D_{j} = 0, \bar{L}_{j-1}) }{ \prod_{j=0}^{s}   \Pr(A =a_D \mid C_{j}= Y_{j} = D_{j} = 0, \bar{L}_{j-1}) }, \nonumber \\
W_{C,s} (a_D) &= \frac{I(C_{s+1} =0) }{ \prod_{j=0}^{s}  \Pr(C_{j+1}=0 \mid C_{j}=D_{j}= Y_{j}=0, \bar{L}_{j},  A = a_D) }, 
\end{align*}
and 
\begin{align}
   \sum_{s=0}^{k} E & \{ W_{C,s}(a_D) W_{Y,s}(a_D,a_Y) W_{L_{A_Y},s}(a_D,a_Y) (1-Y_{s}) (1-D_{s+1}) Y_{s+1} \mid A=a_D \}, 
    \label{eq: alternative id formula 2}
\end{align}
where $W_{C,s} (a_D)$ is defined as in \eqref{eq: alternative id formula 1} and
\begin{align*}
W_{Y,s} (a_D,a_Y) & = \frac{ \Pr(Y_{s+1}=1 \mid C_{s+1}=D_{s+1}= Y_{s}=0, \bar{L}_{s},  A = a_Y) }{\Pr(Y_{s+1}=1 \mid C_{s+1}=D_{s+1}= Y_{s}=0, \bar{L}_{s},  A = a_D) } \\
& \times \frac{\prod_{j=0}^{s-1}  \Pr(Y_{j+1}=0 \mid C_{j+1}=D_{j+1}= Y_{j}=0, \bar{L}_{j},  A = a_Y) }{ \prod_{j=0}^{s-1} \Pr(Y_{j+1}=0 \mid C_{j+1}=D_{j+1}= Y_{j}=0, \bar{L}_{j},  A = a_D) }, \\
W_{L_{A_Y},s} (a_D,a_Y) & = \frac{\prod_{j=0}^{s}   \Pr(A =a_Y  \mid C_{j}= Y_{j} = D_{j} = 0, \bar{L}_{j}) }{ \prod_{j=0}^{s}   \Pr(A =a_D \mid C_{j}= Y_{j} = D_{j} = 0, \bar{L}_{j}) }   \nonumber \\
& \times \frac{\prod_{j=0}^{s}   \Pr(A =a_D \mid C_{j}= Y_{j} = D_{j} = 0,  L_{A_D,j}, \bar{L}_{j-1}) }{ \prod_{j=0}^{s}   \Pr(A =a_Y \mid C_{j}= Y_{j} = D_{j} = 0, L_{A_D,j}, \bar{L}_{j-1}) },
\end{align*}
as formally shown in Appendix \ref{sec: proof of alternative id}. 


Representations \eqref{eq: alternative id formula 1} and \eqref{eq: alternative id formula 2} motivate weighted estimators of the separable effects, which generalize the weighted estimators given by \cite{stensrud2019separable}.  We let $\nu_{a_Y,a_D,k} $ denote \eqref{eq: identifying formula}, and $(\overline{L}_{k,i},L_{A_D,k,i})$ is study individual $i$'s values of $(\overline{L}_{k},L_{A_D,k})$ for a user-chosen partition of $L_k$. 

Define 
\begin{align*}
\hat{W}_{D,k,i} (a_Y,a_D;\hat{\alpha}_D) &= \frac{\prod_{j=0}^{k}  \Pr(D_{j+1}=0 \mid C_{j+1}=D_{j}= Y_{j}=0, \bar{L}_{j,i},  A = a_D; \hat{\alpha}_D) }{ \prod_{j=0}^{k} \Pr(D_{j+1}=0 \mid C_{j+1}=D_{j}= Y_{j}=0, \bar{L}_{j,i},  A = a_Y ; \hat{\alpha}_D) }, \\
\hat{W}_{L_{A_D},k,i} (a_Y,a_D;\hat{\alpha}_{L_D 1},\hat{\alpha}_{L_D 2}) &= \frac{\prod_{j=0}^{k}   \Pr(A =a_D \mid C_{j}= Y_{j} = D_{j} = 0, L_{A_D,j,i}, \bar{L}_{j-1,i};\hat{\alpha}_{L_D 1}) }{ \prod_{j=0}^{k}   \Pr(A =a_Y \mid C_{j}= Y_{j} = D_{j} = 0,  L_{A_D,j,i}, \bar{L}_{j-1,i}; \hat{\alpha}_{L_D 1}) } \nonumber \\
& \times  \frac{\prod_{j=0}^{k}   \Pr(A =a_Y \mid C_{j}= Y_{j} = D_{j} = 0, \bar{L}_{j-1,i}; \hat{\alpha}_{L_D 2}) }{ \prod_{j=0}^{k}   \Pr(A =a_D \mid C_{j}= Y_{j} = D_{j} = 0, \bar{L}_{j-1,i}; \hat{\alpha}_{L_D 2}) }, \nonumber \\
\hat{W}_{C,k,i} (a_D;\hat{\alpha}_{C}) &= \frac{I(C_{k+1} =0) }{ \prod_{j=0}^{k}  \Pr(C_{j+1}=0 \mid C_{j}=D_{j}= Y_{j}=0, \bar{L}_{j,i},  A = a_D; \hat{\alpha}_{C}) }, 
\end{align*}
where $\Pr(D_{j+1}=0 \mid C_{j+1}=D_{j}= Y_{j}=0, \bar{L}_{j},  A = a_D; \alpha_D)$ is a parametric model for the numerator (and denominator) of $W_{D,k}(a_Y,a_D)$ indexed by parameter $ \alpha_D$, and $\hat{\alpha}_D$ is a consistent estimator of $\alpha_D$ (e.g.\ the MLE). The terms in $\hat{W}_{L_{A_D},k,i} (a_Y,a_D;\hat{\alpha}_{L_D 1},\hat{\alpha}_{L_D 2})$ and $\hat{W}_{C,k,i} (a_D;\hat{\alpha}_{C})$ are defined analogously, where $\hat{\alpha}_{L_D 1},\hat{\alpha}_{L_D 2},\hat{\alpha}_{C}$ are consistent estimators of corresponding model parameters $\alpha_{L_D1 },\alpha_{L_D2 },\alpha_{C}$, respectively. 

Let $\alpha_1 = (\alpha_D, \alpha_{L_D1 },\alpha_{L_D2 },\alpha_{C})$, and define the estimator $\hat{\nu}_{1,a_Y,a_D,k} $ of $\nu_{a_Y,a_D,k} $ as the solution to the estimating equation $\sum_{i=1}^{n}U_{1,k,i}(\nu_{a_Y,a_D,k},\hat{\alpha}_1)=0$ with respect to $\nu_{a_Y,a_D,k}$ with
\begin{align*}
& U_{1,k,i}(\nu_{a_Y,a_D,k},\hat{\alpha}_1) \nonumber \\
= & I(A_i=a_Y) \Big[ \sum_{s=0}^{k} \{ \hat{W}_{1,s,i}(a_Y,a_D;\hat{\alpha}_1) Y_{s+1,i} (1-Y_{s,i}) (1-D_{s+1,i}) \} - \nu_{a_Y,a_D,k} \Big],  \nonumber \\
\end{align*}
and $\hat{W}_{1,s,i}(a_Y,a_D; \hat{\alpha}_1) = \hat{W}_{D,s,i} (a_Y,a_D;\hat{\alpha}_{D}) \hat{W}_{L_{A_D},s,i} (a_Y,a_D;\hat{\alpha}_{L_D 1},\hat{\alpha}_{L_D 2}) \hat{W}_{C,s,i}  (a_Y;\hat{\alpha}_{C}) $. 

Provided that the models indexed by elements in $\alpha_1$ are correctly specified and $\hat{\alpha}_1$ is a consistent estimator for $\alpha_1$, then consistency of $\hat{\nu}_{1,a_Y,a_D, k} $ for $\nu_{a_Y,a_D, k} $ follows because \eqref{eq: identifying formula} and \eqref{eq: alternative id formula 1} are equal. We describe an implementation algorithm for $\hat{\nu}_{1,a_Y,a_D, k} $ in Appendix \ref{sec: estimation algorithm}. In practice, we can use popular regression models for binary outcomes to estimate the weights $W_{D,k}(a_Y,a_D)$ and $W_{C,k}(a_Y)$. However, when we parameterize the terms in $\hat{W}_{L_{A_D},k} (a_Y,a_D;\hat{\alpha}_{L_D 1},\hat{\alpha}_{L_D 2})$, we must ensure that the statistical models are congenial, which may fail for popular models, such as logistic regressions models. In Appendix \ref{sec: proof of alternative id}, we have provided an alternative expression of $W_{L_{A_Y},k}(a_Y,a_D)$ that motivates different weighted estimators based on estimation of the conditional joint densities of $L_k$.  These alternative weighted estimators avoid the problem of incongenial models at the expense of the need to model higher dimensional quantities.   

The estimator based on \eqref{eq: alternative id formula 2} is derived analogously to the estimator based on \eqref{eq: alternative id formula 1}. Suppose 
\begin{align*}
\hat{W}_{Y,k,i}  (a_D,a_Y;\hat{\alpha}_{Y}) & = \frac{ \Pr(Y_{k+1}=1 \mid C_{k+1}=D_{k+1}= Y_{k}=0, \bar{L}_{k,i},  A = a_Y; \hat{\alpha}_{Y}) }{\Pr(Y_{k+1}=1 \mid C_{j+1}=D_{k+1}= Y_{k}=0, \bar{L}_{k,i},  A = a_D; \hat{\alpha}_{Y}) } \\
& \times \frac{\prod_{j=0}^{k-1}  \Pr(Y_{j+1}=0 \mid C_{j+1}=D_{j+1}= Y_{j}=0, \bar{L}_{j,i},  A = a_Y; \hat{\alpha}_{L_Y1}) }{ \prod_{j=0}^{k-1} \Pr(Y_{j+1}=0 \mid C_{j+1}=D_{j+1}= Y_{j}=0, \bar{L}_{j,i},  A = a_D; \hat{\alpha}_{L_Y1}) }, \\
\hat{W}_{L_{A_Y},k,i}  (a_D,a_Y;\hat{\alpha}_{L_Y1},\hat{\alpha}_{L_Y2}) & = \frac{\prod_{j=0}^{k}   \Pr(A =a_Y  \mid C_{j}= Y_{j} = D_{j} = 0, \bar{L}_{j,i}; \hat{\alpha}_{L_Y2}) }{ \prod_{j=0}^{k}   \Pr(A =a_D \mid C_{j}= Y_{j} = D_{j} = 0, \bar{L}_{j,i}; \hat{\alpha}_{L_Y2}) }   \nonumber \\
& \times \frac{\prod_{j=0}^{k}   \Pr(A =a_D \mid C_{j}= Y_{j} = D_{j} = 0,  L_{A_D,j,i}, \bar{L}_{j-1,i}; \hat{\alpha}_{C}) }{ \prod_{j=0}^{k}   \Pr(A =a_Y \mid C_{j}= Y_{j} = D_{j} = 0, L_{A_D,j,i}, \bar{L}_{j,i}; \hat{\alpha}_{C}) },
\end{align*}
where the terms in $\hat{W}_{Y,k,i}  (a_D,a_Y;\hat{\alpha}_{Y})$, $\hat{W}_{L_{A_Y},k,i}  (a_D,a_Y;\hat{\alpha}_{L_Y1},\hat{\alpha}_{L_Y2})$ are statistical models for binary outcomes, and where $\hat{\alpha}_{Y}, \hat{\alpha}_{L_Y 1},\hat{\alpha}_{L_Y 2}$ are consistent estimators for $\alpha_Y,\alpha_{L_Y1 },\alpha_{L_Y2 }$, respectively. Similar to $\hat{W}_{L_{A_D},k,i} (a_Y,a_D;\hat{\alpha}_{L_D 1},\hat{\alpha}_{L_D 2})$, however, we must ensure that congenial models are used to estimate the terms in $\hat{W}_{L_{A_Y},k,i}  (a_D,a_Y;\hat{\alpha}_{L_Y1},\hat{\alpha}_{L_Y2})$. 

Let $\alpha_2 = (\alpha_Y, \alpha_{L_Y1 },\alpha_{L_Y2 },\alpha_{C})$, and define the estimator $\hat{\nu}_{2,a_Y,a_D,k}$ of $\nu_{a_Y,a_D,k} $ as the solution to the estimating equation $\sum_{i=1}^{n}U_{2,k,i}(\nu_{a_Y,a_D,k},\hat{\alpha}_2)=0$ with respect to $\nu_{a_Y,a_D,k}$, where 
\begin{align*}
& U_{2,k,i}(\nu_{a_Y,a_D,k},\hat{\alpha}_2) \nonumber \\
= & I(A_i=a_D) \Big[ \sum_{s=0}^{k} \{ \hat{W}_{2,s,i}(a_Y,a_D;\hat{\alpha}_{2}) Y_{s+1,i} (1-Y_{s,i}) (1-D_{s+1,i}) \} -  \nu_{a_Y,a_D,k} \Big],  \nonumber \\
\end{align*}
and $\hat{W}_{2,s,i}(a_Y,a_D;\hat{\alpha}_{2} ) = \hat{W}_{C,s,i} (a_D;\hat{\alpha}_{C}) \hat{W}_{Y,s,i} (a_D,a_Y;\hat{\alpha}_{Y})\hat{W}_{L_{A_Y},s,i} a_D,a_Y;\hat{\alpha}_{L_Y1},\hat{\alpha}_{L_Y2})$. Analogous to the estimator based on \eqref{eq: alternative id formula 1}, provided that the models indexed by elements in $\alpha_2$ are correctly specified and $\hat{\alpha}_2$ is a consistent estimator for $\alpha_2$, then consistency of $\hat{\nu}_{2,a_Y,a_D,k}$ for $\nu_{a_Y,a_D,k} $ follows because \eqref{eq: identifying formula} and \eqref{eq: alternative id formula 2} are equal. 

\subsection{Simplified estimators under assumptions on $L_k$}
\label{sec: simplified estimators}
Given a user-chosen partition of $L_k$ such that $L_{A_Y,k} \equiv L_k, L_{A_D,k} \equiv \emptyset$ for $k = 0,\dots, K$, then $W_{L_{A_D},k} (a_Y,a_D)=1$ and the consistency of $\hat{\nu}_{1,a_Y,a_D,k}$ only requires consistent estimation of the weights $W_{D,k}(a_Y,a_D)$ and $W_{C,k}(a_Y)$. Similarly, the partition $L_{A_D,k} \equiv L_k, L_{A_Y,k} \equiv \emptyset$ gives $W_{L_{A_Y},k} (a_D,a_Y)=1$, such that the consistency of $\hat{\nu}_{2,a_Y,a_D,k}$ only relies on consistent estimation of the weights $W_{Y,k}(a_Y,a_D)$ and $W_{C,k}(a_D)$. Of course, these simplified $L_k$ partitions are only justified if they satisfy the dismissible component conditions.  As discussed in Section \ref{isodis}, identification under these simplified $L_k$ partitions implies partial or full isolation, impacting the interpretation of the separable effects.

\section{Data example: blood pressure therapy and acute kidney injury}
\label{sec: sprint implementation}
As an illustration, we analyzed data from the Systolic Blood Pressure Intervention Trial (SPRINT) \cite{sprint2015randomized}, which randomly assigned individuals to intensive ($A=1$) or standard ($A=0$) blood pressure treatment. We used follow-up data from each month $k+1$, $k=0\ldots,29$ and restricted our analysis to participants aged older than 75 years at baseline in whom the most deaths (competing events) occurred \cite{williamson2016intensive}.  For simplicity, we further restricted to those patients with complete data on baseline covariates (described below).  This resulted in a data set with 1304 and 1297 in the intensive ($A=1$) and standard ($A=0$) blood pressure therapy arms, respectively. During the 30-month follow-up period, 107 and 98  of these patients were lost to follow-up (censored) in some month $k+1\leq 30$ in the intensive and standard arms, respectively.

In order to adjust for informative censoring by loss to follow-up, we used inverse probability of censoring weighted Aalen-Johansen estimators \cite{aalen1978empirical, young2018choice}  to estimate the total effects of treatment assignment on the cause-specific cumulative incidences at each $k+1$ of kidney injury and mortality. We adjusted for the baseline covariates ($L_0$) smoking status, history of clinical or subclinical cardiovascular disease, clinical of subclinical chronic kidney disease, statin use and gender as well as the time-varying covariates ($L_k$) defined by the most recent measurements of systolic and diastolic blood pressure, scheduled monthly for the first 3 months and every 3 months thereafter. The weight denominators were estimated under the following pooled logistic model for the probability of being censored within each month $k+1$ given the measured past, 
\begin{align}
    \text{logit} \{\Pr(C_{k+1}=1 \mid D_{k}= Y_{k}=\bar{C}_{k}=0,A, \bar{L}_k)\} = \alpha_{C,0,k+1} + \alpha_{C,1}A + \alpha_{C,2}Ak + \alpha'_{C,1}L_0 + \alpha'_{C,2}L_k,
    \label{eq: sprint censoring model}
\end{align}
where $\alpha_{C,0,k+1}$ are time-varying intercepts modeled as 3rd degree polynomials. For all analyses, 95\% percent confidence intervals were constructed using 500 nonparametric bootstrap samples. 

The estimated cumulative incidence of acute kidney injury (the event of interest) under the intensive treatment assignment was consistently higher compared to standard treatment assignment (Figure \ref{fig: ci plot}a, solid lines), in line with a harmful total effect on acute kidney injury. Specifically, the total effect estimate (on the additive scale) of intensive therapy assignment versus standard was 0.01 (95\% CI: $[0.00, 0.03]$) at 30 months of follow-up. As discussed in Section \ref{sec: defining treatment eff}, this harmful effect is hard to interpret due to a possible protective effect of intensive treatment assignment on death (the competing event).  This concern is not easily ruled out by the data; the cumulative incidence of death under intensive treatment assignment is consistently slightly lower compared to standard treatment assignment over the 30-month follow-up with differences increasing at 25 months, as shown with dashed lines in Figure \ref{fig: ci plot}a. At 30 months, the total effect estimate on mortality was -0.01 (95\% CI: $[-0.03, 0.00]$).  

As discussed in Section \ref{sec: Ay iso}, for $A_Y$ defined as the component of treatment $A$ that may exert biological effects on the kidneys, e.g.\ by relaxing the efferent arterioles, and $A_D$ defined as all remaining components of $A$, $A_Y$ partial isolation may be a reasonable assumption given background subject matter knowledge. Under this assumption, the $A_Y$ separable effect \eqref{eq: A_Y sep} evaluated at $a_D=1$ does not capture effects of the treatment on the competing event. It also may be of clinical interest as it quantifies the effect of removing the possibly harmful $A_Y$ component from the original treatment $A$.

We used the inverse probability weighted estimator $\hat{\nu}_{2,a_Y,a_D,k}$ from Section \ref{sec: estimation} to estimate \eqref{eq: estimand} and, in turn, the $A_Y$ separable effect \eqref{eq: A_Y sep} evaluated at $a_D=1$ on acute kidney injury at each $k+1$ under the assumption that the measured baseline and time-varying covariates are sufficient to ensure identification, including the dismissible component conditions \eqref{ass: delta 1}-\eqref{ass: delta 3b} under the partitioning $L_k=L_{A_D,k}$. This assumption, at best, approximately holds because $L_k$ contains only intermittent measurements of systolic and diastolic blood pressure. 

We estimated $W_{Y,k}(a_Y,a_D)$ under the pooled logistic models
\begin{align}
& \text{logit} \{\Pr(Y_{k+1}=1 \mid D_{k+1}= Y_{k+1} =\bar{C}_{k+1}=0,A=0, \bar{L}_k) \} \nonumber \\
 & = \alpha_{Y,0,k+1} + \alpha'_{Y,1}L_0 + \alpha'_{Y,2}L_k + \alpha'_{Y,3}L_k^2 +\alpha'_{Y,4}L_k k, \nonumber \\
& \text{logit} \{ \Pr(Y_{k+1}=1 \mid D_{k+1}= Y_{k+1} =\bar{C}_{k+1}=0,A=1, \bar{L}_k) \} \nonumber \\
 & = \alpha_{Y,5,k+1} + \alpha'_{Y,6}L_0 + \alpha'_{Y,7}L_k + \alpha'_{Y,8}L_k^2 + \alpha'_{Y,9}L_k k, \nonumber \\
\label{eq: weight parametierzations}
\end{align}
for $k = 0,\dots,20$, where $\alpha_{Y,0,k+1}$ and $\alpha_{Y,5,k+1}$ are time-varying intercepts modeled as 3rd degree polynomials. The inverse probability of censoring weights $W_{C}(a_D)$ were estimated under \eqref{eq: sprint censoring model}.

Figure \ref{fig: ci plot}b shows estimates of the counterfactual cumulative incidence for acute kidney injury under assignment to different combinations of $a_Y$ and $a_D$ over time. The vertical distance between the black and the green line at $k+1$ is a point estimate of the additive $A_Y$ separable effect when $a_D=1$ at that time, and similarly the vertical distance between the red and the green line is the $A_D$ separable effect when $a_Y=0$. In particular, the estimated $A_Y$ separable effect of 0.00 (95\% CI: $[-0.10, 0.02]$) at $k+1=30$ months when $a_D=1$, suggests that removing the $A_Y$ component from the intensive therapy will not decrease the average risk of kidney injury by 30 months. R code is provided in the supplementary materials.

\section{Discussion}
\label{sec: discussion}
We have provided generalized results for interpretation and identification of separable effects in competing events settings. These results allow the separable effects to be identified and meaningfullly interpreted in much broader settings than those initially considered by Stensrud et al \cite{stensrud2019separable}. Generally these effects clarify the interpretation of total effects when competing events are affected by treatment, provide more information to patients and doctors for current treatment decisions and inform the development of improved treatments with unwanted components removed.   In general, our framework provides a basis for which subject matter experts can formally reason about the mechanisms by which treatments act on time-to-event outcomes and subsequently falsify this reasoning in a future trial. 

Even under our generalized conditions, the separable effects may be difficult to identify given currently available data in many studies. However, they can point to shortcomings of the data typically collected in studies of competing events, and may guide the planning for improved data collection in future studies.. This is particularly true of randomized trials which have historically relied heavily on the treatment randomization; failing to collect data on baseline and time-varying covariates makes it nearly impossible to adjust for selection bias due to censoring and/or to target estimands other than the total effect of the randomization. Furthermore, we outlined strategies for sensitivity analysis to both the dismissible component conditions and isolation conditions in Appendix \ref{sec: sensitivity analysis}.  

We have focused on establishing fundamental results for interpretation and identification of separable effects, as well as suggesting three estimators that are easy to implement. In future work, we aim to derive new estimators from the efficient influence function, which can achieve parametric convergence rates even when machine learning methods are used for model fitting \cite{robins1994estimation, chernozhukov2018double, robins2017minimax, van2018targeted}, such that e.g.\ bias-aware model selection can be performed to minimize bias due to model misspecification \cite{cui2019bias}. We will also extend our results to separable effects of time-varying treatment interventions in competing events settings, including per-protocol effects in trials with nonadherence. 


\section*{Acknowledgements}
This manuscript was prepared using SPRINT POP Research Materials obtained from the NHLBI Biologic Specimen and Data Reppsitory Information Coordinating Center and does not necessarily reflect the opinions or views of the SPRINT POP or the NHLBI. \\

\bibliography{references}
\bibliographystyle{unsrt}

\clearpage

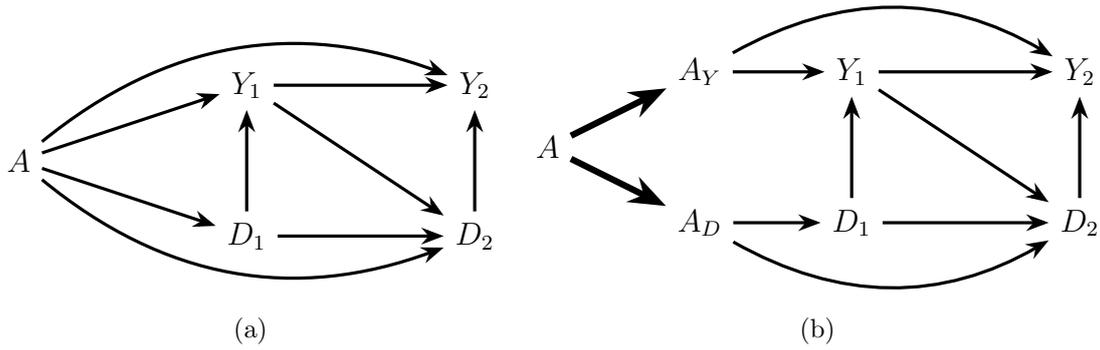
\begin{figure}
\subfloat[]{
\begin{tikzpicture}
\begin{scope}[every node/.style={thick,draw=none}]
    \node (A) at (0,0) {$A$};
  	\node (Y1) at (3,1) {$Y_1$};
    \node (D1) at (3,-1) {$D_1$};
    \node (Y2) at (6,1) {$Y_2$};
    \node (D2) at (6,-1) {$D_2$};
\end{scope}
\begin{scope}[>={Stealth[black]},
              every node/.style={fill=white,circle},
              every edge/.style={draw=black,very thick}]
	\path [->] (A) edge  (D1);
    \path [->] (A) edge[bend right] (D2);
	\path [->] (A) edge[bend left] (Y2);
    \path [->] (A) edge  (Y1);
    \path [->] (Y1) edge (D2);
    \path [->] (Y1) edge (Y2);
    \path [->] (D1) edge (D2);
    \path [->] (D1) edge (Y1);
    \path [->] (D2) edge (Y2);
\end{scope}
\end{tikzpicture}
}
\subfloat[]{
\begin{tikzpicture}
\begin{scope}[every node/.style={thick,draw=none}]
    \node (A) at (-1,0) {$A$};
    \node (Ay) at (1,1) {$A_Y$};
	\node (Ad) at (1,-1) {$A_D$};
	\node (Y1) at (3,1) {$Y_1$};
    \node (D1) at (3,-1) {$D_1$};
    \node (Y2) at (6,1) {$Y_2$};
    \node (D2) at (6,-1) {$D_2$};
\end{scope}

\begin{scope}[>={Stealth[black]},
              every node/.style={fill=white,circle},
              every edge/.style={draw=black,very thick}]
    \path [->] (A) edge[line width=0.85mm] (Ad);
    \path [->] (A) edge[line width=0.85mm] (Ay);
	\path [->] (Ad) edge (D1);
    \path [->] (Ad) edge[bend right] (D2);
	\path [->] (Ay) edge[bend left] (Y2);
    \path [->] (Ay) edge (Y1);	
    \path [->] (Y1) edge (D2);
    \path [->] (Y1) edge (Y2);
    \path [->] (D1) edge (D2);
    \path [->] (D1) edge (Y1);
    \path [->] (D2) edge (Y2);
\end{scope}
\end{tikzpicture}
}
\caption{The directed acyclic graph (DAG) in (a) represents a restrictive data generating assumption on the observed data structure such that there are no common causes of the event of interest and the competing event at any time. The extended DAG in (b) is an augmented version of the graph in (a) representing a treatment decomposition satisfying the generalized decomposition assumption.  The bold arrows encode deterministic relationships.}
    \label{fig: simplest graphs}
\end{figure}

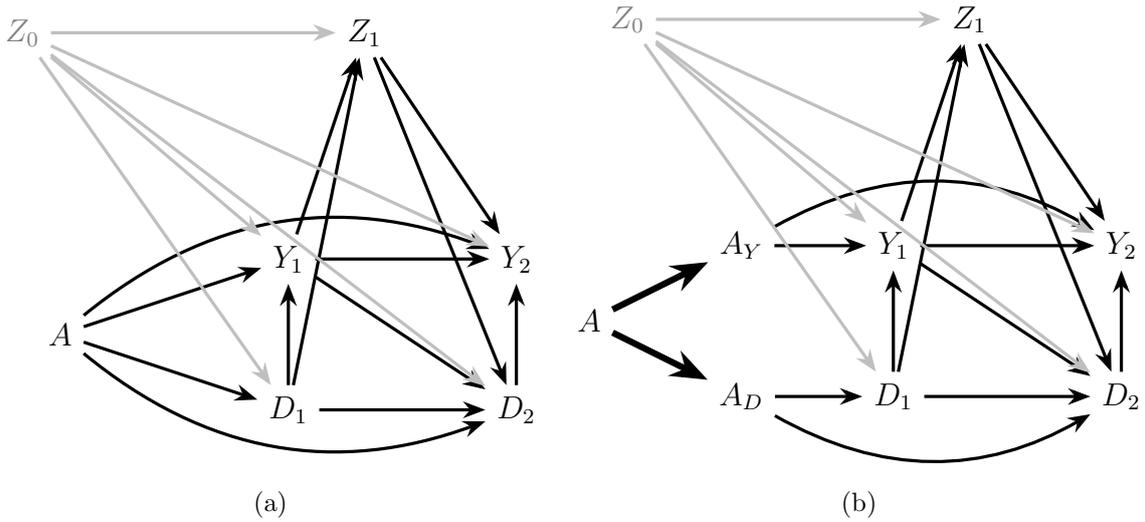
\begin{figure}
\centering
\setlength{\lineskip}{1ex}
\subfloat[]{
\begin{tikzpicture}
\begin{scope}[every node/.style={thick,draw=none}]
    \node (A) at (0,0) {$A$};
  	\node (Y1) at (3,1) {$Y_1$};
    \node (D1) at (3,-1) {$D_1$};
    \node (Y2) at (6,1) {$Y_2$};
    \node (D2) at (6,-1) {$D_2$};
    \node (Z1) at (4,4) {$Z_1$};
    \node[text = gray] (Z0) at (-0.5,4) {$Z_0$};
\end{scope}

\begin{scope}[>={Stealth[black]},
              every node/.style={fill=white,circle},
              every edge/.style={draw=black,very thick}]
	\path [->] (A) edge  (D1);
    \path [->] (A) edge[bend right] (D2);
	\path [->] (A) edge[bend left] (Y2);
    \path [->] (A) edge  (Y1);	
    \path [->] (Z1) edge (D2);
    \path [->] (Z1) edge (Y2);
    \path [->] (Y1) edge (D2);
    \path [->] (Y1) edge (Y2);
    \path [->] (D1) edge (D2);
    \path [->] (D1) edge (Y1);
    \path [->] (D2) edge (Y2);
    \path [->] (Y1) edge (Z1);
    \path [->] (D1) edge (Z1);
    \path [->,>={Stealth[lightgray]}] (Z0) edge[lightgray] (Y1);
    \path [->,>={Stealth[lightgray]}] (Z0) edge[lightgray] (Y2);
    \path [->,>={Stealth[lightgray]}] (Z0) edge[lightgray] (D1);
    \path [->,>={Stealth[lightgray]}] (Z0) edge[lightgray] (D2);
    \path [->,>={Stealth[lightgray]}] (Z0) edge[lightgray] (Z1);
\end{scope}
\end{tikzpicture}
}
\subfloat[]{
\begin{tikzpicture}
\begin{scope}[every node/.style={thick,draw=none}]
    \node (A) at (-1,0) {$A$};
    \node (Ay) at (1,1) {$A_Y$};
	\node (Ad) at (1,-1) {$A_D$};
	\node (Y1) at (3,1) {$Y_1$};
    \node (D1) at (3,-1) {$D_1$};
    \node (Y2) at (6,1) {$Y_2$};
    \node (D2) at (6,-1) {$D_2$};
    \node (Z1) at (4,4) {$Z_1$};
    \node[text = gray] (Z0) at (-0.5,4) {$Z_0$};
\end{scope}

\begin{scope}[>={Stealth[black]},
              every node/.style={fill=white,circle},
              every edge/.style={draw=black,very thick}]
    \path [->] (A) edge[line width=0.85mm] (Ad);
    \path [->] (A) edge[line width=0.85mm] (Ay);
	\path [->] (Ad) edge (D1);
    \path [->] (Ad) edge[bend right] (D2);
	\path [->] (Ay) edge[bend left] (Y2);
    \path [->] (Ay) edge (Y1);	
    \path [->] (Z1) edge (D2);
    \path [->] (Z1) edge (Y2);
    \path [->] (Y1) edge (D2);
    \path [->] (Y1) edge (Y2);
    \path [->] (D1) edge (D2);
    \path [->] (D1) edge (Y1);
    \path [->] (D2) edge (Y2);
    \path [->] (Y1) edge (Z1);
    \path [->] (D1) edge (Z1);
    \path [->,>={Stealth[lightgray]}] (Z0) edge[lightgray] (Y1);
    \path [->,>={Stealth[lightgray]}] (Z0) edge[lightgray] (Y2);
    \path [->,>={Stealth[lightgray]}] (Z0) edge[lightgray] (D1);
    \path [->,>={Stealth[lightgray]}] (Z0) edge[lightgray] (D2);
    \path [->,>={Stealth[lightgray]}] (Z0) edge[lightgray] (Z1);
\end{scope}
\end{tikzpicture}
}
\caption{The causal DAG in (a) allows a pre-randomization common cause ($Z_0$) of $\underline{Y}_1$ and $\underline{D}_1$ and post-randomization common cause ($Z_1$) of $Y_2$ and $D_2$ but assumes $Z_1$ is not affected by treatment $A$. (b) is an extension of (a) satisfying full isolation.}
\label{fig:full isolation}
\end{figure}

\begin{figure}
\setlength{\lineskip}{1ex}
\subfloat{
\begin{tikzpicture}
\begin{scope}[every node/.style={thick,draw=none}]
    \node (A) at (0,0) {$A$};
  	\node (Y1) at (3,1) {$Y_1$};
    \node (D1) at (3,-1) {$D_1$};
    \node (Y2) at (6,1) {$Y_2$};
    \node (D2) at (6,-1) {$D_2$};
    \node (Z1) at (4,4) {$Z_1$};
    \node[text = gray] (Z0) at (-0.5,4) {$Z_0$};
\end{scope}

\begin{scope}[>={Stealth[black]},
              every node/.style={fill=white,circle},
              every edge/.style={draw=black,very thick}]
	\path [->] (A) edge  (D1);
    \path [->] (A) edge[bend right] (D2);
	\path [->] (A) edge[bend left] (Y2);
    \path [->] (A) edge  (Y1);	
    \path [->] (Z1) edge (D2);
    \path [->] (Z1) edge (Y2);
    \path [->] (Y1) edge (D2);
    \path [->] (Y1) edge (Y2);
    \path [->] (D1) edge (D2);
    \path [->] (D1) edge (Y1);
    \path [->] (D2) edge (Y2);
    \path [->] (Y1) edge (Z1);
    \path [->] (D1) edge (Z1);
    \path [->,>={Stealth[blue]}] (A) edge[blue] (Z1);
    \path [->,>={Stealth[lightgray]}] (Z0) edge[lightgray] (Y1);
    \path [->,>={Stealth[lightgray]}] (Z0) edge[lightgray] (Y2);
    \path [->,>={Stealth[lightgray]}] (Z0) edge[lightgray] (D1);
    \path [->,>={Stealth[lightgray]}] (Z0) edge[lightgray] (D2);
    \path [->,>={Stealth[lightgray]}] (Z0) edge[lightgray] (Z1);
\end{scope}
\end{tikzpicture}
}

\caption{A causal DAG representing the assumption that $Z_1$, a common cause of $Y_2$ and $D_2$, may be affected by treatment $A$ (blue arrow).} 
\label{fig:dagAarrowLk}
\end{figure}
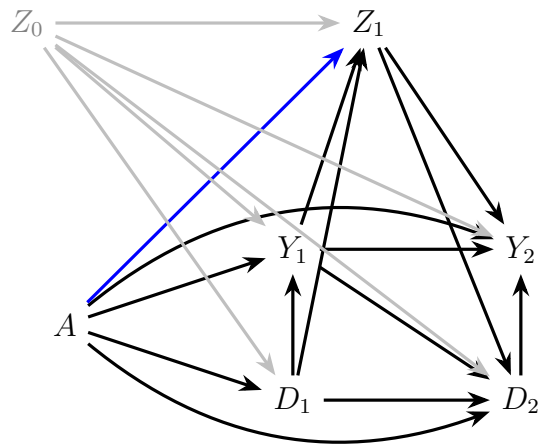

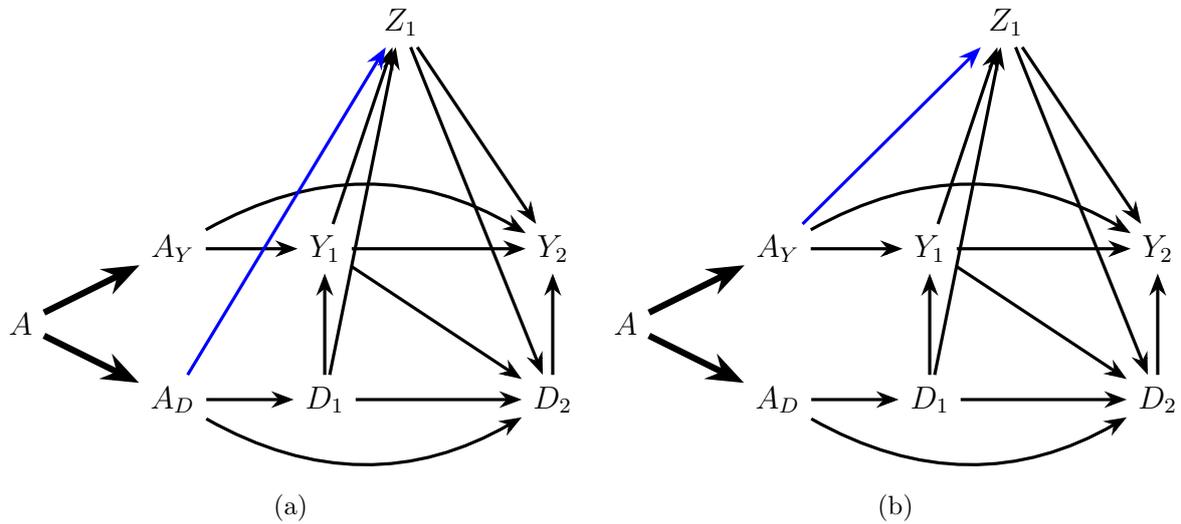
\begin{figure}
\setlength{\lineskip}{1ex}
\subfloat[]{
\begin{tikzpicture}
\begin{scope}[every node/.style={thick,draw=none}]
    \node (A) at (-1,0) {$A$};
    \node (Ay) at (1,1) {$A_Y$};
	\node (Ad) at (1,-1) {$A_D$};
	\node (Y1) at (3,1) {$Y_1$};
    \node (D1) at (3,-1) {$D_1$};
    \node (Y2) at (6,1) {$Y_2$};
    \node (D2) at (6,-1) {$D_2$};
    \node (Z1) at (4,4) {$Z_1$};
\end{scope}

\begin{scope}[>={Stealth[black]},
              every node/.style={fill=white,circle},
              every edge/.style={draw=black,very thick}]
    \path [->] (A) edge[line width=0.85mm] (Ad);
    \path [->] (A) edge[line width=0.85mm] (Ay);
	\path [->] (Ad) edge (D1);
    \path [->] (Ad) edge[bend right] (D2);
	\path [->] (Ay) edge[bend left] (Y2);
    \path [->] (Ay) edge (Y1);	
    \path [->] (Z1) edge (D2);
    \path [->] (Z1) edge (Y2);
    \path [->] (Y1) edge (D2);
    \path [->] (Y1) edge (Y2);
    \path [->] (D1) edge (D2);
    \path [->] (D1) edge (Y1);
    \path [->] (D2) edge (Y2);
    \path [->] (Y1) edge (Z1);
    \path [->] (D1) edge (Z1);
    \path [->,>={Stealth[blue]}] (Ad) edge[blue] (Z1);
\end{scope}
\end{tikzpicture}
}  
\subfloat[]{
\begin{tikzpicture}
\begin{scope}[every node/.style={thick,draw=none}]
    \node (A) at (-1,0) {$A$};
    \node (Ay) at (1,1) {$A_Y$};
	\node (Ad) at (1,-1) {$A_D$};
	\node (Y1) at (3,1) {$Y_1$};
    \node (D1) at (3,-1) {$D_1$};
    \node (Y2) at (6,1) {$Y_2$};
    \node (D2) at (6,-1) {$D_2$};
    \node (Z1) at (4,4) {$Z_1$};
\end{scope}

\begin{scope}[>={Stealth[black]},
              every node/.style={fill=white,circle},
              every edge/.style={draw=black,very thick}]
    \path [->] (A) edge[line width=0.85mm] (Ad);
    \path [->] (A) edge[line width=0.85mm] (Ay);
	\path [->] (Ad) edge (D1);
    \path [->] (Ad) edge[bend right] (D2);
	\path [->] (Ay) edge[bend left] (Y2);
    \path [->] (Ay) edge (Y1);	
    \path [->] (Z1) edge (D2);
    \path [->] (Z1) edge (Y2);
    \path [->] (Y1) edge (D2);
    \path [->] (Y1) edge (Y2);
    \path [->] (D1) edge (D2);
    \path [->] (D1) edge (Y1);
    \path [->] (D2) edge (Y2);
    \path [->] (Y1) edge (Z1);
    \path [->] (D1) edge (Z1);
    \path [->,>={Stealth[blue]}] (Ay) edge[blue] (Z1);
\end{scope}
\end{tikzpicture}
}
\caption{Extensions of the causal DAG in Figure \ref{fig:dagAarrowLk}  illustrating partial isolation. The blue arrow in (a) represents the $A \rightarrow Z_1$ relation in Figure \ref{fig:dagAarrowLk} under $A_Y$ partial isolation, and the blue arrow in (b) represents $A_D$ partial isolation.}
\label{fig:partial isolation}
\end{figure}

\begin{figure}
\centering
\subfloat[]{
\begin{tikzpicture}
\begin{scope}[every node/.style={thick,draw=none}]
    \node (A) at (-1,0) {$A$};
    \node (Ay) at (1,1) {$A_Y$};
	\node (Ad) at (1,-1) {$A_D$};
	\node (Y1) at (3,1) {$Y_1$};
    \node (D1) at (3,-1) {$D_1$};
    \node (Y2) at (6,1) {$Y_2$};
    \node (D2) at (6,-1) {$D_2$};
    \node (Z1) at (4,4) {$Z_1$};
    \node (EMPTY_NODE) at (4,-4) {$ $}; 
\end{scope}

\begin{scope}[>={Stealth[black]},
              every node/.style={fill=white,circle},
              every edge/.style={draw=black,very thick}]
    \path [->] (A) edge[line width=0.85mm] (Ad);
    \path [->] (A) edge[line width=0.85mm] (Ay);
	\path [->] (Ad) edge (D1);
    \path [->] (Ad) edge[bend right] (D2);
	\path [->] (Ay) edge[bend left] (Y2);
    \path [->] (Ay) edge (Y1);	
    \path [->] (Z1) edge (D2);
    \path [->] (Z1) edge (Y2);
    \path [->] (Y1) edge (D2);
    \path [->] (Y1) edge (Y2);
    \path [->] (D1) edge (D2);
    \path [->] (D1) edge (Y1);
    \path [->] (D2) edge (Y2);
    \path [->] (Y1) edge (Z1);
    \path [->] (D1) edge (Z1);
    \path [->,>={Stealth[blue]}] (Ay) edge[blue] (Z1);
    \path [->,>={Stealth[blue]}] (Ad) edge[blue] (Z1);
\end{scope}
\end{tikzpicture}
}\subfloat[]{
\begin{tikzpicture}
\begin{scope}[every node/.style={thick,draw=none}]
    \node (A) at (-1,0) {$A$};
    \node (Ay) at (1,1) {$A_Y$};
	\node (Ad) at (1,-1) {$A_D$};
	\node (Y1) at (3,1) {$Y_1$};
    \node (D1) at (3,-1) {$D_1$};
    \node (Y2) at (6,1) {$Y_2$};
    \node (D2) at (6,-1) {$D_2$};
    \node (Z1y) at (4,4) {$Z_{A_Y,1}$};
    \node (Z1d) at (4,-4) {$Z_{A_D,1}$};
\end{scope}

\begin{scope}[>={Stealth[black]},
              every node/.style={fill=white,circle},
              every edge/.style={draw=black,very thick}]
    \path [->] (A) edge[line width=0.85mm] (Ad);
    \path [->] (A) edge[line width=0.85mm] (Ay);
	\path [->] (Ad) edge (D1);
    \path [->] (Ad) edge[bend right] (D2);
	\path [->] (Ay) edge[bend left] (Y2);
    \path [->] (Ay) edge (Y1);	
    \path [->] (Z1y) edge (D2);
    \path [->] (Z1y) edge (Y2);
    \path [->] (Z1d) edge (D2);
    \path [->] (Z1d) edge (Y2);
    \path [->] (Y1) edge (D2);
    \path [->] (Y1) edge (Y2);
    \path [->] (D1) edge (D2);
    \path [->] (D1) edge (Y1);
    \path [->] (D2) edge (Y2);
    \path [->] (Y1) edge (Z1y);
    \path [->] (D1) edge (Z1y);
    \path [->] (Y1) edge (Z1d);
    \path [->] (D1) edge (Z1d);
    \path [->,>={Stealth[blue]}] (Ay) edge[blue] (Z1y);
    \path [->,>={Stealth[blue]}] (Ad) edge[blue] (Z1d);
\end{scope}
\end{tikzpicture}
} \\ 
\caption{Causal graphs illustrating no isolation. (a) violates $Z_k$ partition while (b) satisfies $Z_k$ partition.}
\label{fig:no isolation}
\end{figure}
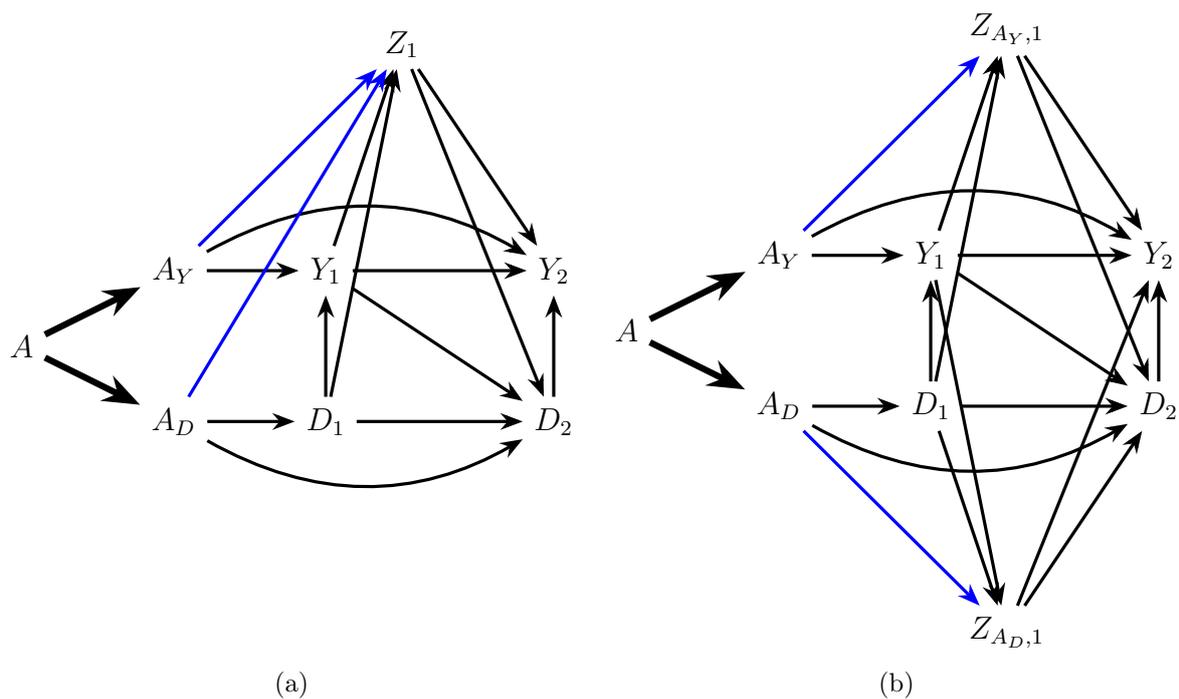

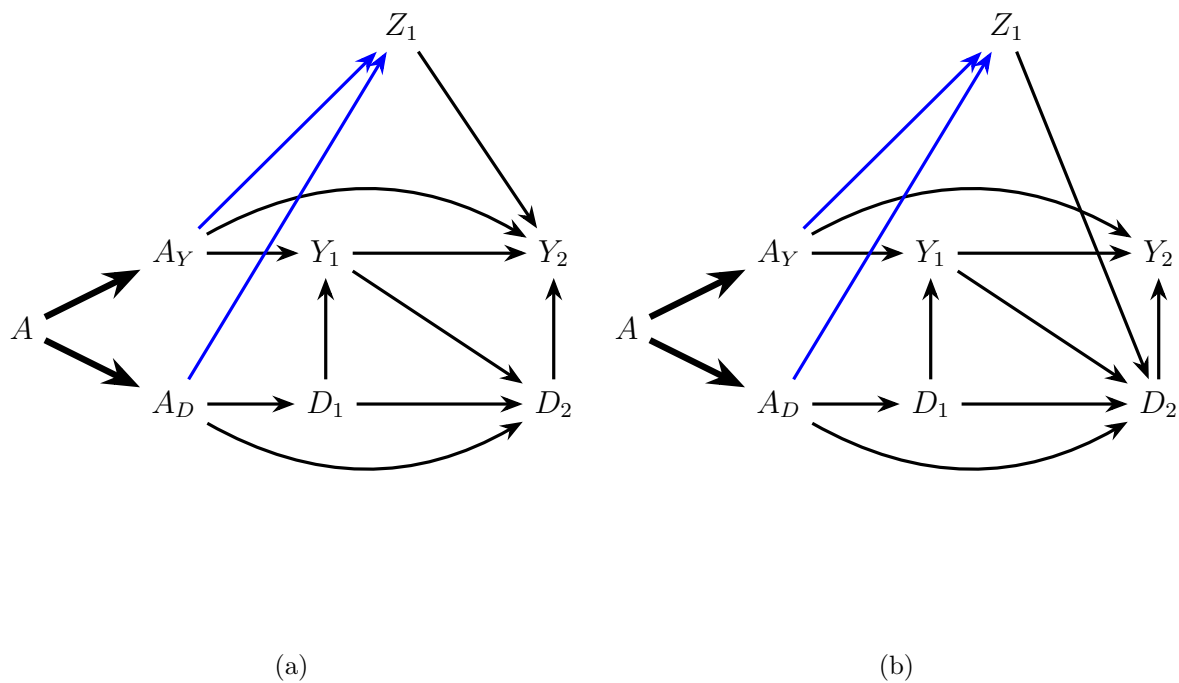
\begin{figure}
    \centering
    \subfloat[]{
\begin{tikzpicture}
\begin{scope}[every node/.style={thick,draw=none}]
    \node (A) at (-1,0) {$A$};
    \node (Ay) at (1,1) {$A_Y$};
	\node (Ad) at (1,-1) {$A_D$};
	\node (Y1) at (3,1) {$Y_1$};
    \node (D1) at (3,-1) {$D_1$};
    \node (Y2) at (6,1) {$Y_2$};
    \node (D2) at (6,-1) {$D_2$};
    \node (Z1) at (4,4) {$Z_1$};
    \node (EMPTY_NODE) at (4,-4) {$ $}; 
\end{scope}

\begin{scope}[>={Stealth[black]},
              every node/.style={fill=white,circle},
              every edge/.style={draw=black,very thick}]
    \path [->] (A) edge[line width=0.85mm] (Ad);
    \path [->] (A) edge[line width=0.85mm] (Ay);
	\path [->] (Ad) edge (D1);
    \path [->] (Ad) edge[bend right] (D2);
	\path [->] (Ay) edge[bend left] (Y2);
    \path [->] (Ay) edge (Y1);	
    \path [->] (Z1) edge (Y2);
    \path [->] (Y1) edge (D2);
    \path [->] (Y1) edge (Y2);
    \path [->] (D1) edge (D2);
    \path [->] (D1) edge (Y1);
    \path [->] (D2) edge (Y2);
    \path [->,>={Stealth[blue]}] (Ay) edge[blue] (Z1);
    \path [->,>={Stealth[blue]}] (Ad) edge[blue] (Z1);
\end{scope}
\end{tikzpicture}
}
  \subfloat[]{
\begin{tikzpicture}
\begin{scope}[every node/.style={thick,draw=none}]
    \node (A) at (-1,0) {$A$};
    \node (Ay) at (1,1) {$A_Y$};
	\node (Ad) at (1,-1) {$A_D$};
	\node (Y1) at (3,1) {$Y_1$};
    \node (D1) at (3,-1) {$D_1$};
    \node (Y2) at (6,1) {$Y_2$};
    \node (D2) at (6,-1) {$D_2$};
    \node (Z1) at (4,4) {$Z_1$};
    \node (EMPTY_NODE) at (4,-4) {$ $}; 
\end{scope}

\begin{scope}[>={Stealth[black]},
              every node/.style={fill=white,circle},
              every edge/.style={draw=black,very thick}]
    \path [->] (A) edge[line width=0.85mm] (Ad);
    \path [->] (A) edge[line width=0.85mm] (Ay);
	\path [->] (Ad) edge (D1);
    \path [->] (Ad) edge[bend right] (D2);
	\path [->] (Ay) edge[bend left] (Y2);
    \path [->] (Ay) edge (Y1);	
    \path [->] (Z1) edge (D2);
    \path [->] (Y1) edge (D2);
    \path [->] (Y1) edge (Y2);
    \path [->] (D1) edge (D2);
    \path [->] (D1) edge (Y1);
    \path [->] (D2) edge (Y2);
    \path [->,>={Stealth[blue]}] (Ay) edge[blue] (Z1);
    \path [->,>={Stealth[blue]}] (Ad) edge[blue] (Z1);
\end{scope}
\end{tikzpicture}
}
\caption{Causal graphs illustrating partial isolation but violation of $Z_k$ partition.  $A_Y$ partial isolation holds in (a) and $A_D$ partial isolation holds in (b).}
    \label{fig: Z_k vs partial iso}
\end{figure}

\clearpage

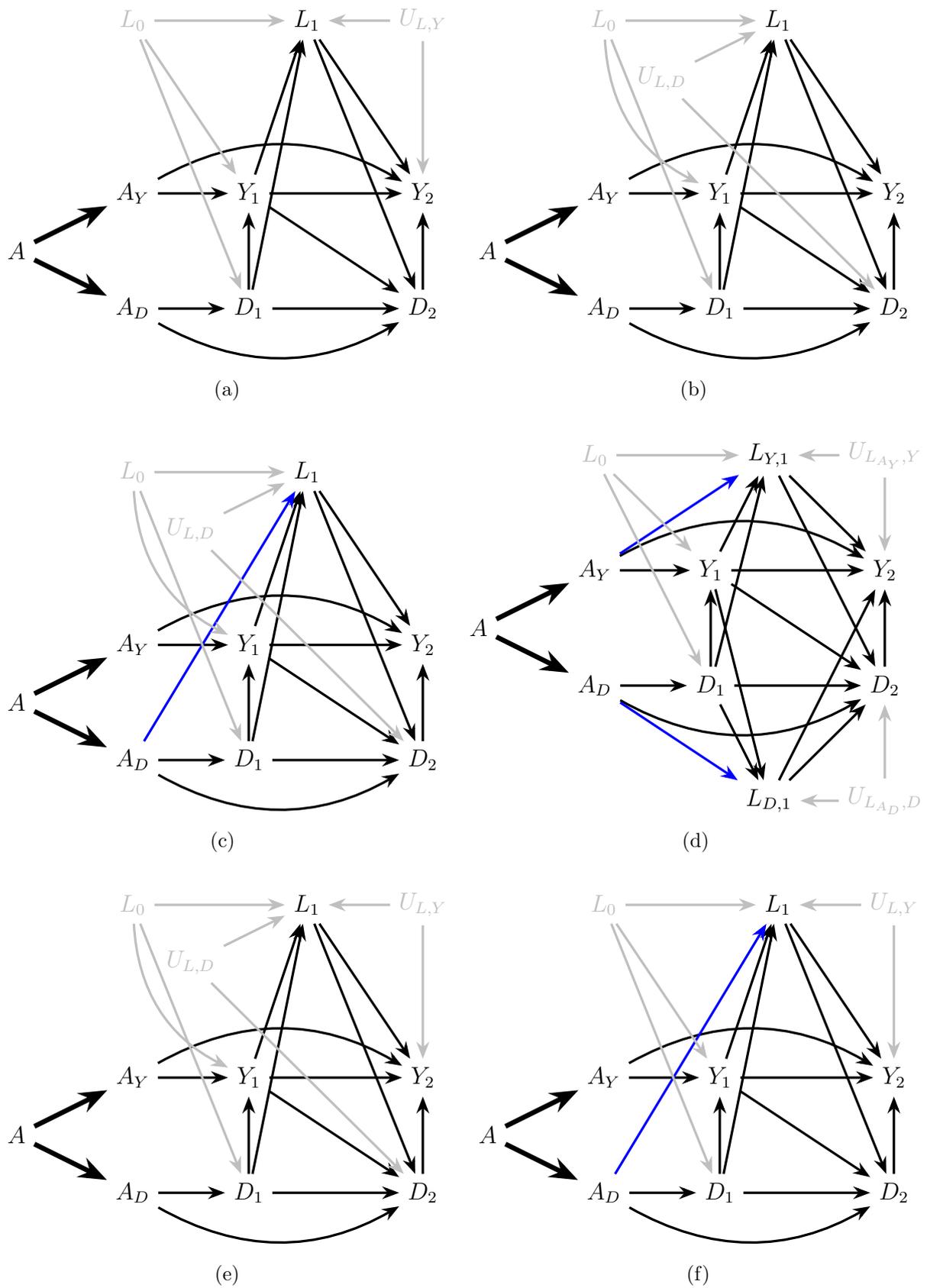
\begin{figure}
\setlength{\lineskip}{1ex}
\subfloat[]{
\begin{tikzpicture}
\begin{scope}[every node/.style={thick,draw=none}]
    \node (A) at (-1,0) {$A$};
    \node (Ay) at (1,1) {$A_Y$};
	\node (Ad) at (1,-1) {$A_D$};
	\node (Y1) at (3,1) {$Y_1$};
    \node (D1) at (3,-1) {$D_1$};
    \node (Y2) at (6,1) {$Y_2$};
    \node (D2) at (6,-1) {$D_2$};
    \node (L2) at (4,4) {$L_1$};
    \node[lightgray] (L0) at (1,4) {$L_0$};
    \node[lightgray] (ULY) at (6,4) {$U_{L,Y}$};
\end{scope}

\begin{scope}[>={Stealth[black]},
              every node/.style={fill=white,circle},
              every edge/.style={draw=black,very thick}]
    \path [->] (A) edge[line width=0.85mm] (Ad);
    \path [->] (A) edge[line width=0.85mm] (Ay);
	\path [->] (Ad) edge (D1);
    \path [->] (Ad) edge[bend right] (D2);
	\path [->] (Ay) edge[bend left] (Y2);
    \path [->] (Ay) edge (Y1);	
    \path [->] (L2) edge (D2);
    \path [->] (L2) edge (Y2);
    \path [->] (Y1) edge (D2);
    \path [->] (Y1) edge (Y2);
    \path [->] (D1) edge (D2);
    \path [->] (D1) edge (Y1);
    \path [->] (D2) edge (Y2);
    \path [->] (Y1) edge (L2);
    \path [->] (D1) edge (L2);
    \path [->,>={Stealth[lightgray]}] (ULY) edge[lightgray] (L2);
    \path [->,>={Stealth[lightgray]}] (ULY) edge[lightgray] (Y2);
    \path [->,>={Stealth[lightgray]}] (L0) edge[lightgray] (Y1);
    \path [->,>={Stealth[lightgray]}] (L0) edge[lightgray] (D1);
    \path [->,>={Stealth[lightgray]}] (L0) edge[lightgray] (L2);
\end{scope}
\end{tikzpicture}
}
\subfloat[]{
\begin{tikzpicture}
\begin{scope}[every node/.style={thick,draw=none}]
    \node (A) at (-1,0) {$A$};
    \node (Ay) at (1,1) {$A_Y$};
	\node (Ad) at (1,-1) {$A_D$};
	\node (Y1) at (3,1) {$Y_1$};
    \node (D1) at (3,-1) {$D_1$};
    \node (Y2) at (6,1) {$Y_2$};
    \node (D2) at (6,-1) {$D_2$};
    \node (L2) at (4,4) {$L_1$};
    \node[lightgray] (L0) at (1,4) {$L_0$};
    \node[lightgray] (ULD) at (2,3) {$U_{L,D}$};
\end{scope}

\begin{scope}[>={Stealth[black]},
              every node/.style={fill=white,circle},
              every edge/.style={draw=black,very thick}]
    \path [->] (A) edge[line width=0.85mm] (Ad);
    \path [->] (A) edge[line width=0.85mm] (Ay);
	\path [->] (Ad) edge (D1);
    \path [->] (Ad) edge[bend right] (D2);
	\path [->] (Ay) edge[bend left] (Y2);
    \path [->] (Ay) edge (Y1);	
    \path [->] (L2) edge (D2);
    \path [->] (L2) edge (Y2);
    \path [->] (Y1) edge (D2);
    \path [->] (Y1) edge (Y2);
    \path [->] (D1) edge (D2);
    \path [->] (D1) edge (Y1);
    \path [->] (D2) edge (Y2);
    \path [->] (Y1) edge (L2);
    \path [->] (D1) edge (L2);
    \path [->,>={Stealth[lightgray]}] (ULD) edge[lightgray] (L2);
    \path [->,>={Stealth[lightgray]}] (ULD) edge[lightgray] (D2);
    \path [->,>={Stealth[lightgray]}] (L0) edge[lightgray] (D1);
    \path [->,>={Stealth[lightgray]}] (L0) edge[lightgray] (L2);
    \path [->,>={Stealth[lightgray]}] (L0) edge[lightgray, bend right] (Y1);
\end{scope}
\end{tikzpicture}
} \\
\subfloat[]{
\begin{tikzpicture}
\begin{scope}[every node/.style={thick,draw=none}]
    \node (A) at (-1,0) {$A$};
    \node (Ay) at (1,1) {$A_Y$};
	\node (Ad) at (1,-1) {$A_D$};
	\node (Y1) at (3,1) {$Y_1$};
    \node (D1) at (3,-1) {$D_1$};
    \node (Y2) at (6,1) {$Y_2$};
    \node (D2) at (6,-1) {$D_2$};
    \node (L2) at (4,4) {$L_1$};
    \node[lightgray] (L0) at (1,4) {$L_0$};
    \node[lightgray] (ULD) at (2,3) {$U_{L,D}$};
\end{scope}

\begin{scope}[>={Stealth[black]},
              every node/.style={fill=white,circle},
              every edge/.style={draw=black,very thick}]
    \path [->] (A) edge[line width=0.85mm] (Ad);
    \path [->] (A) edge[line width=0.85mm] (Ay);
	\path [->] (Ad) edge (D1);
    \path [->] (Ad) edge[bend right] (D2);
	\path [->] (Ay) edge[bend left] (Y2);
    \path [->] (Ay) edge (Y1);	
    \path [->] (L2) edge (D2);
    \path [->] (L2) edge (Y2);
    \path [->] (Y1) edge (D2);
    \path [->] (Y1) edge (Y2);
    \path [->] (D1) edge (D2);
    \path [->] (D1) edge (Y1);
    \path [->] (D2) edge (Y2);
    \path [->] (Y1) edge (L2);
    \path [->] (D1) edge (L2);
    \path [->,>={Stealth[blue]}] (Ad) edge[blue] (L2);
    \path [->,>={Stealth[lightgray]}] (ULD) edge[lightgray] (L2);
    \path [->,>={Stealth[lightgray]}] (ULD) edge[lightgray] (D2);
    \path [->,>={Stealth[lightgray]}] (L0) edge[lightgray] (D1);
    \path [->,>={Stealth[lightgray]}] (L0) edge[lightgray] (L2);
    \path [->,>={Stealth[lightgray]}] (L0) edge[lightgray, bend right] (Y1);
\end{scope}
\end{tikzpicture}
}  
\subfloat[]{
\begin{tikzpicture}
\begin{scope}[every node/.style={thick,draw=none}]
    \node (A) at (-1,0) {$A$};
    \node (Ay) at (1,1) {$A_Y$};
	\node (Ad) at (1,-1) {$A_D$};
	\node (Y1) at (3,1) {$Y_1$};
    \node (D1) at (3,-1) {$D_1$};
    \node (Y2) at (6,1) {$Y_2$};
    \node (D2) at (6,-1) {$D_2$};
    \node (L2y) at (4,3) {$L_{Y,1}$};
    \node (L2d) at (4,-3) {$L_{D,1}$};
    \node[lightgray] (ULY) at (6,3) {$U_{L_{A_Y},Y}$};
    \node[lightgray] (ULD) at (6,-3) {$U_{L_{A_D},D}$};
    \node[lightgray] (L0) at (1,3) {$L_0$};
\end{scope}
\begin{scope}[>={Stealth[black]},
              every node/.style={fill=white,circle},
              every edge/.style={draw=black,very thick}]
    \path [->] (A) edge[line width=0.85mm] (Ad);
    \path [->] (A) edge[line width=0.85mm] (Ay);
	\path [->] (Ad) edge (D1);
    \path [->] (Ad) edge[bend right] (D2);
	\path [->] (Ay) edge[bend left] (Y2);
    \path [->] (Ay) edge (Y1);	
    \path [->] (L2y) edge (D2);
    \path [->] (L2y) edge (Y2);
    \path [->] (L2d) edge (D2);
    \path [->] (L2d) edge (Y2);
    \path [->] (Y1) edge (D2);
    \path [->] (Y1) edge (Y2);
    \path [->] (D1) edge (D2);
    \path [->] (D1) edge (Y1);
    \path [->] (D2) edge (Y2);
    \path [->] (Y1) edge (L2y);
    \path [->] (D1) edge (L2y);
    \path [->] (Y1) edge (L2d);
    \path [->] (D1) edge (L2d);
    \path [->,>={Stealth[blue]}] (Ay) edge[blue] (L2y);
    \path [->,>={Stealth[blue]}] (Ad) edge[blue] (L2d);
    \path [->,>={Stealth[lightgray]}] (ULY) edge[lightgray] (L2y);
    \path [->,>={Stealth[lightgray]}] (ULY) edge[lightgray] (Y2);
    \path [->,>={Stealth[lightgray]}] (ULD) edge[lightgray] (L2d);
    \path [->,>={Stealth[lightgray]}] (ULD) edge[lightgray] (D2);
    \path [->,>={Stealth[lightgray]}] (L0) edge[lightgray] (D1);
    \path [->,>={Stealth[lightgray]}] (L0) edge[lightgray] (L2y);
    \path [->,>={Stealth[lightgray]}] (L0) edge[lightgray] (Y1);
\end{scope}
\end{tikzpicture}
}
\\
\subfloat[]{
\begin{tikzpicture}
\begin{scope}[every node/.style={thick,draw=none}]
    \node (A) at (-1,0) {$A$};
    \node (Ay) at (1,1) {$A_Y$};
	\node (Ad) at (1,-1) {$A_D$};
	\node (Y1) at (3,1) {$Y_1$};
    \node (D1) at (3,-1) {$D_1$};
    \node (Y2) at (6,1) {$Y_2$};
    \node (D2) at (6,-1) {$D_2$};
    \node (L2) at (4,4) {$L_1$};
    \node[lightgray] (ULY) at (6,4) {$U_{L,Y}$};
    \node[lightgray] (ULD) at (2,3) {$U_{L,D}$};        \node[lightgray] (L0) at (1,4) {$L_0$};
\end{scope}

\begin{scope}[>={Stealth[black]},
              every node/.style={fill=white,circle},
              every edge/.style={draw=black,very thick}]
    \path [->] (A) edge[line width=0.85mm] (Ad);
    \path [->] (A) edge[line width=0.85mm] (Ay);
	\path [->] (Ad) edge (D1);
    \path [->] (Ad) edge[bend right] (D2);
	\path [->] (Ay) edge[bend left] (Y2);
    \path [->] (Ay) edge (Y1);	
    \path [->] (L2) edge (D2);
    \path [->] (L2) edge (Y2);
    \path [->] (Y1) edge (D2);
    \path [->] (Y1) edge (Y2);
    \path [->] (D1) edge (D2);
    \path [->] (D1) edge (Y1);
    \path [->] (D2) edge (Y2);
    \path [->] (Y1) edge (L2);
    \path [->] (D1) edge (L2);
    \path [->,>={Stealth[lightgray]}] (ULY) edge[lightgray] (L2);
    \path [->,>={Stealth[lightgray]}] (ULY) edge[lightgray] (Y2);
    \path [->,>={Stealth[lightgray]}] (ULD) edge[lightgray] (L2);
    \path [->,>={Stealth[lightgray]}] (ULD) edge[lightgray] (D2);
    \path [->,>={Stealth[lightgray]}] (L0) edge[lightgray] (D1);
    \path [->,>={Stealth[lightgray]}] (L0) edge[lightgray] (L2);
    \path [->,>={Stealth[lightgray]}] (L0) edge[lightgray, bend right] (Y1);
\end{scope}
\end{tikzpicture}

\label{fig:AllDeltaHoldNoL0}
}
\subfloat[]{
\begin{tikzpicture}
\begin{scope}[every node/.style={thick,draw=none}]
    \node (A) at (-1,0) {$A$};
    \node (Ay) at (1,1) {$A_Y$};
	\node (Ad) at (1,-1) {$A_D$};
	\node (Y1) at (3,1) {$Y_1$};
    \node (D1) at (3,-1) {$D_1$};
    \node (Y2) at (6,1) {$Y_2$};
    \node (D2) at (6,-1) {$D_2$};
    \node (L2) at (4,4) {$L_1$};
    \node[lightgray] (ULY) at (6,4) {$U_{L,Y}$};
    \node[lightgray] (L0) at (1,4) {$L_0$};
\end{scope}

\begin{scope}[>={Stealth[black]},
              every node/.style={fill=white,circle},
              every edge/.style={draw=black,very thick}]
    \path [->] (A) edge[line width=0.85mm] (Ad);
    \path [->] (A) edge[line width=0.85mm] (Ay);
	\path [->] (Ad) edge (D1);
    \path [->] (Ad) edge[bend right] (D2);
	\path [->] (Ay) edge[bend left] (Y2);
    \path [->] (Ay) edge (Y1);	
    \path [->] (L2) edge (D2);
    \path [->] (L2) edge (Y2);
    \path [->] (Y1) edge (D2);
    \path [->] (Y1) edge (Y2);
    \path [->] (D1) edge (D2);
    \path [->] (D1) edge (Y1);
    \path [->] (D2) edge (Y2);
    \path [->] (Y1) edge (L2);
    \path [->] (D1) edge (L2);
    \path [->,>={Stealth[blue]}] (Ad) edge[blue] (L2);
    \path [->,>={Stealth[lightgray]}] (ULY) edge[lightgray] (L2);
    \path [->,>={Stealth[lightgray]}] (ULY) edge[lightgray] (Y2);
    \path [->,>={Stealth[lightgray]}] (L0) edge[lightgray] (D1);
    \path [->,>={Stealth[lightgray]}] (L0) edge[lightgray] (L2);
    \path [->,>={Stealth[lightgray]}] (L0) edge[lightgray] (Y1);
\end{scope}
\end{tikzpicture}
} 
\caption{Extended graphs that explicitly depict measured and unmeasured variables.  The dismissible component conditions are violated in (e)-(f).}
\label{fig:partial isolationExpanded}
\end{figure}
\clearpage

\begin{figure}
    \centering
\subfloat[]{    
\begin{tikzpicture}
\begin{scope}[every node/.style={thick,draw=none}]
    \node (Ay) at (1,1) {$A_Y(G)$};
	\node (Ad) at (1,-1) {$A_D(G)$};
	\node (Y1) at (3,1) {$Y_1(G)$};
    \node (D1) at (3,-1) {$D_1(G)$};
    \node (Y2) at (6,1) {$Y_2(G)$};
    \node (D2) at (6,-1) {$D_2(G)$};
    \node (L1) at (4,4) {$L_1(G)$};
    \node (EMPTY) at (4,-4) {$ $};
\end{scope}

\begin{scope}[>={Stealth[black]},
              every node/.style={fill=white,circle},
              every edge/.style={draw=black,very thick}]
	\path [->] (Ad) edge (D1);
    \path [->] (Ad) edge[bend right] (D2);
	\path [->] (Ay) edge[bend left] (Y2);
    \path [->] (Ay) edge (Y1);	
    \path [->] (L1) edge (D2);
    \path [->] (L1) edge (Y2);
    \path [->] (Y1) edge (D2);
    \path [->] (Y1) edge (Y2);
    \path [->] (D1) edge (D2);
    \path [->] (D1) edge (Y1);
    \path [->] (D2) edge (Y2);
    \path [->] (Y1) edge (L1);
    \path [->] (D1) edge[bend right=10] (L1);
    \path [->,>={Stealth[blue]}] (Ad) edge[blue] (L1);
\end{scope}
\end{tikzpicture}
}
\subfloat[]{    
\begin{tikzpicture}
\begin{scope}[every node/.style={thick,draw=none}]
    \node (Ay) at (1,1) {$A_Y(G)$};
	\node (Ad) at (1,-1) {$A_D(G)$};
	\node (Y1) at (3,1) {$Y_1(G)$};
    \node (D1) at (3,-1) {$D_1(G)$};
    \node (Y2) at (6,1) {$Y_2(G)$};
    \node (D2) at (6,-1) {$D_2(G)$};
    \node (LY1) at (4,4) {$L_{A_Y,1}(G)$};
    \node (LD1) at (4,-4) {$L_{A_D,1}(G)$};
\end{scope}

\begin{scope}[>={Stealth[black]},
              every node/.style={fill=white,circle},
              every edge/.style={draw=black,very thick}]
	\path [->] (Ad) edge (D1);
    \path [->] (Ad) edge[bend right] (D2);
	\path [->] (Ay) edge[bend left] (Y2);
    \path [->] (Ay) edge (Y1);
    \path [->,>={Stealth[blue]}] (Ay) edge[blue] (LY1);
    \path [->,>={Stealth[blue]}] (Ad) edge[blue] (LD1);
    \path [->] (LY1) edge (D2);
    \path [->] (LY1) edge (Y2);
    \path [->] (LD1) edge (D2);
    \path [->] (LD1) edge (Y2);
    \path [->] (Y1) edge (D2);
    \path [->] (Y1) edge (Y2);
    \path [->] (D1) edge (D2);
    \path [->] (D1) edge (Y1);
    \path [->] (D2) edge (Y2);
    \path [->] (Y1) edge (LY1);
    \path [->] (D1) edge[bend right=10] (LY1);
    \path [->] (Y1) edge[bend left=10] (LD1);
    \path [->] (D1) edge[bend right=10] (LD1);
\end{scope}
\end{tikzpicture}
}
    \caption{The graph in (a) is a successive transformation of Figure \ref{fig:partial isolation}a for $L_1=Z_1$ that represents a hypothetical trial $G$ in which both $A_Y$ and $A_D$ are randomly assigned. The graph in (b) is a transformation of Figure \ref{fig:no isolation}b, in which $L_{A_Y,1}(G) \equiv Z_{A_Y,1}(G), L_{A_D,1}(G) \equiv Z_{A_D,1}(G)$. All dismissible component conditions hold in both graphs.}
    \label{fig:transformG}
\end{figure}
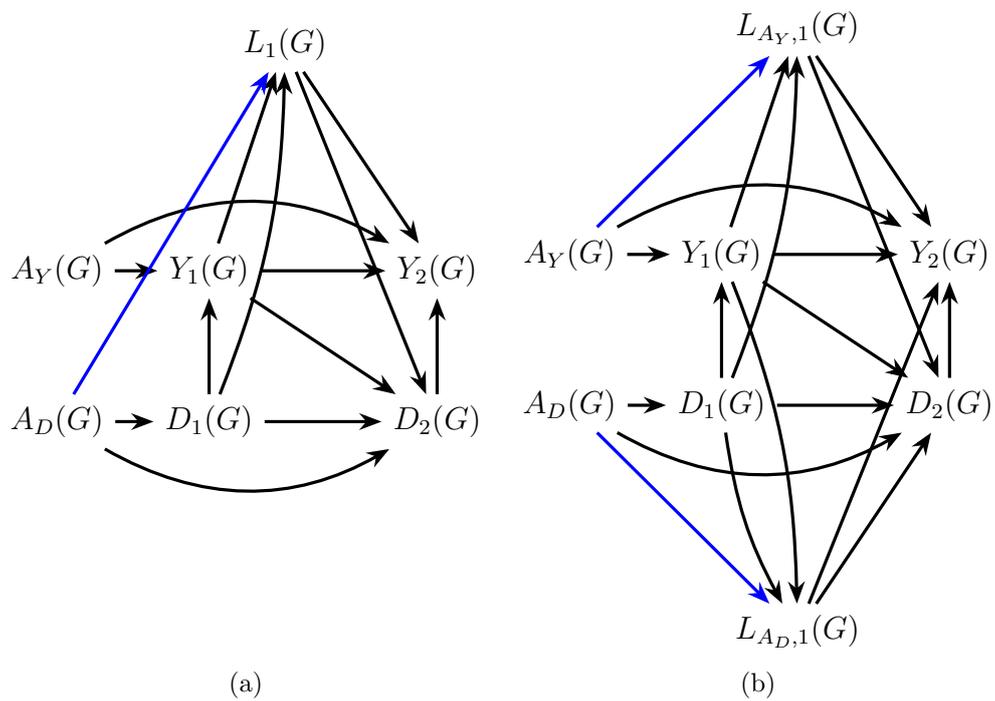

\clearpage

\begin{figure}
\centering
\begin{minipage}{.5\linewidth}
\centering
\subfloat[]{\label{main:a}\includegraphics[scale=.45]{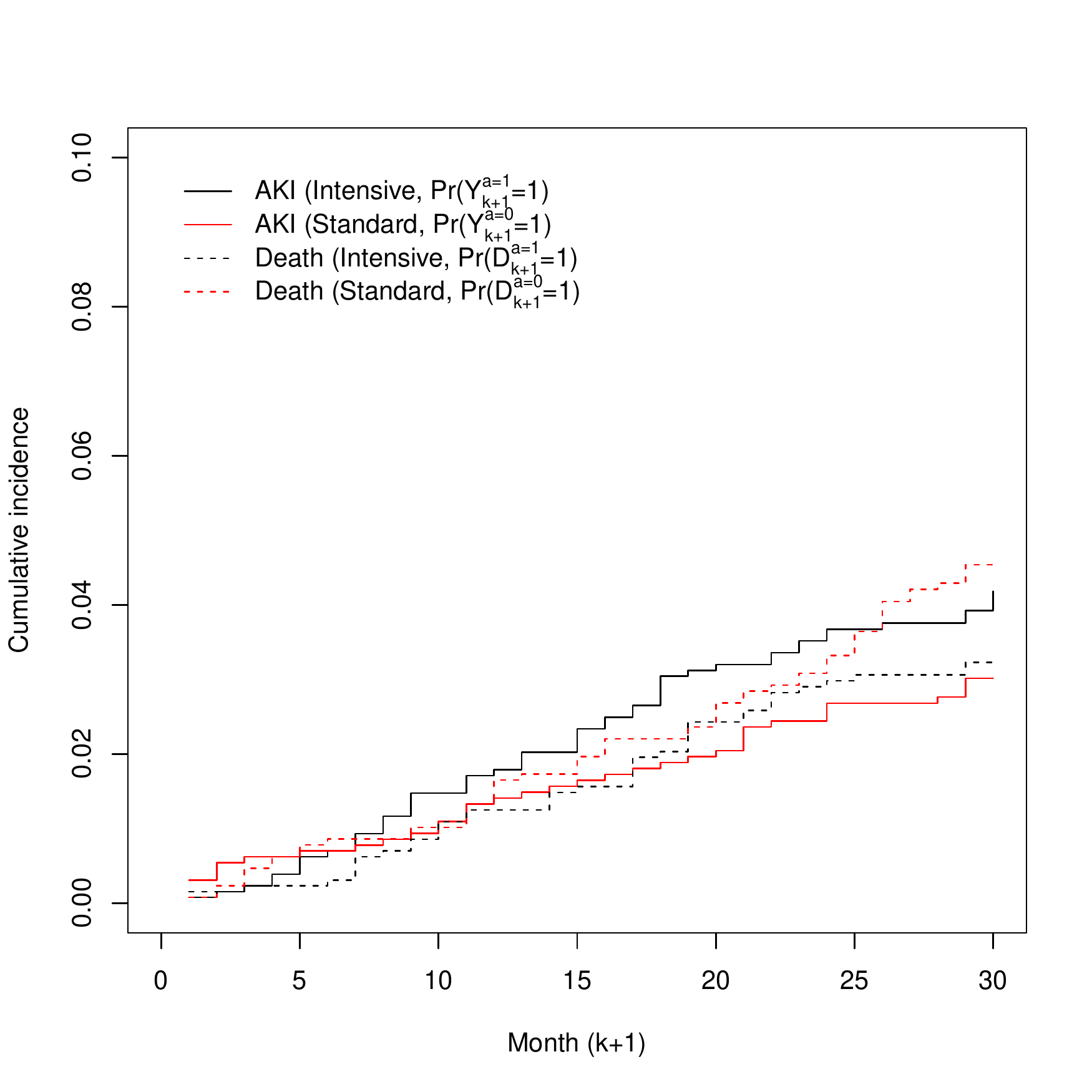}}
\end{minipage}%
\begin{minipage}{.5\linewidth}
\subfloat[]{\label{main:b}\includegraphics[scale=.45]{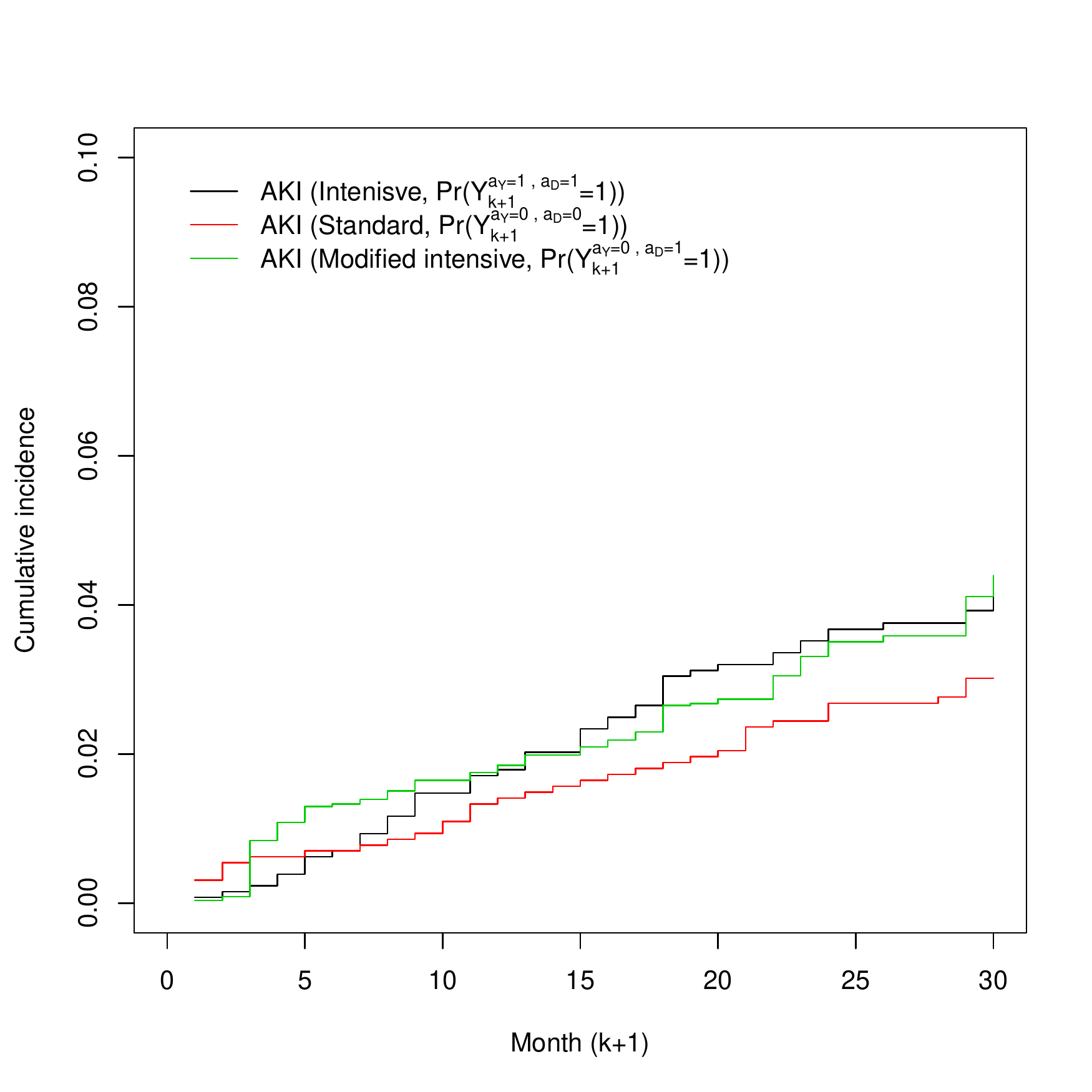}}
\end{minipage}
\caption{(a) Weighted Aalen-Johansen estimates of the cumulative incidence functions for acute kidney injury (AKI, solid lines) and death (dashed lines) under intensive ($a=1$, red) and standard ($a=0$, black) treatment. (b) Estimates of AKI cumulative incidence based on methods of Section \ref{sec: estimation} under a modified treatment containing only the $A_D$ component ($a_D=1,a_Y=0$, green). Cumulative incidence estimates under the original intensive ($a=a_Y=a_D=1$, red) and the standard ($a=a_Y=a_D=0$, black) of (a) are overlaid.}
\label{fig: ci plot}
\end{figure}

\clearpage

\appendix

\section{Modified treatment assumption}
\label{sec: modified treatment}
To define the generalized decomposition assumption in Section \ref{sec: def decomposition}, we considered a decomposition of treatment $A$ into different components, $A_Y$ and $A_D$, satisfying \eqref{assumption: Determinsm}. Yet, a physical decomposition of $A$ into components $A_Y$ and $A_D$ is not necessary for the validity of our results on identification and estimation of separable effects in Sections \ref{sec: identifiability conditions}-\ref{sec: estimation}.  Specifically, the proofs in Appendix \ref{sec: proof of idenditifiability}  only require condition \eqref{eq: definition A=Ay=Ad} of the generalized decomposition assumption, which may also hold for treatments $A_Y$ and $A_D$ that are not components of $A$.  

In this case, the separable effects can still be meaningfully interpreted as the effects of joint assignment to alternative treatments $A_Y$ and $A_D$ in place of assignment to $A$.  The isolation conditions of Section \ref{sec: isolation and interpret} still constitute additional mechanistic assumptions on how these alternative treatments operate on the event of interest and the competing event, as defined relative to a $G$ transformation; that is, relative to a trial in which $A_Y$ and $A_D$ are randomly assigned. The isolation conditions can thus be evaluated in $G$ transformation graphs, as opposed to extended causal DAGs. Note that in the case of a treatment decomposition, the isolation conditions can be evaluated with respect to either extended causal DAGs or their $G$ transformations and the same conclusions will be reached.  

However, when treatments $A_Y$ and $A_D$ are not components of $A$, we require additional assumptions beyond \eqref{eq: definition A=Ay=Ad} for the separable effects to \textit{explain} the mechanism by which the original treatment $A$ exerts its effects on $Y_{k+1}$ for $k \in \{0, \dots, K \}$. The following is an alternative assumption to \eqref{assumption: Determinsm} that, when coupled with \eqref{eq: definition A=Ay=Ad}, is sufficient for the separable effects to explain the total effect of the original treatment $A$ on the event of interest when $A_Y$ and $A_D$ are not a decomposition of $A$.  For variables, $M_Y$ and $M_D$, consider the following assumption:
\begin{align}
& \text{$A_Y$ and $A_D$ exert all their effects through $M_Y$ and $M_D$, respectively, and} \nonumber \\
& \quad M^{a_Y=a,a_D}_Y = M^{a}_Y \quad \text{ for } a_D \in \{0,1\} \nonumber \\ 
& \quad M^{a_Y,a_D=a}_D = M^{a}_D \quad  \text{ for } a_Y \in \{0,1\} \label{ass: m equal},
\end{align}
where the counterfactuals on the left and right hand side of the equality refer to assignment to $A=a$ and no level of $A_Y$ and $A_D$, and assignment to ($A_Y=a_Y$, $A_D=a_D$) and no level of $A$, respectively. A transformed DAG that is consistent with \eqref{ass: m equal} is shown in Figure \ref{fig: modified treatment ass} where $G'$ refers to a six arm trial in which subjects are either randomly assigned to $A$ (and no level of $A_Y$ and $A_D$) or a combination of $A_Y=a_Y$ and $A_D=a_D$ (and no level of $A$).

Note that assumptions \eqref{eq: definition A=Ay=Ad} and \eqref{ass: m equal}  can, in principle, be falsified in a future randomized experiment. For example, by randomly assigning individuals to $A$ or a joint treatment $(A_Y,A_D)$, we can assess whether $E(W \mid A_Y=a, A_D=a)\neq E(W \mid A=a)$, for any $W \in \{Y_{1},\dots,Y_{K+1},D_{1},\dots,D_{k+1},Z_{1},\dots,Z_{k+1}, M_Y, M_D\}$.

To fix ideas, consider a study of estrogen therapy versus placebo in men with prostate cancer, which was the running example in Stensrud et al \cite{stensrud2019separable}. Estrogen therapy is thought to reduce death due to prostate cancer, because it reduces testosterone levels and thus prevents the cancer cells from growing. However, there is concern that estrogen therapy may also increase mortality due to cardiovascular disease, e.g.\ through estrogen-induced synthesis of coagulation factors \cite{turo2014diethylstilboestrol}. Stensrud et al  \cite{stensrud2019separable} used this example to motivate the separable direct and indirect effects under full isolation, and suggested that alternative treatments, such as castration and luteinizing hormone releasing hormone (LHRH) antagonists, can have the same effect as estrogen on testosterone reduction ($M_Y$), but, unlike estrogen, these treatments do not exert effects on the coagulation factors ($M_D$). Whereas Stensrud et al \cite{stensrud2019separable} did not formally define the variables $M_Y$ and $M_D$, providing the story that includes these additional variables, satisfying \eqref{ass: m equal}, is essential to connect the effect of e.g.\ $A_Y=1$ (here, assigning testosterone or LHRH antagonists) to the separable direct and indirect effects of estrogen therapy ($A$) itself. 

\begin{figure}
\centering
\begin{tikzpicture}
\begin{scope}[every node/.style={thick,draw=none}]
    \node (A) at (-1.5,0) {$A(G')$};
    \node (Ay) at (-1,1) {$A_{Y}(G')$};
	\node (Ad) at (-1,-1) {$A_{D} (G')$};
    \node (My) at (1.5,1) {$M_{Y} (G')$};
	\node (Md) at (1.5,-1) {$M_{D} (G')$};
	\node (Y1) at (4,1) {$Y_1 (G') $};
    \node (D1) at (4,-1) {$D_1 (G')$};
    \node (Y2) at (7,1) {$Y_2 (G')  $};
    \node (D2) at (7,-1) {$D_2 (G') $};
    \node (L) at (2.5,4) {$L (G')$};
\end{scope}

\begin{scope}[>={Stealth[black]},
              every node/.style={fill=white,circle},
              every edge/.style={draw=black,very thick}]
	\path [->] (Md) edge (D1);
    \path [->] (Md) edge[bend right] (D2);
	\path [->] (My) edge[bend left] (Y2);
    \path [->] (My) edge (Y1);	
    \path [->] (Ay) edge (My);
    \path [->] (Ad) edge (Md);
    \path [->] (A) edge (My);
    \path [->] (A) edge (Md);
    \path [->] (L) edge (Y1);
    \path [->] (L) edge (Y2);
    \path [->] (L) edge[bend right] (D1);
    \path [->] (L) edge (D2);
    \path [->] (Y1) edge (D2);
    \path [->] (Y1) edge (Y2);
    \path [->] (D1) edge (D2);
    \path [->] (D1) edge (Y1);
    \path [->] (D2) edge (Y2);
\end{scope}
\end{tikzpicture}
\caption{Modified DAG including the additional variables $M_{Y}$ and $M_{D}$ and their relation to $A$, $A_{Y}$ and $A_{D}$.}
\label{fig: modified treatment ass}
\end{figure}
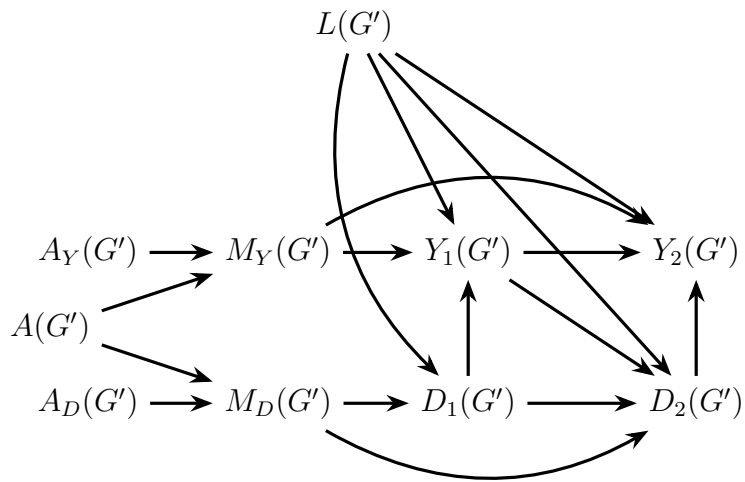

\clearpage

\section{Proof of identifiability}
\label{sec: proof of idenditifiability}
Before we provide a proof of identification formula \eqref{eq: identifying formula}, consider the following identifiability conditions that generalize the conditions from Section \ref{sec: identifiability conditions} to allow for censoring.
\begin{enumerate}
\item[1. ] Exchangeability: 
\begin{align}
&\bar{Y}_{K+1}^{a, \bar{c}=0},\bar{D}_{K+1}^{a, \bar{c}=0},\bar{L}^{a, \bar{c}=0}_{K+1} \independent A \mid L_{0} \label{ass: E1 app} \\
&\underline{Y}^{a, \bar{c}=0}_{k+1}, \underline{D}^{a, \bar{c}=0}_{k+1}, \underline{L}^{a, \bar{c}=0}_{k+1}  \independent C_{k+1} \mid Y_k = D_k = \bar{C}_k = 0, \bar{L}_k, A \label{ass: E2 app}
\end{align}
Condition \eqref{ass: E1 app} holds when $A_Y$ and $A_D$ are randomly assigned at baseline, possibly conditional on $L_0$. Condition \eqref{ass: E2 app} requires that losses to follow-up are independent of future counterfactual events, given the measured past; this assumption does not hold by design in a randomised trial, as losses to follow-up are not randomly assigned in practice.
\item[2.] Positivity:  
\begin{align}
& \Pr(L_0=l_0)>0\implies   \nonumber \\
& \quad \Pr (A=a\mid  L_0=l_0)>0, \label{eq: positivity 1 app} \\
& f_{\overline{L}_k,D_{k+1},C_{k+1},Y_k}(\overline{l}_k,0,0,0) > 0  \implies \nonumber \\ 
& \quad  \Pr(A=a|D_{k+1}=C_{k+1}=Y_k=0,\overline{L}_k=\overline{l}_k)>0 \label{eq: positivity 2 app} \\
& \Pr(A = a,Y_k=0,D_k=0,\bar{C}_k=0,\bar{L}_k=\overline{l}_k) > 0 \implies \nonumber \\
&\quad \Pr(C_{k+1}=0\mid Y_k=0,D_k=0,\bar{C}_k=0,\bar{L}_k=\overline{l}_k,A=a) > 0 \label{eq: positivity 3 app}
\end{align}
for $a\in\{0,1\}$, $k \in \{0,  \dots ,K\}$ and $L_k \in \mathcal{L}$. Conditions \eqref{eq: positivity 1 app} and \eqref{eq: positivity 2 app} were described in the main text. Condition 
\eqref{eq: positivity 3 app} requires that for any possible history of treatment assignment and covariates among those who are event-free and uncensored at $k$, some subjects will remain uncensored at $k+1$.
\item[3.] Consistency: 
\begin{align}
  & \text{if } A=a \text{ and } \bar{C}_{k+1} = 0, \nonumber \\
  & \text{then } \bar{Y}_{k+1} = \bar{Y}^{a,  \bar{c}=0}_{k+1}, \bar{D}_{k+1} = \bar{D}^{a,  \bar{c}=0}_{k+1} \text{ and } \bar{L}_{k+1} = \bar{L}^{a,\bar{c}=0}_{k+1}.
  \label{ass: consistency cens}
\end{align}
Consistency is satisfied if any individual who has data history consistent with the intervention under a counterfactual scenario, would have an observed outcome that is equal to the counterfactual outcome. 
\item[4.] Dismissible component conditions: 
\begin{align}
& Y^{\bar{c}=0}_{k+1}(G) \independent A_D(G) \mid A_Y(G),D^{\bar{c}=0}_{k+1}(G)= Y^{\bar{c}=0}_{k}(G)=0, \bar{L}^{\bar{c}=0}_k(G), \label{ass: delta 1 app} \\
& D^{\bar{c}=0}_{k+1}(G) \independent A_Y(G) \mid A_D(G), D^{\bar{c}=0}_{k}(G)=Y^{\bar{c}=0}_{k}(G)=0, \bar{L}^{\bar{c}=0}_{k}(G), \label{ass: delta 2 app}  \\
& L^{\bar{c}=0}_{A_Y,k}(G) \independent A_D(G) \mid A_Y(G), Y^{\bar{c}=0}_{k}(G)=D^{\bar{c}=0}_{k}(G)=0, \bar{L}^{\bar{c}=0}_{k-1}(G),L^{\bar{c}=0}_{A_D,k}(G), \label{ass: delta 3a app}  \\
& L^{\bar{c}=0}_{A_D,k}(G) \independent A_Y(G) \mid A_D(G), D^{\bar{c}=0}_{k}(G)=Y^{\bar{c}=0}_{k}(G)=0. \bar{L}^{\bar{c}=0}_{k-1}(G) \label{ass: delta 3b app}. 
\end{align}
The dismissible component conditions are identical to the conditions in Section \ref{sec: identifiability conditions}, but the superscript ${\bar{c}=0}$ is included to emphasize that we consider outcomes in a setting in which loss to follow-up is eliminated even under $G$.
\end{enumerate}

\begin{lemma} \label{lemma: delta conditions}
Under a FFRCISTG model, the dismissible component conditions \eqref{ass: delta 1 app}-\eqref{ass: delta 3b app} imply the following equalities for $a_Y,a_D \in \{0,1\}$: 
\begin{align}
 & \Pr(Y^{a_Y,a_D=0,\bar{c}= 0}_{k+1} = 1 \mid  Y^{a_Y,a_D=0,\bar{c}= 0}_k = 0, D^{a_Y,a_D=0,\bar{c}= 0}_{k+1} = 0,\bar{L}^{a_Y,a_D=0,\bar{c}= 0}_{k} = \bar{l}_{k}) \label{eq: ineq y} \\
=&\Pr(Y^{a_Y,a_D=1,\bar{c}= 0}_{k+1} = 1 \mid  Y^{a_Y,a_D=1,\bar{c}=0}_k = D^{a_Y,a_D=1,\bar{c}=0}_{k+1} = 0 , \bar{L}^{a_Y,a_D=1,\bar{c}= 0}_{k} = \bar{l}_{k} ), \nonumber \\
 & \Pr(D^{a_Y=0,a_D,\bar{c}= 0}_{k+1} = 1 \mid  Y^{a_Y=0,a_D,\bar{c}= 0}_k = D^{a_Y=0,a_D,\bar{c}= 0}_k = 0, \bar{L}^{a_Y=0,a_D,\bar{c}= 0}_{k} = \bar{l}_{k}) \label{eq: ineq d}  \\
=&\Pr(D^{a_Y=1,a_D,\bar{c}= 0}_{k+1} = 1 \mid  Y^{a_Y=1,a_D,\bar{c}=0}_k = 0, D^{a_Y=1,a_D,\bar{c}=0}_k = 0 ,  \bar{L}^{a_Y=1,a_D,\bar{c}= 0}_{k} = \bar{l}_{k} ), \nonumber \\ 
   & \Pr(L^{a_Y,a_D=1,\bar{c}= 0}_{A_Y,k+1} = l^{a_Y,a_D=1,\bar{c}= 0}_{A_Y,k+1} \mid  Y^{a_Y,a_D=1,\bar{c}= 0}_{k+1} = D^{a_Y,a_D=1,\bar{c}= 0}_{k+1} = 0, \label{eq: ineq ly} \\ 
   & \quad \bar{L}^{a_Y,a_D=1,\bar{c}= 0}_{k} = \bar{l}_{k},L^{a_Y,a_D=1,\bar{c}= 0}_{A_D,k+1} = l_{A_D,k+1}) \nonumber \\
   =& \Pr (L^{a_Y,a_D=0,\bar{c}= 0}_{A_Y,k+1} = l^{a_Y,a_D=0,\bar{c}= 0}_{A_Y,k+1} \mid  Y^{a_Y,a_D=0,\bar{c}= 0}_{k+1} = D^{a_Y,a_D=0,\bar{c}= 0}_{k+1} = 0, \nonumber \\ 
   & \quad \bar{L}^{a_Y,a_D=0,\bar{c}= 0}_{k} = \bar{l}_{k},L^{a_Y,a_D=0,\bar{c}= 0}_{A_D,k+1} = l_{A_D,k+1}), \nonumber \\
  & \Pr(L^{a_Y=0,a_D,\bar{c}= 0}_{A_D,k+1} = l^{a_Y=0,a_D,\bar{c}= 0}_{A_D,k+1} \mid  Y^{a_Y=0,a_D,\bar{c}= 0}_{k+1} = D^{a_Y=0,a_D,\bar{c}= 0}_{k+1} = 0, \bar{L}^{a_Y=0,a_D,\bar{c}= 0}_{k} = \bar{l}_{k}) \label{eq: ineq ld} \\
   =& \Pr(L^{a_Y=1,a_D,\bar{c}= 0}_{A_D,k+1} = l^{a_Y=1,a_D,\bar{c}= 0}_{A_D,k+1} \mid  Y^{a_Y=1,a_D,\bar{c}= 0}_{k+1} = D^{a_Y=1,a_D,\bar{c}= 0}_{k+1} = 0, \bar{L}^{a_Y=1,a_D,\bar{c}= 0}_{k} = \bar{l}_{k}). \nonumber
\end{align}
\end{lemma}
\begin{proof}


\begin{align*}
& \Pr(Y^{a_Y,a_D=0,\bar{c}= 0}_{k+1} = 1 \mid  Y^{a_Y,a_D=0,\bar{c}= 0}_k = 0, D^{a_Y,a_D=0,\bar{c}= 0}_{k+1} = 0,\bar{L}^{a_Y,a_D=0,\bar{c}= 0}_{k} = \bar{l}_{k} ) \\ 
= & \Pr (Y^{\bar{c}= 0}_{k+1}(G)=1  \mid Y^{\bar{c}= 0}_{k}(G)=0,D^{\bar{c}= 0}_{k+1}(G)=0, \bar{L}^{\bar{c}= 0}_{k}(G)= \bar{l}_{k}, A_Y(G)=a_Y,A_D(G)=0)  \text{ by def.\ of } G \nonumber \\ 
  =& \Pr (Y^{\bar{c}= 0}_{k+1}(G)=1  \mid Y^{\bar{c}= 0}_{k}(G)=0,D^{\bar{c}= 0}_{k+1}(G)=0, \bar{L}^{\bar{c}= 0}_{k}(G)= \bar{l}_{k}, A_Y(G)=a_Y,A_D(G)=1) \text{ due to \eqref{ass: delta 1 app}} \nonumber \\
= & \Pr(Y^{a_Y,a_D=1,\bar{c}= 0}_{k+1} = 1 \mid  Y^{a_Y,a_D=1,\bar{c}= 0}_k = 0, D^{a_Y,a_D=1,\bar{c}= 0}_{k+1} = 0,\bar{L}^{a_Y,a_D=1,\bar{c}= 0}_{k} = \bar{l}_{k} ) \text{ by def.\ of } G, \nonumber \\ 
\end{align*} 
which shows that equality \eqref{eq: ineq y} holds, and \eqref{eq: ineq d}-\eqref{eq: ineq ld} can be shown from analogous arguments, where we use conditions \eqref{ass: delta 2 app}-\eqref{ass: delta 3b app} in the second step, respectively, instead of \eqref{ass: delta 1 app}. 
\end{proof}

\begin{lemma} \label{lemma: conditional counter to obs}
Suppose that conditions \eqref{ass: E1 app}-\eqref{ass: consistency cens} hold.
Then, for $s = 0,\dots, K$ and $a_Y,a_D \in \{0,1\}$, 
\begin{align}
  \Pr & (Y^{a_Y=a_D=a,  \bar{c}=0}_{s+1}=1 \mid D^{a_Y=a_D=a,  \bar{c}=0}_{s+1}= Y^{a_Y=a_D=a,  \bar{c}=0}_{s}=0,\bar{L}^{a_Y=a_D=a,  \bar{c}=0}_{s} = \bar{l}_{s}) \label{eq: counterfactual to obs y} \\
= & \Pr(Y_{s+1}=1 \mid C_{s+1}= D_{s+1}= Y_{s}=0,\bar{L}_{s} = \bar{l}_{s}, A = a), \nonumber \\
  \Pr& (D^{a_Y=a_D=a, \bar{c}=0}_{s+1}=1 \mid D^{a_Y=a_D=a, \bar{c}=0}_{s}= Y^{a_Y=a_D=a, \bar{c}=0}_{s}=0,\bar{L}^{a_Y=a_D=a, \bar{c}=0}_{s} = \bar{l}_{s}) \label{eq: counterfactual to obs d} \\
=& \Pr(D_{s+1}=1 \mid C_{s+1} = D_{s}= Y_{s}=0,\bar{L}_{s} = \bar{l}_{s}, A=a), \nonumber \\
    \Pr(&L^{a_Y=a_D=a,  \bar{c}=0}_{A_Y, s} = l_{A_Y,s} \mid  Y^{a_Y=a_D=a,  \bar{c}=0}_{s} = D^{a_Y=a_D=a,  \bar{c}=0}_{s} = 0, \bar{L}^{a_Y=a_D=a,  \bar{c}=0}_{s} = \bar{l}_{s}) \label{eq: counterfactual to obs ly} \\
    =&\Pr(L_{A_Y,s} = \bar{l}_{A_Y,s} \mid  C_{s} = D_{s} = Y_{s} = 0, \bar{L}_{s-1}=\bar{l}_{s-1}, A=a), \nonumber \\
  \Pr(&L^{a_Y=a_D=a, \bar{c}=0}_{A_D, s} = l_{A_Y,s} \mid  Y^{a_Y=a_D=a, \bar{c}=0}_{s} = D^{a_Y=a_D=a, \bar{c}=0}_{s} = 0, \bar{L}^{a_Y=a_D=a, \bar{c}=0}_{s} = \bar{l}_{s}) \label{eq: counterfactual to obs ld} \\
   =& \Pr(L_{A_Y,s} = l_{A_D,s} \mid  C_{s} = D_{s} = Y_{s} = 0, \bar{L}_{s-1}=\bar{l}_{s-1}, A = a) \nonumber.
\end{align}
\end{lemma}
\begin{proof}
Consider first \eqref{eq: counterfactual to obs y},
\begin{align*}
 & \Pr(Y^{a_Y=a_D=a,  \bar{c}=0}_{s+1}=1 \mid D^{a_Y=a_D=a,  \bar{c}=0}_{s+1}= Y^{a_Y=a_D=a,  \bar{c}=0}_{s}=0,\bar{L}^{a_Y=a_D=a,  \bar{c}=0}_{s} = \bar{l}_{s}) \nonumber \\
=  & \Pr(Y^{a,  \bar{c}=0}_{s+1}=1 \mid D^{a,  \bar{c}=0}_{s+1}= Y^{a,  \bar{c}=0}_{s}=0,\bar{L}^{a,  \bar{c}=0}_{s} = \bar{l}_{s}) \nonumber \\
= & \Pr(Y^{a,  \bar{c}=0}_{s+1}=1 \mid D^{a,  \bar{c}=0}_{s+1}= Y^{a,  \bar{c}=0}_{s}=Y_{0}=D_{0}=\bar{C}_{0}=0,L_0,\bar{L}^{a, \bar{c}=0}_{s} = \bar{l}_{s}) \nonumber \\
=& \frac{\Pr(Y^{a,  \bar{c}=0}_{s+1}=1, \bar{D}^{a,  \bar{c}=0}_{s+1}= \bar{Y}^{a,  \bar{c}=0}_{s}=0, \bar{L}^{a,  \bar{c}=0}_{s} = \bar{l}_{s} \mid Y_{0}=D_{0}=\bar{C}_{0}=0,L_0, A=a)}{\Pr(\bar{D}^{a,  \bar{c}=0}_{s+1}= \bar{Y}^{a,  \bar{c}=0}_{s}=0, \bar{L}^{a,  \bar{c}=0}_{s} = \bar{l}_{s} \mid Y_{0}=D_{0}=\bar{C}_{0}=0,L_0, A=a)},  \nonumber \\
\end{align*}
where we used the fact that all subjects are event-free and uncensored at $t=0$, laws of probability and \eqref{ass: E1 app}. Using \eqref{ass: E2 app} and positivity,
\begin{align}
& \frac{\Pr(Y^{a,  \bar{c}=0}_{s+1}=1, \bar{D}^{a,  \bar{c}=0}_{s+1}= \bar{Y}^{a,  \bar{c}=0}_{s}=0, \bar{L}^{a,  \bar{c}=0}_{s} = \bar{l}_{s} \mid Y_{0}=D_{0}=\bar{C}_{1}=0,L_0, A=a)}{\Pr(\bar{D}^{a,  \bar{c}=0}_{s+1}= \bar{Y}^{a,  \bar{c}=0}_{s}=0, \bar{L}^{a,  \bar{c}=0}_{s} = \bar{l}_{s} \mid Y_{0}=D_{0}=\bar{C}_{1}=0,L_0, A=a)},  \nonumber \\
 =& \Pr(Y^{a,  \bar{c}=0}_{s+1}=1 \mid D^{a,  \bar{c}=0}_{s+1}= Y^{a,  \bar{c}=0}_{s}=0,\bar{L}^{a,  \bar{c}=0}_{s} = \bar{l}_{s}, Y_{0}=D_{0}=\bar{C}_{1}=0,L_0, A=a). \nonumber \\
\end{align}
For $s=0$, under consistency, 
\begin{align}
& \Pr(Y^{a,  \bar{c}=0}_{s+1}=1 \mid D^{a,  \bar{c}=0}_{s+1}= Y^{a,  \bar{c}=0}_{s}=0,\bar{L}^{a,  \bar{c}=0}_{s} = \bar{l}_{s}, Y_{0}=D_{0}=\bar{C}_{1}=0,L_0, A=a) \nonumber \\
=& \Pr(Y^{a,  \bar{c}=0}_{s+1}=1 \mid D^{a,  \bar{c}=0}_{s+1}= Y^{a,  \bar{c}=0}_{s}=0,\bar{L}^{a,  \bar{c}=0}_{s} = \bar{l}_{s}, Y_{0}=D_{1}=\bar{C}_{1}=0,\bar{L}_0, A=a)
\nonumber \\
=& \Pr(Y^{a,  \bar{c}=0}_{s+1}=1 \mid  Y_{0}=D_{1}=\bar{C}_{1}=0,\bar{L}_0, A=a)
\nonumber \\
\end{align}
which proves the lemma for $s=0$.

Further, for $s>1$, using consistency,
\begin{align}
& \Pr(Y^{a,  \bar{c}=0}_{s+1}=1 \mid D^{a,  \bar{c}=0}_{s+1}= Y^{a,  \bar{c}=0}_{s}=0,\bar{L}^{a,  \bar{c}=0}_{s} = \bar{l}_{s}, Y_{0}=D_{0}=\bar{C}_{1}=0,L_0, A=a) \nonumber \\
=& \Pr(Y^{a,  \bar{c}=0}_{s+1}=1 \mid D^{a,  \bar{c}=0}_{s+1}= Y^{a,  \bar{c}=0}_{s}=0,\bar{L}^{a,  \bar{c}=0}_{s} = \bar{l}_{s}, Y_{1}=D_{1}=\bar{C}_{1}=0,\bar{L}_1, A=a)
\nonumber \\
\end{align}

Apply \eqref{ass: E2 app} and positivity again,
\begin{align}
    & \Pr(Y^{a,  \bar{c}=0}_{s+1}=1 \mid D^{a,  \bar{c}=0}_{s+1}= Y^{a,  \bar{c}=0}_{s}=0,\bar{L}^{a,  \bar{c}=0}_{s} = \bar{l}_{s}, Y_{1}=D_{1}=\bar{C}_{1}=0,\bar{L}_1, A=a) \nonumber \\
    =& \frac{\Pr(Y^{a,  \bar{c}=0}_{s+1}=1, \bar{D}^{a,  \bar{c}=0}_{s+1}= \bar{Y}^{a,  \bar{c}=0}_{s}=0, \bar{L}^{a,  \bar{c}=0}_{s} = \bar{l}_{s} \mid Y_{1}=D_{1}=\bar{C}_{2}=0,\bar{L}_1, A=a)}{\Pr(\bar{D}^{a,  \bar{c}=0}_{s+1}= \bar{Y}^{a,  \bar{c}=0}_{s}=0, \bar{L}^{a,  \bar{c}=0}_{s} = \bar{l}_{s} \mid Y_{1}=D_{1}=\bar{C}_{2}=0,\bar{L}_1, A=a)},  \nonumber \\
 =& \Pr(Y^{a,  \bar{c}=0}_{s+1}=1 \mid D^{a,  \bar{c}=0}_{s+1}= Y^{a \bar{c}=0}_{s}=0,\bar{L}^{a,  \bar{c}=0}_{s} = \bar{l}_{s}, Y_{1}=D_{1}=\bar{C}_{2}=0,\bar{L}_1, A=a). \nonumber \\
\end{align}
Using consistency,
\begin{align}
& \Pr(Y^{a,  \bar{c}=0}_{s+1}=1 \mid D^{a,  \bar{c}=0}_{s+1}= Y^{a,  \bar{c}=0}_{s}=0,\bar{L}^{a,  \bar{c}=0}_{s} = \bar{l}_{s}, Y_{1}=D_{1}=\bar{C}_{2}=0,\bar{L}_1, A=a) \nonumber \\
=& \Pr(Y^{a,  \bar{c}=0}_{s+1}=1 \mid D^{a,  \bar{c}=0}_{s+1}= Y^{a,  \bar{c}=0}_{s}=0,\bar{L}^{a,  \bar{c}=0}_{s} = \bar{l}_{s}, Y_{2}=D_{2}=\bar{C}_{2}=0,\bar{L}_2, A=a)
\nonumber \\
\end{align}
Arguing iteratively, we find that
\begin{align}
 & \Pr(Y^{a,  \bar{c}=0}_{s+1}=1 \mid D^{a,  \bar{c}=0}_{s+1}= Y^{a,  \bar{c}=0}_{s}=0,\bar{L}^{a,  \bar{c}=0}_{s} = \bar{l}_{s}) \nonumber \\
= & \Pr(Y_{s+1}=1 \mid C_{s+1}= D_{s+1}= Y_{s}=0,\bar{L}_{s} = \bar{l}_{s}, A=a).
\end{align}
Analogous arguments can be used to show \eqref{eq: counterfactual to obs d}-\eqref{eq: counterfactual to obs ld}. 
\end{proof}

\begin{theorem} \label{theorem: identification formula}
Suppose conditions \eqref{ass: E1 app}-\eqref{ass: delta 3b app} hold. 
Then, for $a_Y,a_D \in \{0,1\}$, 
\begin{align*}
 \Pr ( & Y_{k+1}^{a_Y,a_D,\bar{c}=0}=1) \\
       = & \sum_{\bar{l}_K}  \Big[ \sum_{s=0}^{K} \Pr(Y_{s+1}=1 \mid C_{s+1}= D_{s+1}= Y_{s}=0, \bar{L}_{s} = \bar{l}_{s}, A = a_Y) \nonumber \\ 
       &  \prod_{j=0}^{s}  \big\{ \Pr(D_{j+1}=0 \mid C_{j+1}=D_{j}= Y_{j}=0, \bar{L}_{j} = \bar{l}_{j},  A = a_D)  \nonumber \\
      &  \times \Pr(Y_{j}=0 \mid C_{j}=D_{j}= Y_{j-1}=0, \bar{L}_{j-1} = \bar{l}_{j-1},  A = a_Y) \nonumber \\
    &\times \Pr(L_{A_Y,j} = l_{A_Y,j} \mid C_{j}= Y_{j} = D_{j} = 0, \bar{L}_{j-1} = \bar{l}_{j-1}, L_{A_D,j} = l_{A_D,j}, A = a_Y) \nonumber \\
    &\times \Pr(L_{A_D,j} = l_{A_D,j} \mid C_{j}= Y_{j} = D_{j} = 0, \bar{L}_{j-1} = \bar{l}_{j-1},  A = a_D) \big\} \Big].
\end{align*}
\end{theorem}
\begin{proof}
If $a_Y=a_D \in \{0,1\}$, it is straightforward to use laws of probability to show that the theorem holds. Consider now the case where $a_Y \neq a_D$. In particular, let $a_Y=1 \neq a_D=0$. Using laws of probability and Lemma \ref{lemma: delta conditions}, 
\begin{align}
    \Pr & (Y^{a_Y=1,a_D=0, \bar{c}=0}_{K+1}=1)  \nonumber \\
  =& \sum_{\bar{l}_K}  \Pr(Y^{a_Y=1,a_D=0, \bar{c}=0}_{K+1}=1 \mid  \bar{L}^{a_Y=1,a_D=0, \bar{c}=0}_K) \Pr(\bar{L}^{a_Y=1,a_D=0, \bar{c}=0}_K=\bar{L}^{a_Y=1,a_D=0, \bar{c}=0}_K = \bar{l}_K) \nonumber \\
    =& \sum_{\bar{l}_K}  \Big[ \sum_{s=0}^{K} \Pr(Y^{a_Y=1,a_D=0, \bar{c}=0}_{s+1}=1 \mid D^{a_Y=1,a_D=0, \bar{c}=0}_{s+1}= Y^{a_Y=1,a_D=0, \bar{c}=0}_{s}=0, \bar{L}^{a_Y=1,a_D=0, \bar{c}=0}_{s} = \bar{l}_{s}) \nonumber \\
     &  \prod_{j=0}^{s}  \big\{ \Pr(D^{a_Y=1,a_D=0, \bar{c}=0}_{j+1}=0 \mid D^{a_Y=1,a_D=0, \bar{c}=0}_{j}= Y^{a_Y=1,a_D=0, \bar{c}=0}_{j}=0,  \bar{L}^{a_Y=1,a_D=0, \bar{c}=0}_{j} = \bar{l}_{j})  \nonumber \\
      &  \times \Pr(Y^{a_Y=1,a_D=0, \bar{c}=0}_{j}=0 \mid D^{a_Y=1,a_D=0, \bar{c}=0}_{j}= Y^{a_Y=1,a_D=0, \bar{c}=0}_{j-1}=0, \bar{L}^{a_Y=1,a_D=0, \bar{c}=0}_{j} = \bar{l}_{j})  \nonumber \\
        &\times \Pr(L^{a_Y=1,a_D=0,\bar{c}= 0}_{A_Y,j} = l_{A_Y,j} \mid  Y^{a_Y=1,a_D=0,\bar{c}= 0}_{j} = D^{a_Y=1,a_D=0,\bar{c}= 0}_{j} = 0, \bar{L}^{a_Y=1,a_D=0,\bar{c}= 0}_{j-1} =  \bar{l}_{j-1},  \nonumber \\ 
   & \quad \quad L^{a_Y=1,a_D=0,\bar{c}= 0}_{A_D,j} = l_{A_D,j})  \nonumber \\
            &\times \Pr(L^{a_Y=1,a_D=0,\bar{c}= 0}_{A_D, j} = l_{A_D,j} \mid  Y^{a_Y=1,a_D=0,\bar{c}= 0}_{j} = D^{a_Y=1,a_D=0,\bar{c}= 0}_{j} = 0, \bar{L}^{a_Y=1,a_D=0,\bar{c}= 0}_{j-1} =  \bar{l}_{j-1}) \big\} \Big] \nonumber \\
       =& \sum_{\bar{l}_K}  \Big[ \sum_{s=0}^{K} \Pr(Y^{a=1,  \bar{c}=0}_{s+1}=1 \mid D^{a=1,  \bar{c}=0}_{s+1}= Y^{a=1,  \bar{c}=0}_{s}=0, \bar{L}^{A_Y =  A_D = a1,  \bar{c}=0}_{s} ) \nonumber \\ 
       &  \prod_{j=0}^{s}  \big\{ \Pr(D^{a=0, \bar{c}=0}_{j+1}=0 \mid D^{a=0, \bar{c}=0}_{j}= Y^{a=0, \bar{c}=0}_{j}=0, \bar{L}^{a=0, \bar{c}=0}_{j} = \bar{l}_{j})  \nonumber \\
      &  \times \Pr(Y^{a=1,  \bar{c}=0}_{j}=0 \mid D^{a=1,  \bar{c}=0}_{j}= Y^{a=1,  \bar{c}=0}_{j-1}=0, \bar{L}^{a=1,  \bar{c}=0}_{j}= \bar{l}_{j}) \nonumber \\
    &\times \Pr(L^{a=1,  \bar{c}=0}_{A_Y,j} = l_{A_Y,j} \mid  Y^{a=1,  \bar{c}=0}_{j} = D^{a=1,  \bar{c}=0}_{j} = 0, \bar{L}^{a=1,  \bar{c}=0}_{j-1} = \bar{l}_{j-1}, L^{a=1,\bar{c}= 0}_{A_D,j} = l_{A_D,j}) \nonumber \\
       &\times \Pr(L^{a=0, \bar{c}=0}_{A_D,j} = l_{A_D,j} \mid  Y^{a=0, \bar{c}=0}_{j} = D^{a=0, \bar{c}=0}_{j} = 0, \bar{L}^{a=0, \bar{c}=0}_{j-1} = \bar{l}_{j-1}) \big\} \Big], \nonumber \\
\label{eq: counterfactual id formula}
\end{align}
where $Y^{a_Y,a_D,  \bar{c}=0}_{-1}$, and $L^{a_Y,a_D,  \bar{c}=0}_{-1}$ are empty sets.
Using Lemma \ref{lemma: conditional counter to obs}, we can substitute the terms in the last equality in \eqref{eq: counterfactual id formula}, 
\begin{align*}
       =& \sum_{\bar{l}_K}  \Big[ \sum_{s=0}^{K} \Pr(Y^{a=1,  \bar{c}=0}_{s+1}=1 \mid D^{a=1,  \bar{c}=0}_{s+1}= Y^{a=1,  \bar{c}=0}_{s}=0, \bar{L}^{A_Y =  A_D = a_Y,  \bar{c}=0}_{s} ) \nonumber \\ 
       &  \prod_{j=0}^{s}  \big\{ \Pr(D^{a=0, \bar{c}=0}_{j+1}=0 \mid D^{a=0, \bar{c}=0}_{j}= Y^{a=0, \bar{c}=0}_{j}=0, \bar{L}^{a=0, \bar{c}=0}_{j} = \bar{l}_{j})  \nonumber \\
      &  \times \Pr(Y^{a=1,  \bar{c}=0}_{j}=0 \mid D^{a=1,  \bar{c}=0}_{j}= Y^{a=1,  \bar{c}=0}_{j-1}=0, \bar{L}^{a=1,  \bar{c}=0}_{j}= \bar{l}_{j}) \nonumber \\
    &\times \Pr(L^{a=1,  \bar{c}=0}_{A_Y,j} = l_{A_Y,j} \mid  Y^{a=1,  \bar{c}=0}_{j} = D^{a=1,  \bar{c}=0}_{j} = 0, \bar{L}^{a=1,  \bar{c}=0}_{j-1} = \bar{l}_{j-1}, L^{a=1,\bar{c}= 0}_{A_D,j} = l_{A_D,j}) \nonumber \\
       &\times \Pr(L^{a=0, \bar{c}=0}_{A_D,j} = l_{A_D,j} \mid  Y^{a=0, \bar{c}=0}_{j} = D^{a=0, \bar{c}=0}_{j} = 0, \bar{L}^{a=0, \bar{c}=0}_{j-1} = \bar{l}_{j-1}) \big\} \Big] \nonumber \\
      = & \sum_{\bar{l}_K}  \Big[ \sum_{s=0}^{K} \Pr(Y_{s+1}=1 \mid C_{s+1}= D_{s+1}= Y_{s}=0, \bar{L}_{s} = \bar{l}_{s}, A = a_Y) \nonumber \\ 
       &  \prod_{j=0}^{s}  \big\{ \Pr(D_{j+1}=0 \mid C_{j+1}=D_{j}= Y_{j}=0, \bar{L}_{j} = \bar{l}_{j},  A = a_D)  \nonumber \\
      &  \times \Pr(Y_{j}=0 \mid C_{j}=D_{j}= Y_{j-1}=0, \bar{L}_{j-1} = \bar{l}_{j-1},  A = a_Y) \nonumber \\
    &\times \Pr(L_{A_Y,j} = l_{A_Y,j} \mid C_{j}= Y_{j} = D_{j} = 0, \bar{L}_{j-1} = \bar{l}_{j-1}, L_{A_D,j} = l_{A_D,j}, A = a_Y) \nonumber \\
    &\times \Pr(L_{A_D,j} = l_{A_D,j} \mid C_{j}= Y_{j} = D_{j} = 0, \bar{L}_{j-1} = \bar{l}_{j-1},  A = a_D) \big\} \Big].
\end{align*}
\end{proof}

\clearpage
\section{$Z_k$ partition and the dismissible component conditions}
\label{sec: appendix lemmas}

\begin{lemma}
\label{lemma: recanting witness}
$Z_k$ partition fails if and only if any of the following statements are true for some $k \in \{1,\dots, K\} $: (i) there is a direct arrow from $A_Y$ into $D_{k+1}$,  (ii) there is a direct arrow from $A_D$ into $Y_{k+1}$, or (iii) there exists a node $W \in \bar{Z}_k$ such that there are direct arrows from both $A_Y$ and $A_D$ into $W$. 
\end{lemma}
\begin{proof}

First we directly show that if (i), (ii) or (iii) holds then $Z_k$ partition fails. If (i) holds then \eqref{def: Ay partitioning Z} is violated by the presence of a causal path $A_y\rightarrow D_{k+1}$ and therefore $Z_k$ partition is violated.  If (ii) holds then \eqref{def: Ad partitioning Z} is  violated by the presence of a causal path $A_D\rightarrow Y_{k+1}$ and therefore $Z_k$ partition is violated.  Now suppose  (iii) is true and define a partition of $Z_k$ such that $W \in \bar{Z}_{A_D,k}$. Then a causal path $A_Y\rightarrow W $ will exist and \eqref{def: Ay partitioning Z}, and therefore $Z_k$ partition, is violated.  Alternatively suppose (iii) is true and define a partition of $Z_k$ such that, instead, $W \in \bar{Z}_{A_Y,k}$  Then a causal path $A_D\rightarrow W$ will exist and \eqref{def: Ad partitioning Z}, and therefore $Z_k$ partition, is violated.

Next we show by contradiction that if $Z_k$ partition fails then (i), (ii) or (iii) must hold. Suppose $Z_k$ partition holds and (i), (ii) or (iii) also holds. If (i), (ii) or (iii) holds then one of the following causal paths must be present: $A_Y\rightarrow D_{k+1}$, for some $k \in \{1,  \dots ,K\} $; $A_D\rightarrow Y_{k+1}$, for some $k \in \{1,  \dots ,K\}$; $A_Y\rightarrow W$ for some $W \in \bar{Z}_{A_D,K}$; or $A_D\rightarrow W$ for some $W \in \bar{Z}_{A_Y,K}$.  The presence of any of these paths violates either \eqref{def: Ay partitioning Z} or \eqref{def: Ad partitioning Z} such that $Z_k$ partition fails.  Thus, we have a contradiction and we are done.  




\end{proof}

\begin{lemma}
\label{lemma: dismissible and Zk partition}
If $Z_k$ partition fails, then at least one of the dismissible component conditions fail. 
\end{lemma}
\begin{proof}
Suppose $Z_k$ partition fails.  Then, by lemma \ref{lemma: recanting witness}, (i), (ii) or (iii) must hold such that at least one of the following paths must be present for $W \in \bar{Z}_{j}$, $j \leq k$, a cause of $Y_{k+1}$ and/or $D_{k+1}$, for some $k \in \{1,  \dots ,K\}$: $A_Y \rightarrow D_{k+1}$; $A_D\rightarrow Y_{k+1}$; $A_Y \rightarrow W\rightarrow \ldots \rightarrow D_{k+1}$ and $A_D\rightarrow W \rightarrow \ldots \rightarrow D_{k+1}$; or $A_D\rightarrow W\rightarrow \ldots \rightarrow Y_{k+1}$ and $A_Y\rightarrow W\rightarrow \ldots \rightarrow Y_{k+1}$.  

If the path $A_Y \rightarrow D_{k+1}$ is present then the dismissible component condition\eqref{ass: delta 2} fails for any choice of $\overline{L}_k$.   If the path $A_D \rightarrow Y_{k+1}$ is present then the dismissible component condition \eqref{ass: delta 1} fails for any choice of $\overline{L}_k$. 

Suppose $W \notin \bar{L}_K$, ($W$ is unmeasured).  If the path $A_Y\rightarrow W\rightarrow \ldots \rightarrow D_{k+1}$ is present then the dismissible component condition \eqref{ass: delta 2} fails for any choice of $\overline{L}_k$. If the path $A_D\rightarrow W\rightarrow Y_{k+1}$ is present then the dismissible component condition \eqref{ass: delta 1} fails for any choice of $\overline{L}_k$.  

Suppose $W \in \bar{L}_{j+1}$ for some $j \in \{0,  \dots ,K\}$ ($W$ is measured). If the paths $A_Y\rightarrow W\rightarrow \ldots \rightarrow D_{k+1}$ and $A_D\rightarrow W\rightarrow  \ldots \rightarrow D_{k+1}$ are present then, no matter our choice of $\overline{L}_j$, if we choose $W \in \bar{L}_{A_Y,j+1}$ the dismissible component condition \eqref{ass: delta 3a} fails and if we choose $W \in \bar{L}_{A_D,j+1}$ the dismissible component condition \eqref{ass: delta 3b} fails. Similarly, if the paths $A_D\rightarrow W\rightarrow \ldots \rightarrow Y_{k+1}$ and $A_Y\rightarrow W\rightarrow \ldots \rightarrow Y_{k+1}$ are present then, no matter our choice of $\overline{L}_j$, if we choose $W \in \bar{L}_{A_Y,j+1}$ the dismissible component condition \eqref{ass: delta 3a} fails and if we choose $W \in \bar{L}_{A_D,j+1}$ the dismissible component condition \eqref{ass: delta 3b} fails.  


\end{proof}

\begin{lemma}
\label{lemma: Ly empty and AY partial}
If the dismissible component conditions hold when we define $\overline{L}_{A_Y,K}=\emptyset$, then $A_Y$ partial isolation holds. 
\end{lemma}
\begin{proof}
We give a proof by contradiction. Suppose the dismissible component conditions hold under $L_{A_Y,k}=\emptyset$, but $A_Y$ partial isolation does not hold. Then, if there is a direct arrow $A_Y \rightarrow D_{k+1}$ for any $k \in \{0,  \dots ,K\}$, this arrow would violate \eqref{ass: delta 2}, which is a contradiction. Alternatively, $A_Y$ partial isolation can only be violated if there exists a $W$ such that $A_Y \rightarrow W \rightarrow ... \rightarrow D_{k+1}$ for any $k \in \{0,  \dots ,K\}$. However, if $W$ is measured then $W \in \overline{L}_{A_D,k}$ (because $\overline{L}_{A_Y,k}=\emptyset$), and then \eqref{ass: delta 3b} is violated, which is a contradiction. If $W$ is unmeasured, then either $\eqref{ass: delta 2}$ is violated or $\eqref{ass: delta 3b}$ is violated, which is a contradiction. 
\end{proof}

\begin{lemma}
\label{lemma: Ld empty and Ad partial}
If the dismissible component conditions hold when we define $L_{A_D,k}=\emptyset$, then $A_D$ partial isolation holds. 
\end{lemma}
\begin{proof}
The proof is analogous to the proof of lemma \ref{lemma: Ly empty and AY partial}.
\end{proof}

\begin{lemma}
\label{lemma: full iso}
If the dismissible component conditions hold for both the partition $L_{A_D,k}=\emptyset, L_{A_Y,k}=L_k$ and the partition $L_{A_Y,k}=L_k, L_{A_D,k}=\emptyset$, then full isolation holds. 
\end{lemma}
\begin{proof}
It follows immediately from Lemma \ref{lemma: Ly empty and AY partial} and Lemma \ref{lemma: Ld empty and Ad partial}, because full isolation holds by definition if both $A_Y$ partial isolation and $A_D$ partial isolation holds.
\end{proof}

\begin{lemma}
If the dismissible component conditions hold for both the partition $L_{A_D,k}=\emptyset, L_{A_Y,k}=L_k$ and the partition $L_{A_Y,k}=L_k, L_{A_D,k}=\emptyset$, then $$L_{k+1} \independent A \mid  D_{k+1}=Y_{k+1}=0, \bar{L}_{k}.$$
\end{lemma}
\begin{proof}

We give a proof by contradiction. Suppose that the dismissible component conditions hold for both the partition $L_{A_D,k}=\emptyset, L_{A_Y,k}=L_k$ and the partition $L_{A_Y,k}=L_k, L_{A_D,k}=\emptyset$, and there is a conditional dependence such that
$$L_{k} \not\!\perp\!\!\!\perp A \mid  D_{k}=Y_{k}=0, \bar{L}_{k-1},$$ for at least one $k=0, \ldots, K$. Using the rules of d-separation \cite{pearl2009causality}, we will consider the 4 possible ways in which $A$ and $L_k$ can be d-connected, conditional on $D_{k}=Y_{k}=0, \bar{L}_{k-1}$.

Suppose that the conditional dependence is due to a direct arrow from $A$ into $L_k$. Then, under the generalized decomposition assumption, either there is a direct arrow from $A_D$ into $L_k$ or a direct arrow from $A_Y$ into $L_k$, and these arrows would, repsectively, violate \eqref{ass: delta 3a} under the partition $L_{A_Y,k}=L_k, L_{A_D,k}=\emptyset$ and \eqref{ass: delta 3b} under the partition $L_{A_D,k}=L_k, L_{A_Y,k}=\emptyset$, which is a contradiction.

Suppose that the conditional dependence is due the unmeasured common cause $W$ of of $A$ and $L_k$. Then, under the generalized decomposition assumption, either (i) $W$ is a common cause of either $A_Y$ and $L_k$ or (ii) $W$ is a common cause of $A_D$ and $L_k$. However, (i) or (ii) would violate dismissible component condition \eqref{ass: delta 3a} or \eqref{ass: delta 3b}, which is a contradiction. 

Suppose that the conditional dependence is due to an unblocked path due to conditioning on $D_{k}=0, Y_{k}=0$ and $\bar{L}_{k-1}$, that is, by conditioning on a collider or a descendant of a collider. Then, under the generalized decomposition assumption, this path would lead to a conditional dependence between either $L_k$ and $A_Y$ or $L_k$ and $A_D$. Any such path would violate \eqref{ass: delta 3a} or \eqref{ass: delta 3b}, which is a contradiction of the result in Lemma \ref{lemma: full iso}. 

Finally, suppose that the conditional dependence is due to a direct arrow from $L_k$ where $k =1,\ldots, K$ into $A$. This would violate our assumption of a temporal order, that is, it would imply that $A$ occurs after $L_k$, which is a contradiction.
\end{proof}



\clearpage
  
\section{Proof of weighted representation of \eqref{eq: identifying formula}}   
\label{sec: proof of alternative id}
First, using laws of probability we can re-formulate the weights $W_{L_{A_Y},k} $ and $W_{L_{A_D},k}$,
\begin{align*}
W_{L_{A_D},k} (a_Y,a_D) & = \frac{\prod_{j=0}^{k} \Pr(A = a_D \mid C_{j}= Y_{j} = D_{j} = 0, L_{A_D,j}, \bar{L}_{j-1} = \bar{l}_{j-1})}{\prod_{j=0}^{k} \Pr(A = a_Y \mid C_{j}= Y_{j} = D_{j} = 0, L_{A_D,j}, \bar{L}_{j-1} = \bar{l}_{j-1})} \\
& \times \frac{\prod_{j=0}^{k} \Pr(A = a_Y \mid C_{j}= Y_{j} = D_{j} = 0, \bar{L}_{j-1} = \bar{l}_{j-1})}{\prod_{j=0}^{k} \Pr(A = a_D \mid C_{j}= Y_{j} = D_{j} = 0, \bar{L}_{j-1} = \bar{l}_{j-1})} \\
& = \frac{ \prod_{j=0}^{k} \frac{\Pr (A = a_D \mid C_{j}= Y_{j} = D_{j} = 0, L_{A_D,j}, \bar{L}_{j-1} = \bar{l}_{j-1}) \Pr( C_{j}= Y_{j} = D_{j} = 0, L_{A_D,j}, \bar{L}_{j-1} = \bar{l}_{j-1})   }{   \Pr(C_{j}= Y_{j} = D_{j} = 0,\bar{L}_{j-1} = \bar{l}_{j-1},  A = a_D) } }{ \prod_{j=0}^{k} \frac{\Pr (A = a_Y \mid C_{j}= Y_{j} = D_{j} = 0, L_{A_D,j}, \bar{L}_{j-1} = \bar{l}_{j-1}) f( C_{j}= Y_{j} = D_{j} = 0, L_{A_D,j}, \bar{L}_{j-1} = \bar{l}_{j-1})   }{   \Pr(C_{j}= Y_{j} = D_{j} = 0,\bar{L}_{j-1} = \bar{l}_{j-1},  A = a_Y) } } \\
& = \frac{\prod_{j=0}^{k}   \Pr(L_{A_D,j} = l_{A_D,j} \mid C_{j}= Y_{j} = D_{j} = 0, \bar{L}_{j-1} = \bar{l}_{j-1},  A = a_D) }{ \prod_{j=0}^{k}   \Pr(L_{A_D,j} = l_{A_D,j} \mid C_{j}= Y_{j} = D_{j} = 0, \bar{L}_{j-1} = \bar{l}_{j-1},  A = a_Y) }, \\
\end{align*}
and 
\begin{align*}
W_{L_{A_Y},k} (a_Y,a_D) &  = \frac{\prod_{j=0}^{k} \Pr(A = a_Y \mid C_{j}= Y_{j} = D_{j} = 0, \bar{L}_{j} = \bar{l}_{j})}{\prod_{j=0}^{k}\Pr(A = a_D \mid C_{j}= Y_{j} = D_{j} = 0, \bar{L}_{j} = \bar{l}_{j})} \\
& \times \frac{\prod_{j=0}^{k} \Pr(A = a_D \mid C_{j}= Y_{j} = D_{j} = 0,L_{A_D,j}, \bar{L}_{j-1} = \bar{l}_{j-1})}{\prod_{j=0}^{k} \Pr(A = a_Y \mid C_{j}= Y_{j} = D_{j} = 0,L_{A_D,j}, \bar{L}_{j-1} = \bar{l}_{j-1})} \\
& = \frac{ \prod_{j=0}^{k} \frac{\Pr (A = a_Y \mid C_{j}= Y_{j} = D_{j} = 0, \bar{L}_{j} = \bar{l}_{j}) \Pr( C_{j}= Y_{j} = D_{j} = 0, \bar{L}_{j} = \bar{l}_{j})   }{   f(C_{j}= Y_{j} = D_{j} = 0,L_{A_D,j}, \bar{L}_{j-1} = \bar{l}_{j-1},  A = a_Y) } }{ \prod_{j=0}^{k} \frac{\Pr (A = a_D \mid C_{j}= Y_{j} = D_{j} = 0, \bar{L}_{j} = \bar{l}_{j}) f( C_{j}= Y_{j} = D_{j} = 0, \bar{L}_{j} = \bar{l}_{j})   }{   f(C_{j}= Y_{j} = D_{j} = 0,L_{A_D,j}, \bar{L}_{j-1} = \bar{l}_{j-1},  A = a_D) } } \\
& = \frac{\prod_{j=0}^{k}   f(L_{A_Y,j} = l_{A_Y,j} \mid C_{j}= Y_{j} = D_{j} = 0,L_{A_D,j}, \bar{L}_{j-1} = \bar{l}_{j-1},  A = a_Y) }{ \prod_{j=0}^{k}   f(L_{A_Y,j} = l_{A_Y,j} \mid C_{j}= Y_{j} = D_{j} = 0,L_{A_D,j}, \bar{L}_{j-1} = \bar{l}_{j-1},  A = a_D) }. \\
\end{align*}

Define
\begin{align*}
    W'_{C,k} (a_Y) = \frac{1 }{ \prod_{j=0}^{k}  \Pr(C_{j+1}=0 \mid C_{j}=D_{j}= Y_{j}=0, \bar{L}_{j} = \bar{l}_{j},  A = a_D) }.
\end{align*}

Consider the expression
\begin{align*}
    E & [ W_{C,k}(a_Y) W_{D,k}(a_Y,a_D) W_{L_{A_D},k}(a_Y,a_D) Y_{k+1} (1-Y_{k}) (1-D_{k+1}) \mid A=a_Y] \\ 
    =  & E [ W'_{C,k} (a_Y) W_{D,k}(a_Y,a_D) W_{L_{A_D},k}(a_Y,a_D) Y_{k+1} (1-Y_{k}) (1-D_{k+1}) (1-C_{k+1}) \mid A=a_Y] \\ 
    =& \sum_{\bar{l}_k}  \sum_{\bar{y}_{k+1}} \sum_{\bar{d}_{k+1}} [ f(\bar{y}_{k+1}, d_{k+1},c_{k+1},\bar{l}_k \mid A = a_Y) W'_{C,k} (a) W_{D,k}(a_Y,a_D) W_{L_{A_D},k}(a_Y,a_D) \\
    & \times y_{k+1} (1-y_{k}) (1-d_{k+1}) (1-c_{k+1}) ]  \\ 
    =& \sum_{\bar{l}_k}  [\Pr(Y_{k+1}=1,Y_k=D_{k+1}=C_{k+1}=0,\bar{l}_k \mid A = a_Y) W'_{C,k} (a_Y) W_{D,k}(a_Y,a_D) W_{L_{A_D},k}(a_Y,a_D)]  \\
    =& \sum_{\bar{l}_k}  [ \Pr(Y_{k+1}=1 \mid Y_k= D_{k+1}=C_{k+1}=0,\bar{l}_k, A=a_Y) \Pr( D_{k+1}=0 \mid \bar{C}_{k+1}= \bar{D}_k= \bar{Y}_k =0,\bar{l}_k,A=a_Y )    \\
    & \times \Pr(C_{k+1}= 0 \mid \bar{D}_k= \bar{Y}_{k}= \bar{C}_{k} =0,\bar{l}_{k}, A=a_Y) f( \bar{l}_{k} \mid \bar{C}_{k}= \bar{D}_k= \bar{Y}_{k} =0,  A=a_Y)    \\
    & \times \Pr(\bar{Y}_{k}=\bar{D}_{k}=\bar{C}_{k}=0 \mid  A= a_Y  ) \\
    & \times W'_{C,k} (a_Y) W_{D,k}(a_Y,a_D) W_{L_{A_D},k}(a_Y,a_D)]  \\ 
   =& \sum_{\bar{l}_k}  [ \Pr(Y_{k+1}=1 \mid Y_k= D_{k+1}=C_{k+1}=0,\bar{l}_k,A=a_Y ) \Pr( D_{k+1}=0 \mid \bar{C}_{k+1}= \bar{D}_k= \bar{Y}_k =0,\bar{l}_k,A=a_Y )    \\
    & \times \Pr(C_{k+1}= 0 \mid \bar{D}_k= \bar{Y}_{k}= \bar{C}_{k} =0,\bar{l}_{k}, A= a_Y)  f(l_{k} \mid \bar{Y}_{k}=\bar{D}_{k}=\bar{C}_{k}= 0,\bar{l}_{k-1}, A= a_Y) \\
      & \times \Pr(\bar{Y}_{k}=\bar{D}_{k}=\bar{C}_{k}=0,\bar{L}_{k-1}={l}_{k-1} \mid A= a_Y) \\
     & \times W'_{C,k} (a_Y) W_{D,k}(a_Y,a_D) W_{L_{A_D},k}(a_Y,a_D)],  \\ 
\end{align*}
where we use the definition of expected value in the second equation, the fact that $Y_k$ and $D_k$ are binary in the third equation, laws of probability in the fourth and fifth equation. 

We use laws of probability to express $ f(\bar{Y}_{k}=\bar{D}_{k}=\bar{C}_{k}=0, \bar{l}_{k-1}  \mid A= a_Y )$ as
\begin{align*}
 &   \Pr(Y_{k}=0 \mid C_{k}=D_{k}= Y_{k-1}=0, \bar{l}_{k-1},  A = a_Y) \nonumber \\
   & \times \Pr(D_{k}=0 \mid C_{k}=D_{k-1}= Y_{k-1}=0, \bar{l}_{k-1},  A = a_Y)  \nonumber \\
    & \times \Pr(C_{k}= 0 \mid D_{k-1} = Y_{k-1}= C_{k-1} =0,\bar{l}_{k-1}, A = a_Y)  \\
   & \times f(l_{k-1} \mid C_{k-1}=D_{k-1}= Y_{k-1}=0,\bar{l}_{k-2}, A = a_Y)  \\
    & \times f(\bar{Y}_{k-1}=\bar{D}_{k-1}=0, \bar{l}_{k-2}, \bar{C}_{k-1}=0 \mid  A= a_Y ), \\
\end{align*}
where any variable indexed with a number $m < 0$ is defined to be the empty set.

Arguing iteratively for $k-1,k-2,...,0$ we find that
\begin{align*}
    E  [   W'_{C,k}  & (a_Y)  W_{D,k}(a_Y,a_D) W_{L_{A_D},k}(a_Y,a_D) Y_{k+1} (1-Y_{k}) (1-D_{k+1}) (1- C_{k+1}) \mid A=a_Y ] \\ 
        = \sum_{\bar{l_k}} & \Big[ \Pr(Y_{k+1}=1 \mid Y_k= D_{k+1}=C_{k+1}=0,\bar{l}_k, A=a_Y) \\
         \prod_{j=0}^{k} &  \big\{ \Pr(D_{j+1}=0 \mid C_{j+1}=D_{j}= Y_{j}=0, \bar{L}_{j} = \bar{l}_{j},  A = a_Y)  \nonumber \\
      &  \times \Pr(Y_{j}=0 \mid C_{j}=D_{j}= Y_{j-1}=0, \bar{L}_{j-1} = \bar{l}_{j-1},  A = a_Y) \nonumber \\
     & \times \Pr(C_{j+1}= 0 \mid \bar{D}_j= \bar{Y}_{j}= \bar{C}_{j} =0,\bar{L}_{j} = \bar{l}_{j},  a_Y)  \\
   & \times \Pr(L_{j}=l_{j} \mid C_{j}=D_{j}= Y_{j}=0,\bar{L}_{j-1} = \bar{l}_{j-1}, A = a_Y) \big\}  \\
     \times  & W'_{C,k} (a_Y) W_{D,k}(a_Y,a_D) W_{L_{A_D},k}(a_Y,a_D) \Big]  \\
        = \sum_{\bar{l_k}} & \Big[ \Pr(Y_{k+1}=1 \mid Y_k= D_{k+1}=C_{k+1}=0,\bar{l}_k, A=a_Y) \\
        \prod_{j=0}^{k} &  \big\{ \Pr(D_{j+1}=0 \mid C_{j+1}=D_{j}= Y_{j}=0, \bar{L}_{j} = \bar{l}_{j},  A = a_Y)  \nonumber \\
      &  \times \Pr(Y_{j}=0 \mid C_{j}=D_{j}= Y_{j-1}=0, \bar{L}_{j-1} = \bar{l}_{j-1},  A = a_Y) \nonumber \\
     &  \times \Pr(C_{j+1}= 0 \mid \bar{D}_j= \bar{Y}_{j}= \bar{C}_{j} =0,\bar{L}_{j} = \bar{l}_{j},  A = a_Y)  \\
    &\times \Pr(L_{A_Y,j} = l_{A_Y,j} \mid C_{j}= Y_{j} = D_{j} = 0, \bar{L}_{j-1} = \bar{l}_{j-1}, L_{A_D,j} = l_{A_D,j}, A = a_Y) \nonumber \\
   & \Pr(L_{A_D,j} = l_{A_D,j} \mid C_{j}= Y_{j} = D_{j} = 0, \bar{L}_{j-1} = \bar{l}_{j-1},  A = a_Y) \big\} \\
     \times & W'_{C,k} (a_Y) W_{D,k}(a_Y,a_D) W_{L_{A_D},k}(a_Y,a_D) \Big], \\
\end{align*}
where we use that $L_k = (L_{A_Y,k},L_{A_D,k})$ in the second equality.

By plugging in the expression for $ W'_{C,k} (a_Y) $, we get 
\begin{align*}
        =& \sum_{\bar{l_k}}  [ \Pr(Y_{k+1}=1 \mid Y_k= D_{k+1}=C_{k+1}=0,\bar{L}_k=\bar{l}_k, A=a_Y) \\
        \times  & \prod_{j=0}^{k}  \big\{ \Pr(D_{j+1}=0 \mid C_{j+1}=D_{j}= Y_{j}=0, \bar{L}_{j} = \bar{l}_{j},  A = a_Y)  \nonumber \\
      &  \times \Pr(Y_{j}=0 \mid C_{j}=D_{j}= Y_{j-1}=0, \bar{L}_{j} = \bar{l}_{j},  A = a_Y) \nonumber \\
    &\times \Pr(L_{A_Y,j} = l_{A_Y,j} \mid C_{j}= Y_{j} = D_{j} = 0, \bar{L}_{j-1} = \bar{l}_{j-1}, L_{A_D,j} = l_{A_D,j}, A = a_Y) \nonumber \\
    &\times \Pr(L_{A_D,j} = l_{A_D,j} \mid C_{j}= Y_{j} = D_{j} = 0, \bar{L}_{j-1} = \bar{l}_{j-1},  A = a_Y) \big\}  \\
  & \times W_{D,k}(a_Y,a_D) W_{L_{A_D},k}(a_Y,a_D)],  \\
\end{align*}
By plugging in the expression for the weights $W_{L_{A_D},k}(a_Y,a_D)$ and $W_{D,k}(a_Y,a_D)$ we obtain 
\begin{align*}
        =& \sum_{\bar{l_k}}  [ \Pr(Y_{k+1}=1 \mid Y_k= D_{k+1}=C_{k+1}=0,\bar{L}_k=\bar{l}_k, A=a_Y) \\
        \times  & \prod_{j=0}^{k}  \big\{ \Pr(D_{j+1}=0 \mid C_{j+1}=D_{j}= Y_{j}=0, \bar{L}_{j} = \bar{l}_{j},  A = a_D)  \nonumber \\
      &  \times \Pr(Y_{j}=0 \mid C_{j}=D_{j}= Y_{j-1}=0, \bar{L}_{j} = \bar{l}_{j},  A = a_Y) \nonumber \\
    &\times \Pr(L_{A_Y,j} = l_{A_Y,j} \mid C_{j}= Y_{j} = D_{j} = 0, \bar{L}_{j-1} = \bar{l}_{j-1}, L_{A_D,j} = l_{A_D,j}, A = a_Y) \nonumber \\
    &\times \Pr(L_{A_D,j} = l_{A_D,j} \mid C_{j}= Y_{j} = D_{j} = 0, \bar{L}_{j-1} = \bar{l}_{j-1},  A = a_D) \big\},  \\
\end{align*}
and the final expression is equal to \eqref{eq: identifying formula}.

\clearpage

\section{Treatment decomposition of $A$ into $A_Y$, $A_D$ and $A_Z$}
\label{sec: A_Z decomposition}
Hitherto we have described settings in which the treatment is decomposed into 2 components, $A_D$ and $A_Y$. Consider now a hypothetical treatment decomposition into 3 components $A_D$, $A_Y$ and $A_Z$, as illustrated in Figure \ref{fig: 3 comp decomp}, which is similar to Robins and Richardson's decomposition in a mediation setting \cite[Figure 6(d)]{robins2010alternative}. Analogous to the 2 way decomposition, we define a generalized decomposition assumption:
\begin{enumerate}
    \item[] \underline{3 way generalized decomposition assumption}: The treatment $A$ can be decomposed into three binary components $A_Y\in \{0,1\}$, $A_D\in \{0,1\}$ and $A_Z\in \{0,1\}$ such that, in the observed data, the following determinism holds 
\begin{align}
    A\equiv A_D\equiv A_Y \equiv A_Z,
    \label{assumption: Determinsm 3 way}
\end{align}
but, in a future study, $A_Y$,$A_D$ and $A_Z$ could be assigned different values under a hypothetical intervention. For any individual in the study population and for $k \in \{0,\ldots,K\}$, let $Y_{k+1}^{a_Y,a_D,a_Z}$ be the indicator of the event of interest by interval $k+1$ had, possibly contrary to fact, he/she been assigned to $A_Y=a_Y$, $A_D=a_D$ and  $A_Z=a_Z$, where $a_Y,a_D,a_Z\in \{0,1\}$. We assume that an  intervention that assigns $A=a$ results in the same outcome as an intervention that assigns $A_Y=A_D=A_Z=a$, that is, 
\begin{align}
Q_{k+1}^{a_Y=a_D=a_Z=a}=Q_{k+1}^{a},
\label{eq: definition A=Ay=Ad=Az}
\end{align}
\end{enumerate}
for $Q_{k+1} \in \{Y_{k+1},D_{k+1},Z_{k+1}\}$. Analogous to the 2 way decomposition, the 3 way decomposition may be practically interesting in settings where we can conceive interventions on all 3 components of $A$. Furthermore, in settings where $Z_k$ partition fails, it may be possible to define a 3 way decomposition that allows identifiability of separable effects. For example, Figure \ref{fig: 3 comp decomp} can represent an alternative decomposition of the setting described in Figure \ref{fig:no isolation}a, where $Z_k$ partition fails.

To define identifiability conditions that apply to settings with 3 way decompositions, we continue to use superscripts to denote counterfactuals and for notational simplicity we consider settings without censoring, such that e.g.\ $Y_{k+1}^{a_Y,a_D,a_Z}$ is the counterfactual value of $Y_{k+1}$ if, possibly contrary to fact, $A_Y= a_Y, A_D= a_D,A_Z = a_Z \in \{0,1\}$. 


Here we will only consider settings that satisfy the following assumptions: 
\begin{align}
& \text{the only causal paths from } A_Y \text{ to } D_{k+1} \text{ and } Z_{k+1}, k \in \{0,\ldots,K\} \text{ are through } Y_{j}, \nonumber  \\
&  j=0,...,k,  \label{def: 3 part Ay iso}  \\
& \text{the only causal paths from } A_D \text{ to } Y_{k+1} \text{ and } Z_{k+1},  k=0,\dots,K \text{ are through } D_{j+1},  \nonumber \\ 
& j=0,\dots,k. \label{def: 3 part Ad iso}  \\
& \text{the only causal paths from } A_Z \text{ to } Y_{k+1} \text{ and } D_{k+1},  k=0,\dots,K \text{ are through } Z_{j+1}, \nonumber \\ 
& j=0,\dots,k.  \label{def: 3 part Al iso}
\end{align}

For $k = 0, \ldots, K $, consider the separable effects
\begin{equation}
\Pr (Y_{k+1}^{a_Y=1,a_D,a_Z}=1)\text{ vs. }\Pr
(Y_{k+1}^{a_Y=0,a_D,a_Z }=1),
\label{eq: dir effect A}
\end{equation}
$\text{ for } a_D,a_Z \in \{0,1\}$,


\begin{equation}
\Pr (Y_{k+1}^{a_Y,a_D=1,a_Z}=1)\text{ vs. }\Pr
(Y_{k+1}^{a_Y,a_D=0,a_Z}=1),
\label{eq: indir effect D}
\end{equation}%
for $a_Y,a_Z \in \{0,1\}$, and
\begin{equation}
\Pr (Y_{k+1}^{a_Y,a_D,a_Z=1}=1)\text{ vs. }\Pr
(Y_{k+1}^{a_Y,a_D,a_Z=0}=1),
\label{eq: indir effect L}
\end{equation}%
for $a_Y,a_D \in \{0,1\}$. 

Similar to the two component decomposition, the total effect can be expressed as a sum of the separable direct and indirect effects, in particular, 
\begin{align*}
& \Pr (Y_{k+1}^{a_Y=1,a_D=1,a_Z=1}=1)-\Pr (Y_{k+1}^{a_Y=0,a_D=1,a_Z=1}=1) \\
& +\Pr (Y_{k+1}^{a_Y=0,a_D=1,a_Z=1}=1)-\Pr (Y_{k+1}^{a_Y=1,a_D=0,a_Z=1}=1) \\
& +\Pr (Y_{k+1}^{a_Y=0,a_D=0,a_Z=1}=1)-\Pr (Y_{k+1}^{a_Y=0,a_D=0,a_Z=0}=1) \\
& =\Pr (Y_{k+1}^{a=1}=1)-\Pr (Y_{k+1}^{a=0}=1).
\end{align*}

\subsection{Interpretation of the 3 component decomposition}
Under \eqref{def: 3 part Ay iso}-\eqref{def: 3 part Al iso}, the 3 way decomposition of $A$ into $A_D$, $A_Y$ and $A_Z$ allows us to interpret the separable effects as direct and indirect effects; \eqref{eq: dir effect A} is the effect not emanating from $A_D$ or $A_Z$, i.e.\ a separable direct effect, \eqref{eq: indir effect D} is the separable indirect effect on the event of interest only emanating from $A_D$, and \eqref{eq: indir effect L} is the separable indirect effect on the event of interest only emanating from $A_Z$. 

In our running example, where $Z_k=L_k$ encodes the (systolic and diastolic) blood pressure, it is not obvious that the 3 part decomposition is of interest; to interpret effects defined by the 3 part decomposition, we would need to conceptualize a treatment decomposition of blood pressure therapy into 3 components: the $A_D$ component could now be defined as the component that exerts effects on mortality not through blood pressure reduction or kidney injury; that is, the substantive meaning of an intervention on $A_D$ fundamentally changes. The $A_Z$ component would affect the outcome of interest only through blood pressure reduction; the effect exerted by $A_Z$ is analogous to an indirect mediation effect described by Didelez \cite{didelez2018defining} under an agnostic causal model, but in our setting we also allow for competing risks. We note that under this 3 way decomposition, the $A_Y$ component is identical to the $A_Y$ component in the 2 way decomposition, that is, the component of blood pressure therapy only exerting direct effects on kidney injury not through blood pressure reduction. 

In other settings, however, the 3 part decomposition may be feasible. For example, Robins and Richardson \cite[Figure 6(d)]{robins2010alternative} consider a similar decomposition in a conceptual example on the effect of cigarettes on lung cancer; they consider the effect of cigarettes smoking through nicotine, tar and other pathways. 

\subsection{Identification of the 3 component decomposition}
The identifiability conditions are straightforward extensions of the conditions in Section \ref{sec: identifiability conditions}. Now we must identify 
\begin{align*}
\Pr (Y_{k+1}^{a_Y, a_D, a_Z}=1) \text{ for } a_Y,a_D,a_Z \in \{0,1\}.
\end{align*}%

First the exchangeability, consistency and positivity conditions are identical to the condition in Section \ref{sec: identifiability conditions}.
The dismissible component conditions read
\begin{align*} 
 Y_{k+1}(G) \independent (A_D(G), A_Z(G)) \mid A_Y(G), Y_{k}(G)=D_{k+1}(G)=0, \bar{L}_k(G), \\
 D_{k+1}(G) \independent (A_Y(G), A_Z(G)) \mid A_D(G), D_{k}(G)=Y_{k}(G)=0, \bar{L}_k(G), \\ 
 L_{k+1}(G) \independent (A_Y(G), A_D(G)) \mid A_Z(G), D_{k+1}(G)=Y_{k+1}(G)=0, \bar{L}_k(G). \\
\end{align*}

Under these assumptions we can identify $\Pr (Y_{k+1}^{a_Y, a_D, a_Z}=1)$ for $k = 0, \ldots, K $ from 
\begin{align}
       & \sum_{\bar{l}_k}  \Big[ \sum_{s=0}^{k} \Pr(Y_{s+1}=1 \mid  D_{s+1}= Y_{s}=0, \bar{L}_{s} = \bar{l}_{s}, A = a_Y) \nonumber \\ 
       &  \prod_{j=0}^{s}  \big\{ \Pr(D_{j+1}=0 \mid D_{j}= Y_{j}=0, \bar{L}_{j} = \bar{l}_{j},  A = a_D)  \nonumber \\
      &  \times \Pr(Y_{j}=0 \mid D_{j}= Y_{j-1}=0, \bar{L}_{j-1} = \bar{l}_{j-1},  A = a_Y) \nonumber \\
    &\times \Pr(L_{j} = l_{A_Y,j} \mid Y_{j} = D_{j} = 0, \bar{L}_{j-1} = \bar{l}_{j-1}, A = a_Z) \big\} \Big],
\label{eq: identifying formula L}
\end{align}
which follows from a similar derivation from that in Appendix \ref{sec: proof of idenditifiability}. The identifiability conditions under the 3 component decomposition require stronger restrictions on the unmeasured variables, compared to the settings in Section \ref{sec: identifiability conditions}; unmeasured common causes of any pair in $(Y_{k+1}, D_{j+1}, L_{m+1}), k,j,m \in \{0,  \dots ,K\}$ can violate the dismissible component conditions. In particular, an unmeasured common cause $U_{L,Y}$ of $L_k$ and $Y_k$ will violate the dismissible component condition, as shown in grey in Figure \ref{fig: 3 comp with L}.

\clearpage
\begin{figure}
\centering
\begin{tikzpicture}
\begin{scope}[every node/.style={thick,draw=none}]
    \node (A) at (-1,1.5) {$A$};
    \node (Ay) at (1,1.5) {$A_Y$};
	\node (Ad) at (1,-1) {$A_D$};
	\node (Al) at (1,4) {$A_Z$};
	\node (Y1) at (3,1.5) {$Y_1$};
    \node (D1) at (3,-1) {$D_1$};
    \node (Y2) at (6,1.5) {$Y_2$};
    \node (D2) at (6,-1) {$D_2$};
    \node (Z1) at (3,4) {$Z_1$};
\end{scope}

\begin{scope}[>={Stealth[black]},
              every node/.style={fill=white,circle},
              every edge/.style={draw=black,very thick}]
    \path [->] (A) edge[line width=0.85mm] (Ad);
    \path [->] (A) edge[line width=0.85mm] (Ay);
    \path [->] (A) edge[line width=0.85mm] (Al);
	\path [->] (Ad) edge (D1);
	\path [->] (Al) edge (Z1);
    \path [->] (Ad) edge[bend right] (D2);
	\path [->] (Ay) edge[bend left] (Y2);
    \path [->] (Ay) edge (Y1);	
    \path [->] (Z1) edge (D2);
    \path [->] (Z1) edge (Y2);
    \path [->] (Y1) edge (D2);
    \path [->] (Y1) edge (Y2);
    \path [->] (D1) edge (D2);
    \path [->] (D1) edge (Y1);
    \path [->] (D2) edge (Y2);
    \path [->] (Y1) edge (Z1);
    \path [->] (D1) edge[bend left] (Z1);
\end{scope}
\end{tikzpicture}
\caption{Treatment $A$ is decomposed into 3 components.}
\label{fig: 3 comp decomp}
\end{figure}
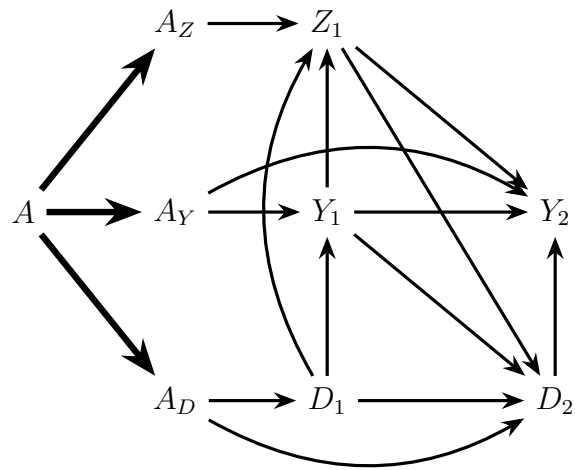

\clearpage
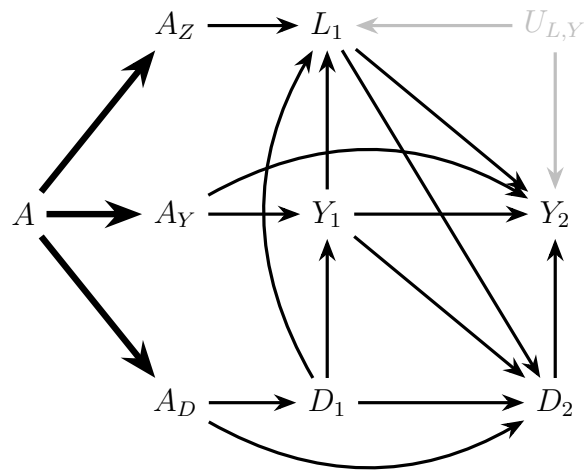
\begin{figure}
\centering
\begin{tikzpicture}
\begin{scope}[every node/.style={thick,draw=none}]
    \node (A) at (-1,1.5) {$A$};
    \node (Ay) at (1,1.5) {$A_Y$};
	\node (Ad) at (1,-1) {$A_D$};
	\node (Al) at (1,4) {$A_Z$};
	\node (Y1) at (3,1.5) {$Y_1$};
    \node (D1) at (3,-1) {$D_1$};
    \node (Y2) at (6,1.5) {$Y_2$};
    \node (D2) at (6,-1) {$D_2$};
    \node (L2) at (3,4) {$L_1$};
    \node[lightgray] (ULY) at (6,4) {$U_{L,Y}$};
\end{scope}

\begin{scope}[>={Stealth[black]},
              every node/.style={fill=white,circle},
              every edge/.style={draw=black,very thick}]
    \path [->] (A) edge[line width=0.85mm] (Ad);
    \path [->] (A) edge[line width=0.85mm] (Ay);
    \path [->] (A) edge[line width=0.85mm] (Al);
	\path [->] (Ad) edge (D1);
	\path [->] (Al) edge (L2);
    \path [->] (Ad) edge[bend right] (D2);
	\path [->] (Ay) edge[bend left] (Y2);
    \path [->] (Ay) edge (Y1);	
    \path [->] (L2) edge (D2);
    \path [->] (L2) edge (Y2);
    \path [->] (Y1) edge (D2);
    \path [->] (Y1) edge (Y2);
    \path [->] (D1) edge (D2);
    \path [->] (D1) edge (Y1);
    \path [->] (D2) edge (Y2);
    \path [->] (Y1) edge (L2);
    \path [->] (D1) edge[bend left] (L2);
    \path [->,>={Stealth[lightgray]}] (ULY) edge[lightgray] (L2);
    \path [->,>={Stealth[lightgray]}] (ULY) edge[lightgray] (Y2);
\end{scope}
\end{tikzpicture}
\caption{Treatment is decomposed into 3 components, such that $L_1 = Z_1$. The variable $U_{L,Y}$ would violate the dismissible component conditions here.}
\label{fig: 3 comp with L}
\end{figure}

\clearpage

\section{Estimation algorithms}
\label{sec: estimation algorithm}
Here we describe an algorithm to estimate the separable effects using estimators based on \eqref{eq: alternative id formula 1}; i.e. the estimator $\hat{\nu}_{1,a_Y,a_D,k} $ described in Section \ref{sec: estimation}. We initially construct our input data set such that each subject has $K^*+1$ lines, indexed by $k = 0,\dots,K^*$, and there are measurements of $(A, C_{k+1},D_{k+1}, Y_{k+1},\bar{L}_{k+1})$ on each line $k$. For each subject, $K^*=K$ if $C_{K+1}=D_{K+1}=Y_{K+1}=0$, otherwise $K^*=m$, where $C_{m}=D_{m}=Y_{m}=0$ and either $C_{m+1}=1$, $D_{m+1}=1$ or $Y_{m+1}=1$. Due to the temporal ordering, we do the following: if $C_{k+1}=1$, then $D_{k+1}$ and $Y_{k+1}$ are set missing. Similarly, if $C_{k+1}=0$ and  $D_{k+1}=1$, then $Y_{k+1}=1$ is set missing. Then we do the following to estimate \eqref{eq: alternative id formula 1} at $K$:
\begin{enumerate}
    \item Using all subject-intervals records, i.e.\ all lines in the data set, obtain $\hat{\alpha}_D$ by fitting a parametric model (e.g.\ pooled logistic regression model) with dependent variable $D_{k+1}$ and independent variables a specified function of $k = 0, \dots K$, $\bar{L}_k$ and $A$. 
    \item Using all subject-intervals records, obtain $\hat{\alpha}_C$ by fitting a parametric model (e.g.\ pooled logistic regression model) with dependent variable $C_{k+1}$ and independent variables a specified function of $k = 0, \dots K$, $\bar{L}_k$ and $A$.
    \item Using all subject-intervals records, estimate $\hat{\alpha}_{L_D,1}$ by fitting a parametric model with dependent variable $A$ and independent variables a specified function of $k = 0, \dots K$, $\bar{L}_k$ and $A$.
    \item Using all subject-intervals records, estimate $\hat{\alpha}_{L_D,2}$ by fitting a parametric model with dependent variable $A$ and independent variables a specified function of $k = 0, \dots K$, $\bar{L}_{k-1}$,$L_{A_D,k}$ and $A$, ensuring that the models used to fit $\hat{\alpha}_{L_D,1}$ and $\hat{\alpha}_{L_D,2}$ are compatible. Notice that this step is redundant if we can define a $L_k$ partition such that $L_{A_D,k} = \emptyset, k = 0, \dots K$, which implies that $A_D$ partial isolation holds. 
    \item For subject $i$, attach a weight to line $k$ with predicted outcome probabilities derived from the parametric models indexed by parameters  $\hat{\alpha}_{D}, \hat{\alpha}_{L_D 1},\hat{\alpha}_{L_D 2}$ and $\hat{\alpha}_{C}$ to estimate $\hat{W}_{1,i,k}(a_Y,a_D; \hat{\alpha}_1) = \hat{W}_{i,D,k} (a_Y,a_D;\hat{\alpha}_{D}) \hat{W}_{i,L_{A_D},k} (a_Y,a_D;\hat{\alpha}_{L_D 1},\hat{\alpha}_{L_D 2}) \hat{W}_{i,C,k}  (a_Y;\hat{\alpha}_{C})$.
    \item Compute an estimate of $\Pr(Y^{a_Y,a_D,\bar{c}=0}_{K+1})$ from 
    \begin{align*}
  \frac{1}{\sum^{n}_{j=1}I(A_j=a_Y)}    \sum^{n}_{i=1}  \sum_{k=0}^{K}  \hat{W}_{1,i,k}(a_Y,a_D;\hat{\alpha}_1) Y_{i,k+1} (1-Y_{i,k}) (1-D_{i,k+1})I(A_i=a_Y).
    \end{align*}
\end{enumerate}
An estimator based on \eqref{eq: alternative id formula 2} could be derived analogously, where step $(1)$ we would fit a model with $Y_{k+1}$ as dependent variable, in step $(4)$ we would fit a model where we replace $L_{A_D,k}$ with $L_{A_Y,k}$, and we finally compute an estimate of  $\Pr(Y^{a_Y,a_D,\bar{c}=0}_{K+1})$ from 
\begin{align*}
     \frac{1}{\sum^{n}_{j=1}I(A_j=a_D)} \sum^{n}_{i=1}  \sum_{k=0}^{K} \hat{W}_{2,i,k}(a_Y,a_D;\hat{\alpha}_1) Y_{i,k+1} (1-Y_{i,k}) (1-D_{i,k+1})I(A_i=a_D).
\end{align*}


\clearpage

\section{Sensitivity analysis}
\label{sec: sensitivity analysis}
To illustrate a sensitivity analysis technique for violations of the dismissible component conditions, consider a selection bias function for dismissible component condition \eqref{ass: delta 1},
\begin{align*}
    t_k(\bar{l}_k, a_Y) = &   \Pr(Y^{a_Y,a_D=0,\bar{c}= 0}_{k+1}=1 \mid  D^{a_Y,a_D=0,\bar{c}= 0}_{k+1}= Y^{a_Y,a_D=0,\bar{c}= 0}_{k}=0, \bar{L}^{a_Y,a_D=0,\bar{c}= 0}_{k} = \bar{l}_{k}) \\
    & - \Pr(Y^{a_Y,a_D=1,\bar{c}= 0}_{k+1}=1 \mid  D^{a_Y,a_D=1,\bar{c}= 0}_{k+1}= Y^{a_Y,a_D=1,\bar{c}= 0}_{k}=0, \bar{L}^{a_Y,a_D=1,\bar{c}= 0}_{k} = \bar{l}_{k}),
\end{align*}
which is identified in a setting in which $A_Y$ and $A_D$ are randomly assigned. Analogous sensitivity functions could be defined for dismissible component conditions \eqref{ass: delta 2}-\eqref{ass: delta 3b}. 
If dismissible component condition \eqref{ass: delta 1} holds for $\bar{L}_k$, we know that $t(\bar{L}_k, a_Y) = 0$. However, if \eqref{ass: delta 1} was violated, we would expect that $t_k(\bar{l}_k, a_Y) \neq 0$ for some values of $\bar{l}_k$ and $a_Y$. In particular, we would expect \eqref{ass: delta 1} to be violated in the presence of any unmeasured cause of $Y_k$ and $D_j$, where $0 < j \leq k$.

While the following strategy for sensitivity analysis is applicable to any setting in which $Z_k$ partition holds, we consider a simpler setting in which (i) $A_Y$ partial isolation holds, (ii) dismissible component condition \eqref{ass: delta 1} is satisfied for some $L' \equiv L'_D$ which contains the measured variable $L$ as a subset, $L \subset L'$, and (iii) dismissible component condition \eqref{ass: delta 2}-\eqref{ass: delta 3b} are satisfied. This is coherent with our blood pressure example in Section \ref{sec: sprint implementation}, and one such setting is described in Figure \ref{fig:partial isolationExpanded}f where \eqref{ass: delta 1} is violated due to failure of measuring $U_{L,Y}$.  Now, suppose that $t_k(\bar{l}_k, a_Y)$ is known. Then the separable effects can be identified through the modified version of identification formula \eqref{eq: alternative id formula 2}, 
\begin{align}
       \sum_{s=0}^{k} E & \{ W_{C,s}(a_D) W^{\dagger}_{Y,s}(a_D,a_Y)  (1-Y_{s}) (1-D_{s+1}) Y_{s+1} \mid A=a_D \},
       \label{eq: alternative id formula 2 sensitivity}
\end{align}
where 
\begin{align*}
 W^{\dagger}_{Y,s}(a_D,a_Y) & = \frac{  (-1)^{a_Y} t_{s+1}(\bar{l}_{s+1}, a_Y) + \Pr(Y_{s+1}=1 \mid C_{s+1}=D_{s+1}= Y_{s}=0, \bar{L}_{s},  A = a_Y) }{\Pr(Y_{s+1}=1 \mid C_{s+1}=D_{s+1}= Y_{s}=0, \bar{L}_{s},  A = a_D) } \\
& \times \frac{\prod_{j=0}^{s-1} (-1)^{a_Y} t_j(\bar{l}_j, a_Y) + \Pr(Y_{j+1}=0 \mid C_{j+1}=D_{j+1}= Y_{j}=0, \bar{L}_{j},  A = a_Y) }{ \prod_{j=0}^{s-1} \Pr(Y_{j+1}=0 \mid C_{j+1}=D_{j+1}= Y_{j}=0, \bar{L}_{j},  A = a_D) }, \\
\end{align*}
which is equal to \eqref{eq: alternative id formula 2} under $A_Y$ partial isolation when $t_k(\bar{l}_k, a_Y)=0$ for all $k$, $\bar{l}_k$ and $ a_Y$. Formula \eqref{eq: alternative id formula 2 sensitivity} motivates the estimator $\hat{\nu}^{\dagger}_{2,a_Y,a_D,k}$, a modified version of $\hat{\nu}_{2,a_Y,a_D,k}$ from Section \ref{sec: estimation}, such that
$\hat{\nu}^{\dagger}_{2,a_Y,a_D,k}$ is the solution to the estimating equation $\sum_{i=1}^{n}U^{\dagger}_{2,k,i}(\nu_{a_Y,a_D,k},\hat{\alpha}_2)=0$ with respect to $\nu_{a_Y,a_D,k}$, where 
\begin{align*}
& U^{\dagger}_{2,k,i}(\nu_{a_Y,a_D,k},\hat{\alpha}_2) \nonumber \\
= & I(A_i=a_D) \Big[ \sum_{s=0}^{k} \{\hat{W}_{C,s,i} (a_D;\hat{\alpha}_{C}) \hat{W}^{\dagger}_{2,s,i}(a_Y,a_D;\hat{\alpha}_{2}) Y_{s+1,i} (1-Y_{s,i}) (1-D_{s+1,i}) \} -  \nu_{a_Y,a_D,k} \Big],  \nonumber \\
\end{align*}
and $\hat{W}^{\dagger}_{2,s,i}(a_Y,a_D;\hat{\alpha}_{2} ) = \hat{W}_{C,s,i} (a_D;\hat{\alpha}_{C}) \hat{W}^{\dagger}_{Y,s,i} (a_D,a_Y;\hat{\alpha}_{Y})$, where 
\begin{align*}
    & \hat{W}^{\dagger}_{Y,k,i}  (a_D,a_Y;\hat{\alpha}_{Y}) \\
    & = \frac{ (-1)^{a_Y} t_{k+1}(\bar{l}_{k+1}, a_Y) + \Pr(Y_{k+1}=1 \mid C_{k+1}=D_{k+1}= Y_{k}=0, \bar{L}_{k,i},  A = a_Y; \hat{\alpha}_{Y}) }{\Pr(Y_{k+1}=1 \mid C_{j+1}=D_{k+1}= Y_{k}=0, \bar{L}_{k,i},  A = a_D; \hat{\alpha}_{Y}) } \\
& \times \frac{\prod_{j=0}^{k-1}  (-1)^{a_D} t_{j+1}(\bar{l}_{j+1}, a_Y) + \Pr(Y_{j+1}=0 \mid C_{j+1}=D_{j+1}= Y_{j}=0, \bar{L}_{j,i},  A = a_Y; \hat{\alpha}_{L_Y1}) }{ \prod_{j=0}^{k-1} \Pr(Y_{j+1}=0 \mid C_{j+1}=D_{j+1}= Y_{j}=0, \bar{L}_{j,i},  A = a_D; \hat{\alpha}_{L_Y1}) }. \\
\end{align*}

\begin{proof}
The following equality holds by definition of $t_{k}(\bar{l}_k, a_Y)$,
\begin{align}
   & \Pr(Y^{a_Y,a_D=0,\bar{c}= 0}_{k+1} = 1 \mid  Y^{a_Y,a_D=0,\bar{c}= 0}_k = D^{a_Y,a_D=0,\bar{c}= 0}_{k+1} = 0,\bar{L}^{a_Y,a_D=0,\bar{c}= 0}_{k} = \bar{l}_{k}) \nonumber \\
   = & t_{k}(\bar{l}_k, a_Y)+\Pr(Y^{a_Y,a_D=1,\bar{c}= 0}_{k+1} = 1 \mid  Y^{a_Y,a_D=1,\bar{c}= 0}_k = D^{a_Y,a_D=1,\bar{c}= 0}_{k+1} = 0,\bar{L}^{a_Y,a_D=1,\bar{c}= 0}_{k} = \bar{l}_{k}).  \nonumber \\ \label{eq: ineq y sensitivity}
\end{align}
While \eqref{eq: ineq y} of Lemma \ref{lemma: delta conditions} is violated in our setting where dismissible component condition \eqref{ass: delta 1} is violated, note that \eqref{eq: ineq d}-\eqref{eq: ineq ld} of Lemma \ref{lemma: delta conditions} holds and that Lemma \ref{lemma: conditional counter to obs} holds regardless of violations of the dismissible component conditions. Thus, following analogous steps as in the proof of Theorem \ref{theorem: identification formula}, we use \eqref{ass: E1 app}-\eqref{ass: consistency cens}, \eqref{ass: delta 2 app}-\eqref{ass: delta 3b app}, as well as \eqref{eq: ineq y sensitivity} instead of \eqref{ass: delta 1 app},  to obtain the following identification formula for settings where $a_Y \neq a_D$,
\begin{align*}
 \Pr ( & Y_{k+1}^{a_Y,a_D,\bar{c}=0}=1) \\
       = & \sum_{\bar{l}_K}  \Big[ \sum_{s=0}^{K}  (-1)^{a_Y} t_{k}(\bar{l}_k, a_Y) +  \Pr(Y_{s+1}=1 \mid C_{s+1}= D_{s+1}= Y_{s}=0, \bar{L}_{s} = \bar{l}_{s}, A = a_Y) \nonumber \\ 
       &  \prod_{j=0}^{s}  \big\{ \Pr(D_{j+1}=0 \mid C_{j+1}=D_{j}= Y_{j}=0, \bar{L}_{j} = \bar{l}_{j},  A = a_D)  \nonumber \\
      &  \times [ (-1)^{a_D} t_{k}(\bar{l}_k, a_Y) + \Pr(Y_{j}=0 \mid C_{j}=D_{j}= Y_{j-1}=0, \bar{L}_{j-1} = \bar{l}_{j-1},  A = a_Y)] \nonumber \\
    &\times \Pr(L_{A_Y,j} = l_{A_Y,j} \mid C_{j}= Y_{j} = D_{j} = 0, \bar{L}_{j-1} = \bar{l}_{j-1}, L_{A_D,j} = l_{A_D,j}, A = a_Y) \nonumber \\
    &\times \Pr(L_{A_D,j} = l_{A_D,j} \mid C_{j}= Y_{j} = D_{j} = 0, \bar{L}_{j-1} = \bar{l}_{j-1},  A = a_D) \big\} \Big],
\end{align*}
and a weighted representation of this identification formula is analogous to identification formula \eqref{eq: alternative id formula 2}, where  $\hat{W}_{2,s}(a_Y,a_D)$ is replaced by $W^{\dagger}_{Y,s}(a_D,a_Y)$, which can be shown by an argument that is analogous to the proof in Appendix \ref{sec: proof of alternative id}.

\end{proof}

Now a formal sensitivity analysis can be conducted by repeatedly estimating $\hat{\nu}^{\dagger}_{2,a_Y,a_D,k}$ for each choice of $t_k(\bar{l}_k, a_Y)$ for a set of functions $\mathcal{T} = \{ t_{k,\lambda}(\bar{l}_k, a_Y) : \lambda \}$, where $\lambda$ is a finite dimensional parameter and $t_{k,0}(\bar{l}_k, a_Y) \equiv 0$ describes the setting with no bias, that is, no unmeasured common causes of  $Y_k$ and $D_j$ or of $Y_k$ and $L_j$, for any $j,k$ such that $0 < j \leq k$.

Subject matter knowledge may help us to reason about the sensitivity function $t_k(\bar{l}_k, a_Y)$. To fix ideas, suppose that the graph in Figure \ref{fig:partial isolationExpanded}f represents the blood pressure example, where $U_{L,Y}$ is an unmeasured common cause that increases the blood pressure ($L_k$) and the risk of kidney failure ($Y_k$). Then we would expect $t_k(\bar{l}_k, a_Y)$ to be negative due to selection over time: subjects who do not receive the treatment component that intensively reduces blood pressure ($a_D = 0$) are less likely to be alive with larger values of  $U_{L,Y}$ compared to those who received the component that intensively reduces blood pressure ($a_D = 1$).

Our sensitivity analysis technique is inspired by Tchetgen Tchetgen \cite{tchetgen2014identification}. However, unlike Tchetgen Tchetgen \cite{tchetgen2014identification}, the terms in our sensitivity function are not cross-world quantities that are unobservable in principle, but conditional expectations that can be identified in an experiment in which $A_Y$ and $A_D$ are randomly assigned.

Furthermore, note that our identification results from Section \ref{sec: identifiability conditions} also motivate sensitivity analyses of violations of the isolation conditions from Section \ref{sec: isolation and interpret}. In particular, suppose that an investigator assumed that full isolation was satisfied and, thus, used the simplified identification formula that was introduced in Stensrud et al \cite{stensrud2019separable}. Then, the assumption of full isolation could be falsified by comparing these estimates to estimates derived from the estimators in Section \ref{sec: estimation}, only assuming $Z_k$ partition. To do this sensitivity analysis, the investigator need to measure a set of time-varying covariates $L_k, k \in \{0,\dots, K\}$.




\end{document}